\documentclass[a4paper,10pt]{article}

\usepackage{ucs}
\usepackage[utf8x]{inputenc}
\usepackage{graphicx}
\usepackage{amsfonts}
\usepackage{dsfont}
\usepackage{amssymb}
\usepackage{amsmath}
\usepackage{amsthm}
\usepackage{enumerate}
\usepackage{stmaryrd}
\usepackage{fullpage}
\usepackage{ifthen}
\usepackage{subfigure}
\usepackage{epic}
\usepackage{authblk}
\usepackage{textcomp}

\usepackage[hypertexnames=false,colorlinks=true,linkcolor=blue,citecolor=blue]{hyperref}
\usepackage[numbers,comma,square,sort&compress]{natbib}
\usepackage[a4paper,text={6.5in,10in},centering]{geometry}

\usepackage{color}
\usepackage{titlesec}

\graphicspath{{eps/}{pdf/}}

\newcommand{\bqq}{\begin{equation}}
\newcommand{\eqq}{\end{equation}}
\newcommand{\bqs}{\begin{equation*}}
\newcommand{\eqs}{\end{equation*}}

\newcommand{\D}{\mathbb{D}}
\newcommand{\Z}{\mathbb{Z}}
\newcommand{\R}{\mathbb{R}} 
\newcommand{\C}{\mathbb{C}}
\newcommand{\N}{\mathbb{N}}

\newcommand{\A}{\mathcal{A}}
\newcommand{\g}{\mathcal{G}}
\newcommand{\cH}{\mathcal{H}}
\newcommand{\Q}{\mathcal{Q}}
\newcommand{\Oct}{\mathcal{O}}
\newcommand{\mP}{\mathcal{P}}

\newcommand{\mL}{\mathbf{L}}

\newcommand{\mT}{\mathbf{T}}
\newcommand{\mC}{\mathbf{C}}
\newcommand{\mD}{\mathbf{D}}
\newcommand{\cR}{\mathcal{R}}

\newcommand{\SU}{\textbf{SU}(1,1)}
\newcommand{\U}{\textbf{U}(1,1)}

\newcommand{\atanh}{\mathrm{atanh} \,}

\numberwithin{equation}{section}

\newtheorem{rmk}{Remark}[section]

\newtheorem{thm}{Theorem}[section]
\newtheorem{defi}{Definition}[section]
\newtheorem{lem}{Lemma}[section]
\newtheorem{prop}{Proposition}[section]

\newcommand{\etal}{\textit{et al.}\ }

\graphicspath{{Figures_review/}}

\title{Pattern formation for the Swift-Hohenberg equation on the hyperbolic plane}
\author[1]{Pascal Chossat}
\author[2]{Gr\'egory Faye}
\affil[1]{\small Department of Mathematics, University of Nice Sophia-Antipolis, JAD Laboratory and CNRS, Parc Valrose, 06108 Nice Cedex 02, France}
\affil[2]{\small University of Minnesota, School of Mathematics,  206 Church Street S.E.,  Minneapolis, MN 55455}

\begin{document}
\maketitle
\centerline{\em Dedicated to Klaus Kirchg\"assner}
\bigskip

\begin{abstract}
In this paper we present an overview of pattern formation analysis for an analogue of the Swift-Hohenberg equation posed on the real hyperbolic space of dimension two, which we identify with the Poincar\'e disc $\D$. Different types of patterns are considered: spatially periodic stationary solutions, radial solutions and traveling waves, however there are significant differences in the results with the Euclidean case. We apply equivariant bifurcation theory to the study of spatially periodic solutions on a given lattice of $\D$ also called H-planforms in reference with the "planforms" introduced for pattern formation in Euclidean space. We consider in details the case of the regular octagonal lattice and give a complete descriptions of all H-planforms bifurcating in this case. For radial solutions (in geodesic polar coordinates), we present a result of existence for stationary localized radial solutions, which we have adapted from techniques on the Euclidean plane. Finally, we show that unlike the Euclidean case, the Swift-Hohenberg equation in the hyperbolic plane undergoes a Hopf bifurcation to traveling waves which are invariant along horocycles of $\D$ and periodic in the "transverse" direction. We highlight our theoretical results with a selection of numerical simulations.
\end{abstract}


\section{Introduction}\label{sec1}

The goal of this paper is to lay the foundations of pattern formation theory for PDE's which are defined in an hyperbolic space and more specifically in the hyperbolic plane, which we identify from now on with the Poincar\'e disc $\D$. The origin of our interest in this question comes from a model for the visual perception of textures by the cortex of mammals, in which it is assumed that textures are represented by the "structure tensor", a $2\times 2$ symmetric, positive definite matrix, and the neurons in the visual cortex region named V1 are sensitive to the values of the structure tensor. It can be shown that there is no loss of generality to assume that the structure tensor has determinant 1, and it turns out that the space of such structure tensors is isomorphic to the hyperbolic plane. Spontaneous activity in the cortex then leads to a bifurcation problem with pattern formation. Several papers have been published on this topic \cite{chossat-faugeras:09,faye-chossat-etal:11,chossat-faye-etal:11,faye-chossat:11,faye-chossat:12}. The model equations are integro-differential. However under certain conditions they can be converted to PDE's which look like the Swift-Hohenberg equation in $\R^2$, except that the Laplacian is replaced by the Laplace-Beltrami operator $\Delta_\D$ in $\D$ \cite{faye:12}. We shall therefore concentrate on this equation now defined in $\D$, which we take of the general form
\bqq
\label{eq:sh}
u_t=-(\alpha^2+\Delta_\D)^2u+\lambda u+\mathcal{N}(u), \quad z\in\D,
\eqq  
where 
\bqs
\mathcal{N}(u)=\nu u^2- \eta u^3
\eqs
and $\lambda$, $\alpha$, $\nu$ and $\eta$ are real coefficients. There is another reason to consider (\ref{eq:sh}). The Swift-Hohenberg equation was initially derived to model pattern formation in the B\'enard problem, that is the onset of convection for a fluid flow set between two horizontal and extended plates with a positive temperature gradient across the shell \cite{swift-hohenberg:77}. This equation turns out to be the "simplest" PDE that contains the ingredients for steady state pattern formation from an homogeneous state: symmetries (Euclidean invariance), existence of a critical non zero wave number, nonlinearity. Most basic phenomena that are associated with pattern formation in physical and chemical/biological systems are captured by this equation. 

We may similarly consider the Swift-Hohenberg equation in hyperbolic space (or plane) as a model equation for pattern and wave formation in problems related to quantum chaos  \cite{balazs-voros:86, aurich-steiner:89,schmit:91,cornish-turok:98} and cosmological theories \cite{inoue:99,cornish-spergel:99,lehoucq-weeks-etal:02}. Note that more recently, some aspects of dispersive and concentration phenomena have been studied for evolution equations such as the Nonlinear Schr\"odinger equation and the wave equation posed on hyperbolic space \cite{banica:07,anker-pierfelice:09,anker-pierfelice-etal:12}. The methods presented here should be transposable to these types of equations.

Our aim is to "translate" to the hyperbolic case the methods which have been successful in the Euclidean case. As we said above, these methods are not restricted to the Swift-Hohenberg equation and can be applied to many other systems. However this equation is sufficiently simple and in a sense, "universal", to serve as a guideline for our study.

In order to facilitate the reading of the paper we first recall below the main ideas, which lie the analysis of pattern formation in Euclidean space.

The mathematical modeling of pattern formation in morphogenesis originates in the seminal work of Alan Turing on "The Chemical Basis of Morphogenesis" \cite{turing:52} (1952). The model is based on systems of reaction-diffusion equations. A theory of pattern formation in physical systems was developed at about the same time by Lev Landau who derived and analyzed with V. L. Ginzburg an equation for phase transition in supraconductivity, which turned out to be adequate in many other problems involving phase transitions \cite{landau-ginzburg:50}. Another root of pattern formation theory can be found in the work of G.I. Taylor on stability of Couette flow (primary flow of a fluid filling a cylindrical shell with rotating inner cylinder) \cite{taylor:23} (1923). There the flow is governed by Navier-Stokes equations, as is the B\'enard problem of onset of convection. In the framework of pattern formation all these problems share similar properties with the Swift -Hohenberg equation, which we may formally write
\bqq \label{eq:basique}
\frac{du}{dt} = F(u,\lambda)
\eqq
where $u=u(x,t)\in\R$ is a scalar field, $x\in\R^n$ (in general $n=1$, $2$ or $3$), and $F$ is a smooth operator defined in some Banach space. The coefficients $\alpha$, $\nu$ and $\eta$ are fixed. This equation is {\em invariant} under isometric transformations of $U$: let us define $T_gu(x,t)=u(g^{-1}\cdot x,t)$ where $g$ is any isometry in $\R^n$. Then $F(T_g u,\lambda)=T_gF(u,\lambda)$ for all $g\in \mathbf{E}(n,\R)$ (Euclidean group), $x\in\R^n$ and $t\in\R$. 
The linear stability analysis of the trivial, homogeneous, state $u=0$ is done by solving the linearized equation
\bqs
u_t=-(\alpha^2+\Delta_{\R^n})^2u+\lambda u
\eqs
for disturbances of the form $u(x,t)=e^{\sigma t}(A\cos{({\bf k}\cdot x)}+B\sin{({\bf k}\cdot x)})$ as $\lambda$ is varied. The vectors ${\bf k}\in\R^n$ are called {\em wave vectors}. This leads to a {\em dispersion relation} 
\bqs
\sigma = -(-\|{\bf k}\|^2+\alpha^2)^2+\lambda .
\eqs
When $\lambda<0$ all such modes are exponentially damped to 0 as $t\rightarrow+\infty$\footnote{It can be shown that 0 is globally attracting against periodic perturbations for the Swift-Hohenberg equation thanks to its variational structure.}. At the critical value $\lambda=0$ the eigenvalues $\sigma$ fill the half-line $(-\infty,~0]$. For the critical eigenvalue $0$ the "neutral" modes are $A\cos{({\bf k}\cdot x)}+B\sin{({\bf k}\cdot x)}$ where ${\bf k}$ is any vector with $\|{\bf k}\|=\alpha$. This value $\alpha$ is the {\em critical wave number}. When $\lambda >0$ the critical modes become unstable and we may expect the bifurcation of solutions which are not homogeneous but instead have a wavy structure in space. However a {\em continuum} of Fourier modes become simultaneously unstable. Moreover, to the critical wave number there corresponds a full circle of critical wave vectors with that length. These facts forbid us from directly applying standard reduction techniques such as Lyapunov-Schmidt decomposition \cite{chossat-lauterbach:00} and center manifold theorem \cite{haragus-iooss:10}, to the bifurcation analysis of this instability. Indeed, a crucial hypothesis in these methods is that the neutral (or center) part of the spectrum of the critical linearized operator can be separated from the rest of the spectrum and moreover consists of a finite number of eigenvalues with finite multiplicity. We therefore need to be more specific in the type of patterns we are looking for.     

There are basically two main cases which have been thoroughly investigated and which lead to elegant as well as relevant solutions from the observational point of view: (i) {\em periodic patterns} and (ii)  {\em radially symmetric (localized) patterns}. Other types of patterns have also been studied by various authors, such as quasi periodic patterns (see \cite{iooss-rucklidge:10} for recent advances in this case), spirals, defects in periodic patterns and others. All of these are important from physical point of view. However we concentrate in this paper on the fundamental patterns (i) and (ii) .

Let us briefly recall the basic ideas which lie behind the two main approaches.
\begin{itemize}
\item[(i)] {\bf Periodic patterns}. In the 60's and early 70's several authors, among whom Klaus Kirchg\"assner played a prominent role, established a rigorous nonlinear theory for the bifurcation of cellular solutions in the Couette-Taylor and B\'enard problem, see \cite{kirchgassner-kielhofer:73} for a general exposition. The idea was to restrict the problem to spatially periodic solutions in order to obtain a finite dimensional bifurcation equation by applying Lyapunov-Schmidt decomposition for Navier-Stokes equations. Later these ideas have been generalized and geometrically formalized (see \cite{chossat-lauterbach:00} for a bibliography).  On the circle of critical wave numbers we can select any two non colinear vectors ${\bf k}_1$ and ${\bf k}_2$, and span a {\em periodic lattice} ${\cal L}^*=\{n_1{\bf k}_1+n_2{\bf k}_2~,n_1,n_2\in\Z\}$. Any harmonic function $\exp{(i{\bf k}\cdot x)}$ with ${\bf k}\in {\cal L}^*$, is biperiodic in the plane with respect to translations $p_1{\bf e}_1+p_2{\bf e}_2$, where ${\bf e}_j\in\R^2$ are such that ${\bf e}_i\cdot {\bf k}_j=\delta_{ij}$ and $p_1$, $p_2$ are integers. The set of these translations forms a {\em lattice group} ${\cal L}$ (${\cal L}^*$ is the dual lattice of ${\cal L}$). Suppose now we look for the bifurcation of solutions which are invariant under the action of ${\cal L}$. These solutions are therefore spatially periodic. This implies that the equation can be restricted to the class of ${\cal L}$ periodic functions and the symmetry group of translations $\R^2$ has now to be taken "modulo translations" in ${\cal L}$, in other words the translation group of symmetries of the equation becomes $\R^2/\Z^2\simeq \mathbb{T}^2$ (the 2-torus). Moreover within this class, the spectrum of the linearized operator is discrete with finite multiplicity eigenvalues. In particular the critical eigenmodes, which corresponds to those harmonic functions in ${\cal L}$ with $\|{\bf k}\|=\alpha$, are in finite number. This allows to apply the center manifold reduction theorem and get a bifurcation equation in $\R^d$ where $d$ is the number of wave vectors in ${\cal L}^*$ with critical length. There are three possible cases:
\begin{enumerate}
\item If ${\bf k}_1\bot{\bf k}_2$, ${\cal L}$ is a square lattice. Then $d=4$ (critical wave numbers $\pm{\bf k}_1$, $\pm{\bf k}_2$). ${\cal L}$ is invariant under rotation of angle $\pi/2$ as well as under reflections through the axis along ${\bf e}_1$ for example. This spans the symmetry group of the square ${\bf D}_4$. The full symmetry group of the equation in this class of functions is the semi-direct product ${\bf D}_4\ltimes \mathbb{T}^2$. 
\item If ${\bf k}_1$ and ${\bf k}_2$ belong to the vertices of a regular hexagon, ${\cal L}$ is a hexagonal lattice. Then $d=6$. ${\cal L}$ is invariant under the symmetries of the hexagon, which form the 12 element group ${\bf D}_6$. The full symmetry group of the equation in this class of functions is the semi-direct product ${\bf D}_6\ltimes \mathbb{T}^2$. 
\item In all other cases the lattice is rhombic (rectangular). This however implies that there exist in ${\cal L}^*$ vectors with length smaller than $k_c$. Such a vector can be chosen as a basis element for the lattice. However perturbations with the corresponding wave function are damped and should not be observable in general (certain conditions may allow for these modes to coexist as critical modes, we do not enter in these cases here).
\end{enumerate}
Finally, the symmetries of the resulting, finite dimensional system can be exploited to classify the bifurcated branches by their isotropy (residual symmetry) and to determine their stability. Details can be found in a number of references, see \cite{golubitsky-stewart-etal:88}, \cite{chossat-lauterbach:00} for the general context of equivariant bifurcation theory, and \cite{hoyle:06} for a thorough exposition of the mathematical theory of pattern formation. 

It turns out that in hyperbolic geometry the same procedure applies but leads to very different results. The reason is that there are infinitely many types of lattices in $\D$. Given an integer $n\geq 3$ one can always build an $n$-gon in $\D$ which tiles the entire hyperbolic plane, while in contrast there are only 4 types of fundamental tiles in $\R^2$ (oblique, rhombic, square and hexagonal). Moreover if $\Gamma$ is a lattice group in $\D$ (we give a precise definition in Section\ref{sec3}), then the compact surface $\D/\Gamma$ is a torus, the genus of which can take any value $p\geq 2$ (depending on $\Gamma$). This torus admits a finite group of automorphisms (again in contrast to the Euclidean case where the symmetry group of the torus has the form ${\bf D}_k\rtimes \mathbb{T}^2$ with $k=2$ (rhombic), $4$ (square) or $6$ (hexagonal). This will be exposed in Section~\ref{sec3}. We have applied the method to the simplest case of a regular octagonal lattice. All possible bifurcations are described and it is shown that non trivial dynamics can also bifurcate generically in one case. 

\item[(ii)] {\bf Localized states and radially symmetric patterns}. Experiments in pattern formation can also produce localized solutions: fronts, bumps, spirals, to name a few. The method exposed above for analyzing the bifurcation of periodic patterns does not work and another approach must be undertaken. The fundamental idea was introduced in the early 80's by Klaus Kirchg\"assner \cite{kirchgassner:88} in the context of solitary water waves (Korteweg-de Vries equation) and further developed by G. Iooss and others (see \cite{Dias-Iooss}). Suppose for a moment that the domain is one dimensional: $x\in\R$ and one is looking for steady states: $F(u,\lambda)=0$.
This equation is an ODE in the space variable $x$: $0 = -(\alpha^2+\partial ^2_x)^2u + \lambda u + \nu u^2 -\eta u^3$. Bounded solutions of this ODE will produce admissible patterns. This equation can be written as a 4 dimensional evolution equation
\bqs
\frac{d}{dx}\left(\begin{array}{c} U \\ V \\ W \\ Z \end{array}\right) = 
\left(
\begin{array}{cccc}
0  & 1  & 0 & 0 \\
 0 & 0  & 1 & 0 \\
 0 & 0  & 0 & 1 \\
 \lambda-\alpha^2 & 0 &-2\alpha^2 & 0   
\end{array}
\right)
 \left(\begin{array}{c} U \\ V \\ W \\ Z \end{array} \right) + \left(\begin{array}{c} 0 \\ 0 \\ 0 \\ \nu U^2-\eta U^3 \end{array} \right)
\eqs
where $U=u$, $V = \partial _x u$, $W = \partial ^2_x u$ and $Z = \partial ^3_x u$.
The linear part has a pair of purely imaginary, double eigenvalues $\pm i\alpha$ at $\lambda=0$. Another essential property of this system is "time" reversibility: changing $x$ to $-x$, $V$ to $-V$ and $Z$ to $-Z$ does not change the equations. This bifurcation problem corresponds to a reversible Hopf bifurcation with $1:1$ resonance, which has been studied in details for example in \cite{iooss-peroueme:93,iooss-adelmeyer:98}. Applying reversible normal form theory (see \cite{iooss-peroueme:93,iooss-adelmeyer:98,haragus-iooss:10}) one recovers the bifurcation of $x$ periodic patterns, but also one can show depending on the values of coefficients $\nu$ and $\eta$ the bifurcation of homoclinic orbits to the trivial state or to periodic orbits, which correspond to "bumps" decreasing to 0 or approaching asymptotically periodic patterns at $x\rightarrow\pm\infty$. 

Extending this approach to multidimensional problems is technically involved. A rigorous theory was handled in the 90's by Scheel \cite{scheel:98,scheel:03} for the bifurcation of radially symmetric solutions of reaction-diffusion equations. His idea was again to consider the system as an evolution problem, but now with respect to the radial variable $r$. This leads to a singular problem which can be tackled by applying successive changes of variables and a suitable center manifold theory. A detailed analytical and numerical investigation has also been presented in \cite{burke-knobloch:06,burke-knobloch:07,burke-knobloch:07c,lloyd-sandstede-etal:08,lloyd-sandstede:09,avitabile-lloyd-etal:10,mccalla-sandstede:10,mccalla:11} for the Swift-Hohenberg equation. We shall detail the procedure in the context of hyperbolic geometry in Section 4 and show that in some respect, the situation is simpler in hyperbolic than in Euclidean plane.
\end{itemize}

A final remark in this introduction is that pattern formation has been intensively studied in Euclidean and spherical geometry, see \cite{dionne-golubitsky:92,dionne-silber-etal:97,chossat-lauterbach-etal:90}, and the present study completes the landscape by adding hyperbolic geometry to the classification. We believe this gives an intrinsic interest to this study. 

In the next section we expose basic notions of the geometry and harmonic analysis in the Poincar\'e disc, in particular the Fourier Helgason transform which is the equivalent to Fourier transform in $\R^2$.

\section{Geometry and harmonic analysis in the hyperbolic plane}\label{sec2}

\subsection{Isometries of the Poincar\'e disk}\label{subsec21}
The hyperbolic plane is a Riemannian space of dimension 2 with constant curvature $-1$. It admits several representations among which the most suitable for our purpose is the "Poincar\'e disc" $\D=\{z\in\C~|~|z|<1\}$ equipped with the distance
\bqq 
\label{eq:distance}
d_\D(z,z')=2\atanh \left( \frac{|z-z'|}{|1-z\bar z'|}\right).
\eqq
Geodesics in $\D$ are carried by circles which intersect the boundary $\partial\D$ orthogonally. The corresponding measure element is given by
\bqq
\label{eq:measure_element}
\text{dm}(z)=\frac{4dzd\bar z}{(1-|z|^2)^2}.
\eqq

We now describe the isometries of $\D$, i.e the transformations that preserve the distance $d_\D$. We refer to the classical textbook in hyperbolic geometry for details, e.g, \cite{katok:92}. The direct isometries (preserving the orientation) in $\D$ are the elements of the special unitary group, noted $\SU$, of $2\times 2$ Hermitian matrices with determinant equal to $1$. This is a 3-parameter real simple Lie group. Given:
\bqq
\label{eq:element_su}
\gamma =\left(\begin{array}{ll}
 \alpha & \beta \\
 \bar\beta & \bar \alpha
\end{array}
 \right) \text{ such that } |\alpha|^2-|\beta|^2=1, 
\eqq
an element of $\SU$, the corresponding isometry $\gamma$ in $\D$ is defined by:
\bqq
\gamma \cdot z=\frac{\alpha z+\beta}{\bar\beta z+\bar \alpha},\quad z\in\D.
\label{eq:corresping_isometry}
\eqq
Orientation reversing isometries of $\D$ are obtained by composing any transformation \eqref{eq:corresping_isometry} with the reflexion $\kappa:z\rightarrow \bar z$. The full symmetry group of the Poincar\'e disc is therefore:
\bqs
\U=\SU\cup\kappa\cdot \SU.
\eqs

Let us now describe the different kinds of direct isometries acting in $\D$. We first define the following one parameter subgroups of $\SU$:
\bqs
\left\{ \begin{array}{lll}
 K\overset{def}{=}\{\text{r}_{\phi}=\left( \begin{array}{ll}
 e^{i\frac{\phi}{2}} & 0\\
0 & e^{-i\frac{\phi}{2}}
\end{array}
\right),\phi\in \mathbb{S}^1\},\\
 A\overset{def}{=}\{a_\tau= \left( \begin{array}{ll}
 \cosh( \tau/2) & \sinh (\tau/2)\\
 \sinh (\tau/2) & \cosh (\tau/2)
\end{array}
\right),\tau \in\R\},\\
 N\overset{def}{=}\{n_s= \left( \begin{array}{ll}
 1+is & -is\\
 is & 1-is
\end{array}
\right),s\in\R\}.
\end{array}
\right.
\eqs

Note that $\text{r}_{\phi}\cdot z=e^{i\phi}z$ for $z\in\D$ and also $a_\tau\cdot 0=\tanh(\tau/2)$. The following decomposition holds (see \cite{iwaniec:02}).
\begin{thm}[Iwasawa decomposition]
\bqs
\SU=KAN.
\eqs
\end{thm}
The group $K$ is the orthogonal group $\mathbf{SO}(2)$ which fixes the center $O$ of $\D$. Its orbits are concentric circles. The orbits of $A$ converge to the same limit points of the unit circle $\partial\D$, $b_{\pm 1}=\pm 1$ when $\tau \rightarrow\pm \infty$. The elements of $A$ are sometimes called \textit{boosts} in the theoretical Physics literature \cite{balazs-voros:86}. They are circular arcs in $\D$ going through the points $b_1$ and $b_{-1}$. The orbits of $N$ are the circles inside $\D$ and tangent to the unit circle at $b_1$. These circles are called \textit{horocycles} with base point $b_1$. $N$ is called the horocyclic group. These orbits are shown in Figure~\ref{fig:orbits}.
\begin{figure}[htp]
 \centering
 \includegraphics[width=0.8\textwidth]{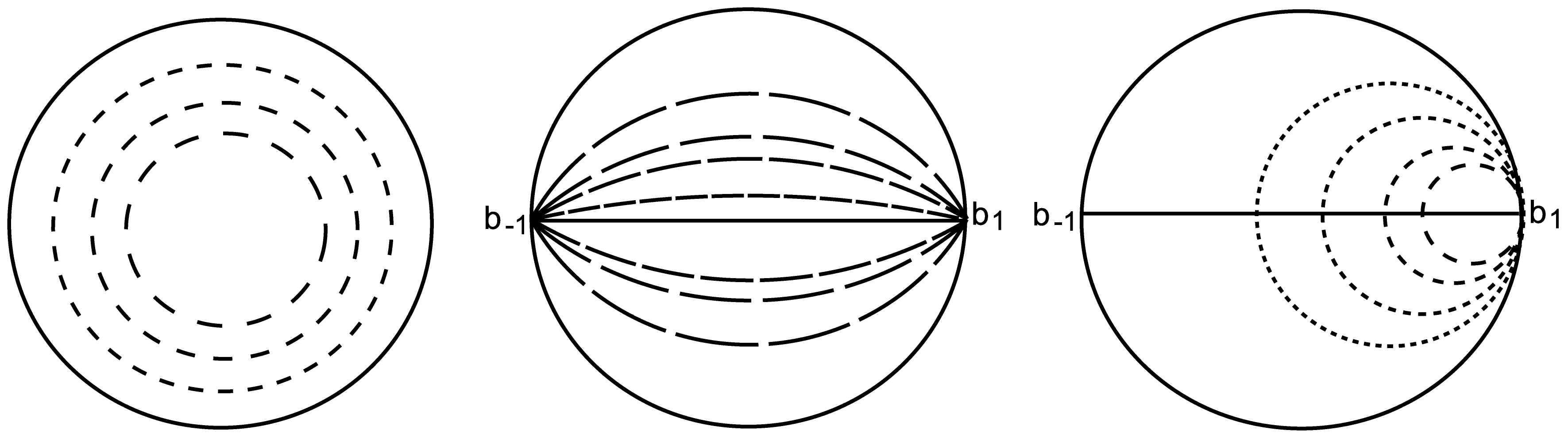}
 \caption{The orbits in the Poincar\'e disk of the three groups $K$, $A$ and $N$.}
 \label{fig:orbits}
\end{figure}

The Iwasawa decomposition allows us to decompose any isometry of $\D$ as the product of at most three elements in the groups, $K,A$ and $N$. Then, it is possible to express each point $z\in\D$ as an image of the origin $O$ by some combination of elements in $K$, $A$ and $N$. There are two important systems of coordinates: 
\begin{enumerate}
\item  the {\em geodesic polar coordinates}: $\text{r}_\phi a_\tau\cdot O=\tanh( \tau/2)e^{i\phi}$ with $\tau=d_\D(z,0)$,
\item the {\em horocyclic coordinates}: $z=n_s a_\tau\cdot O\in\D$, where $n_s$ are the transformations associated with the group $N$ ($s\in\R$) and $a_\tau$ the transformations associated with the subgroup $A$ ($\tau\in\R$).
\end{enumerate}

\subsection{Periodic lattices of the Poincar\'e disk}\label{subs22}

A \textit{Fuchsian group} is a discrete subgroup $\Gamma$ of $\SU$. We are going to be concerned with \textit{fundamental regions} of Fuchsian groups.

\begin{defi}
~To any Fuchsian group we can associate a fundamental region which is the closure, noted $F_\Gamma$, of an open set $\overset{o}{F}_\Gamma\subset \D$ with the following properties:  
\begin{itemize}
\item[(i)] if $\gamma\neq Id\in\Gamma$, then $\gamma\cdot  F_\Gamma\cap \overset{o}{F}_\Gamma = \emptyset$;
\item[(ii)] $\underset{\gamma\in\Gamma}{\bigcup}\, \gamma\cdot F_\Gamma =\D$.
\end{itemize}
The familly $\{ \gamma\cdot F_\Gamma ~|~ \gamma \in\Gamma\}$ is a tesselation of $\D$.
\end{defi}

Fundamental regions may be unnecessarily complicated, in particular they may not be connected. An alternative definition is that of a \textit{Dirichlet region} of a Fuchsian group.

\begin{defi}
~Let $\Gamma$ be a Fuchsian group and $z\in\D$ be not fixed by any element of $\Gamma\setminus Id$. We define the Dirichlet region for $\Gamma$ centered at $z$ to be the set:
\bqs
D_z(\Gamma)=\{z'\in\D ~|~ d_\D(z',z)\leq d_\D(z',\gamma\cdot z)~\forall \gamma \in\Gamma\}.
\eqs
\end{defi}

From \cite{katok:92}, we have the following theorem.
\begin{thm}
~If $z\in\D$ is not fixed by any element of $\Gamma\backslash Id$, then $D_z(\Gamma)$ is a connected fundamental region for $\Gamma$.
\end{thm}

Let $\Gamma$ be a Fuchsian group acting on $\D$ with $\mu(\D/\Gamma)<\infty$, and $F_\gamma$ be a fundamental region for this action. We write $\pi:\D\rightarrow \D/\Gamma$ the natural projection and the points of $\D/\Gamma$ are identified with the $\Gamma$-orbits. The restriction of $\pi$ to $F_\Gamma$ identifies the congruent points of $F_\Gamma$ that necessarily belong to its boundary $\partial F_\Gamma$, and makes $\D/\Gamma$ into an oriented surface. Its topological type is determined by the number of \textit{cusps} and by its \textit{genus}: the number of handles if we view the surface as a sphere with handles. By choosing $F_\Gamma$ to be Dirichlet region, we can find the topological type of $\D/\Gamma$ (in this case $\D/\Gamma$ is homeomorphic to $F_\Gamma/ \Gamma$, see \cite{katok:92}). Furthermore, if finite, the area of a fundamental region (with nice boundary) is a numerical invariant of the group $\Gamma$. Since the area of the quotient space $\D/\Gamma$ is induced by the hyperbolic area on $\D$, the hyperbolic area of $\D/\Gamma$, denoted by $\mu(\D/\Gamma)$, is well defined and equal to $\mu(F_\Gamma)$ for any fundamental region $F_\Gamma$. If $\Gamma$ has a compact Dirichlet region $F_\Gamma$, then $F_\Gamma$ has finitely many sides and the quotient space $\D/\Gamma$ is compact. If, in addition, $\Gamma$ acts on $\D$ without fixed points, $\D/\Gamma$ is a compact Riemann surface and its fundamental group is isomorphic to $\Gamma$.

\begin{defi}
 A Fuchsian group is called cocompact if the quotient space $\D/\Gamma$ is compact.
\end{defi}

When a Fuchsian group is cocompact, then it contains no parabolic elements and its area is finite \cite{katok:92}. Furthermore a fundamental region can always be built as a polygon. The following definition is just a translation to the hyperbolic plane of the definition of an Euclidean lattice.

\begin{defi}
A lattice group of $\D$ is a cocompact Fuchsian group which contains no elliptic element.
\end{defi}

The action of a lattice group has no fixed point, therefore the quotient surface $\D/\Gamma$ is a (compact) manifold and it is in fact a Riemann surface. A remarkable theorem states that any compact Riemann surface is isomorphic to a lattice fundamental domain of $\D$ if and only if it has genus $g\geq 2$ \cite{katok:92}. The case $g=1$ corresponds to lattices in the Euclidean plane (in this case there are three kinds of fundamental domains: rectangles, squares and hexagons). The simplest lattice in $\D$, with genus 2, is generated by an octagon and will be studied in detail in Section~\ref{sec3}. 

Given a lattice, we may ask what is the symmetry group of the fundamental domain $F_\Gamma$, identified with the quotient surface $\D/ \Gamma$. Indeed, this information will play a fundamental role in the subsequent bifurcation analysis. In the case of Euclidean lattice, we recall that the quotient $\R^2/\Gamma$ is a torus $\mathbb{T}^2$ (genus one surface), and the group of automorphisms is $\mathcal{H}\ltimes \mathbb{T}$ where $\mathcal{H}$ is the holohedry of the lattice: ${\mathcal{H}}=\mathbf{D}_2, \mathbf{D}_4$ or $\mathbf{D}_6$ for the rectangle, square and hexagonal lattices respectively. In the hyperbolic case the group of automorphisms of the surface is finite. In order to build this group we need first to introduce some additional definitions.

Tilings of the hyperbolic plane can be generated by reflections through the edges of a triangle $\mT$ with vertices $P$, $Q$, $R$ and angles $\pi/\ell$, $\pi/m$ and $\pi/n$ respectively, where $\ell,~m,~n$ are integers such that $1/\ell+1/m+1/n<1$ \cite{katok:92}. 

\begin{figure}[htp]
\centering \includegraphics[width=0.4\textwidth]{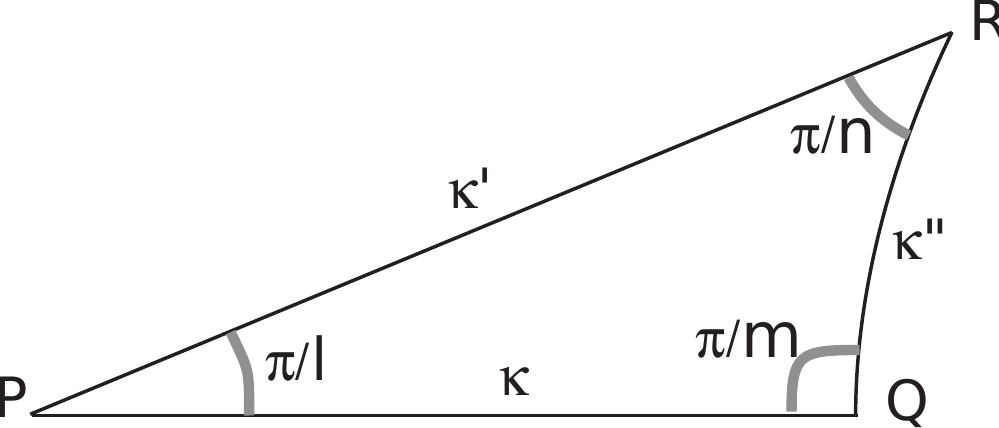}
\caption{{The triangle $\mT$. The values of $l$, $m$, and $n$ are $l=8$, $m=2$ and $n=3$. }}
\label{fig:triangle}
\end{figure}

Remember that reflections are orientation-reversing isometries. We note $\kappa$, $\kappa'$ and $\kappa''$ the reflections through the edges $PQ$, $QR$ and $RP$  respectively (Figure~\ref{fig:triangle}). The group generated by these reflections contains 
an index 2 Fuchsian subgroup $\Lambda$ called a triangle group, which always contains elliptic elements because the product of 
the reflections through two adjacent edges of a polygon is elliptic with fixed point at the corresponding vertex. One easily shows that $\Lambda$ is generated by the rotations of angles $2\pi/l$, $2\pi/m$ and $2\pi/n$ around the vertices $P$, $Q$, $R$ respectively. A fundamental domain of $\Lambda$ is the "quadrangle" $F_\Lambda =\tau\cup\kappa\tau$ \cite{katok:92}. Note that $F
_\Lambda\simeq \D/\Lambda$ is a sphere (genus 0 surface) obtained by identifying the three edges of $\mT$.
The subgroup of hyperbolic translations in $\Lambda$ is a lattice group $\Gamma$, normal in $\Lambda$, whose fundamental domain is filled with copies of the basic tile $\mT$. The group of orientation-preserving automorphisms of $F_\Gamma \simeq \D / \Gamma$ is therefore $G=\Lambda/\Gamma$. From the algebraic point of view, $G$ is generated by three elements $a,~b,~c$ satisfy the relations $a^\ell=b^m=c^n=1$ and $a\cdot b\cdot c=1$. We say that $G$ is an $(l,m,n)$ group.  Taking account of orientation-reversing isometries, the full symmetry group of $F_\Gamma$ is $\g=G\cup \kappa G=G\rtimes \Z_2(\kappa)$. This is also a tiling group of $F_\Gamma$ with tile $\mT$: the orbit $\g \cdot \mT$ fills $F_\Gamma$ and its elements can only intersect at their edges. 

Given a lattice, how to determine the groups $G$ and $\g$? The following theorem gives conditions for this, see \cite{hartshorne:77}. 
\begin{thm}
An $(l,m,n)$ group $G$ is the tiling rotation group of a compact Riemann surface of genus $g$ if and only if its order satisfies the Riemann-Hurwitz relation
$$ |G|=\frac{2g-2}{1 - (\frac{1}{\ell}+\frac{1}{m}+\frac{1}{n})}.$$
\end{thm}
Tables of triangle groups for surfaces of genus up to 13 can be found in \cite{broughton-dirks-etal:01}.

\subsection{Laplace-Beltrami operator}\label{subs23}

\subsubsection{Definition}

The Laplace-Beltrami operator in $\D$ is defined by the expression in cartesian coordinates $z=z_1+ i z_2\in \D$:
\begin{equation}\label{eq:laplace}
\Delta_\D = \frac{(1-z_1^2-z_2^2)^2}{4}\left( \frac{\partial^2}{\partial z_1^2}+\frac{\partial^2}{\partial z_2^2}\right). 
\end{equation}
It has the fundamental property of being equivariant under isometric transformations. More precisely, let us define the following transformation on functions in $\D$. We set
\bqs
T_gf(z) = f(g^{-1}\cdot z) \text{ for all } g\in\U,~z\in\D.
\eqs
Then $\Delta_\D T_gf = T_g\Delta_\D f$. The set of transformations $T_g$ defines a representation of the group $\U$ in the space of functions in $\D$.

\subsubsection{General eigenfunctions of the Laplace-Beltrami operator}

Let $b$ be a point on the circle $\partial \D$. For $z\in \D$, we define the "inner product" $\langle z,b \rangle$ as the algebraic distance to the origin of the (unique) horocycle based at $b$ and passing through $z$. This distance is defined as the hyperbolic signed length of the segment $O\xi$ where $\xi$ is the intersection point of the horocycle and the line (geodesic) $Ob$. This is illustrated in Figure~\ref{fig:horocycle} in the case $b=b_1=1$. Note that $\langle z,b \rangle$ does not depend on the position of $z$ on the horocycle. In other words, $\langle z,b \rangle$ is invariant under the action of the one-parameter group $N$. Each point $z$ of $\D$ can be written in horocyclic coordinates $z=n_s\cdot (a_\tau\cdot O)$. Then $\langle z,b_1 \rangle=\tau$. Note that $\tau$ is negative if $O$ is inside the horocycle and positive otherwise.
\begin{figure}[htp]
 \centering
 \includegraphics[width=0.8\textwidth]{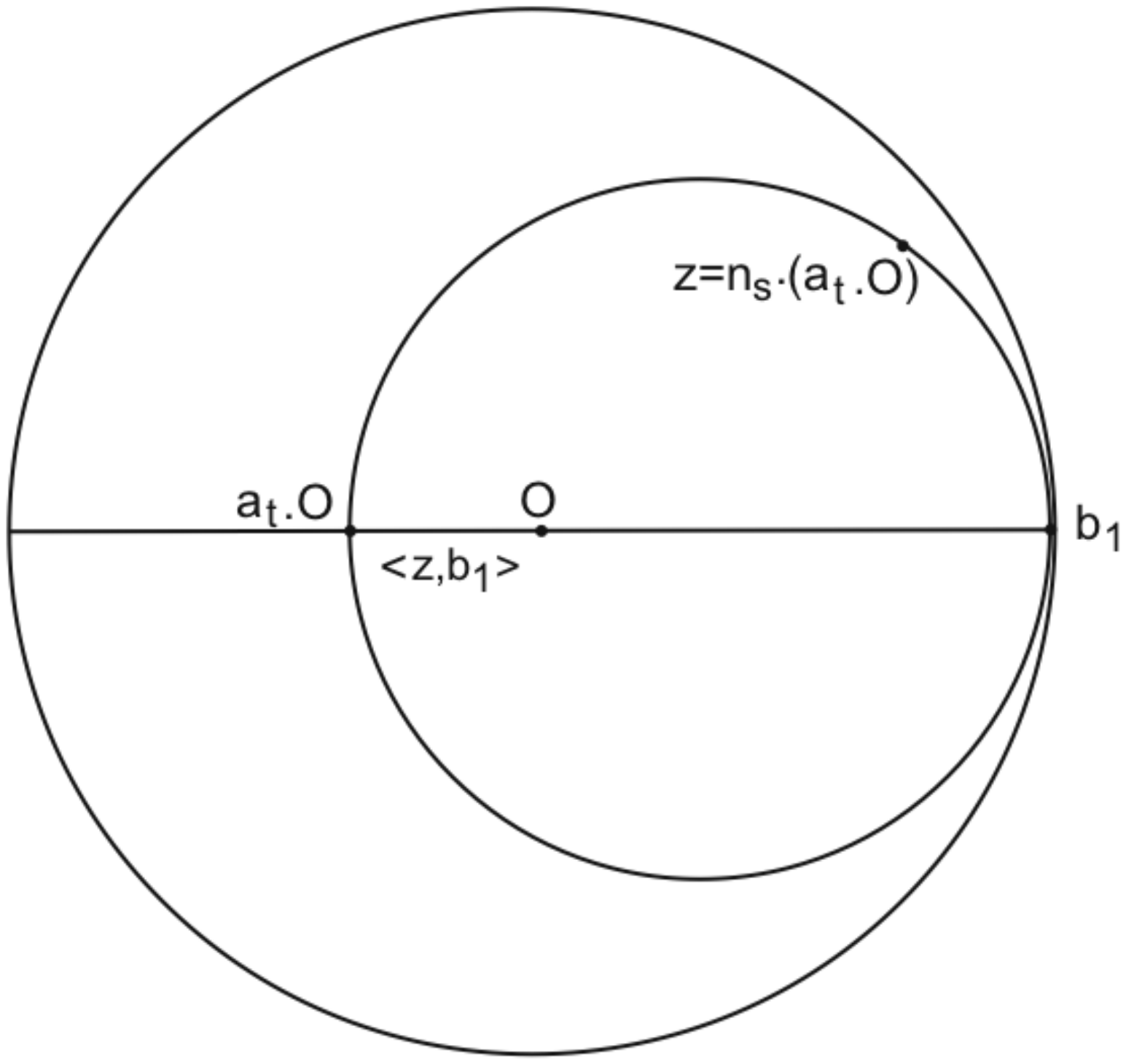}
 \caption{The construction of the "signed inner product" $\langle z,b \rangle$ with $b=1$.}
 \label{fig:horocycle}
\end{figure}

In analogy to the Euclidean plane waves, we define the "hyperbolic plane waves" as the function 
\bqq
\label{eq:hyp_plane_wave}
e_{\rho,b}(z)=e^{\left(i\rho+\frac{1}{2}\right)\langle z,b \rangle}, \quad \rho\in\C.
\eqq
\begin{lem}\label{lem:spectre_laplace_beltrami}
For the Laplace-Beltrami operator defined in equation \eqref{eq:laplace}, we have
\bqq
\label{eq:spectre_laplace_beltrami}
\Delta_\D ~e_{\rho,b}(z)=-\left(\rho^2+\frac{1}{4}\right)e_{\rho,b}(z), \quad \forall~(\rho,b,z)\in\C\times\partial \D \times \D.
\eqq
\end{lem}

\begin{proof}
See \cite{helgason:00}.
\end{proof}

These elementary eigenfunctions allow the construction of general eigenfunctions of $\Delta_\D$. Let $\A(\partial \D)$ denote the space of analytic functions on the boundary $\partial\D$ of the Poincar\'e disk, considered as an analytical manifold. Let $U$ be an open annulus containing $\partial \D$, $\cH(U)$ the space of holomorphic functions on $U$ equipped with the topology of uniform convergence on compact subsets. We identify $\A(\partial \D)$ with the union $\bigcup_U\cH(U)$ and give it the limit topology. The element of the dual space $\A'(\partial \D)$ are called \textit{analytic functionals} or \textit{hyperfunctions}. Since elements of $\A'(\partial \D)$ generalize measures, it is convenient to write
\bqs
\textbf{m}(f)=\int_{\partial \D}f(b)d\textbf{m}(b),\quad f\in \A(\partial \D) \text{ and } \textbf{m}\in \A'(\partial \D).
\eqs
From Helgason's theory \cite{helgason:00} we have the following theorem.
\begin{thm}\label{thm:generalized_eigenfunction}
The eigenfunctions of the Laplace-Beltrami operator on $\D$ are precisely the functions
\bqq
\label{eq:generalized_eigenfunction}
\Psi(z)=\int_{\partial \D} e_{\rho,b}(z) d\textbf{m}_\rho(b),
\eqq
where $\rho\in\C$, $\textbf{m}_\rho\in \A'(\partial \D)$ and the eigenvalue is $-\left(\rho^2+\frac{1}{4}\right)$.
\end{thm} 

Note that real eigenvalues  $-\left(\rho^2+\frac{1}{4}\right)$ of $\Delta_\D$ correspond to taking $\rho$ real or $\rho \in i\R$. The latter case is irrelevant for the following study as it corresponds to exponentially diverging eigenfunctions. Therefore the real spectrum of $L_\D$ is continuous and is bounded from above by $-\frac{1}{4}$.

\subsubsection{Periodic eigenfunctions of the Laplace-Beltrami operator}\label{subsub:eigenfunctionsLaplaceBeltrami}

In the following we will look for solutions of bifurcation problems in $\D$, which are invariant under the action of a lattice group: $\gamma\cdot u(z)=u(\gamma^{-1}\cdot z)=u(z)$ for $\gamma\in\Gamma$. This reduces to look at the problem restricted to a fundamental region with suitable boundary conditions imposed by the $\Gamma$-periodicity, or, equivalently, to looking for the solutions of the problem projected onto the orbit space $\D/\Gamma$ (which inherits a Riemannian structure from $\D$). Because the fundamental region is compact, it follows from general spectral theory that $\Delta_\D$ is self-adjoint, non negative and has compact resolvent in $L^2(\D/\Gamma)$ \cite{buser:92}. Hence its spectrum consists of real positive and isolated eigenvalues of finite multiplicity.  

Coming back to Theorem~\ref{thm:generalized_eigenfunction}, we observe that those eigenvalues $\lambda$ of $\Delta_\D$ which correspond to  $\Gamma$-invariant eigenfunctions, must have $\rho\in\R$ or $\rho\in i\R$. The case $\rho$ real corresponds to the Euclidean situation of planar waves with a given wave number, the role of which is played by $\rho$ in $\D$. In this case the eigenvalues of $\Delta_\D$ satisfy $\frac{1}{4}< \lambda$. On the other hand there is no Euclidean equivalent of the case $\rho\in i\R$, for which the eigenvalues $0 < \lambda \leq \frac{1}{4}$ are in finite number. It turns out that such "exceptional" eigenvalues do not occur for "simple" groups such as the octagonal group to be considered in more details in the Section \ref{sec3}. This follows from formulas which give lower bounds for these eigenvalues. Let us give two examples of such estimates (derived by Buser \cite{buser:92}, see also \cite{iwaniec:02}): (i) if $g$ is the genus of the surface $\D/\Gamma$, there are at most $3g-2$ exceptional eigenvalues; (ii) if $d$ is the diameter of the fundamental region, then the smallest (non zero) eigenvalue is bounded from below by $\left( 4\pi\, \sinh \frac{d}{2}\right)^{-2}$. 

Suppose now that the eigenfunction in Theorem~\ref{thm:generalized_eigenfunction} is $\Gamma$-periodic. Then the distribution $\mathbf{m}_\rho$ satisfies  the following equivariance relation \cite{pollicott:89}. Let $\gamma (\theta)$ denote the image of $\theta\in\partial \D$ under the action of $\gamma\in\Gamma$. Then 
\bqq
\label{eq:action_periodic_distrib}
\mathbf{m}_\rho(\gamma\cdot\theta)=|\gamma'(\theta)|^{\frac{1}{2}+i\rho}\,\mathbf{m}_\rho(\theta).
\eqq

\begin{rmk}
As observed by \cite{series:87}, this condition is not compatible with $\mathbf{m}_\rho$ being a "nice" function. In fact, not only does there exist no explicit formula for these eigenfunctions, but their approximate computation is itself an uneasy task. We shall come back to this point in the next chapter.
\end{rmk}

\subsection{The Helgason-Fourier transform}\label{subsec24}

Based on the elementary eigenfunctions \ref{eq:hyp_plane_wave}, Helgason built a Fourier transform theory for the Poincar\'e disc, see \cite{helgason:00} which we recall now.

\begin{defi}\label{def:helgason_transform}
 If $f$ is a complex-valued function on $\D$, its Helgason-Fourier transform is defined by
\bqq
\label{eq:helgason_transform}
\tilde{f}(\rho,b)=\int_\D f(z) e^{\left(-i\rho +\frac{1}{2}\right)\langle z,b \rangle}\text{dm}(z)
\eqq
for all $\rho\in\C$, $b\in\partial \D$ for which this integral exists.
\end{defi}

If we denote $\mathcal{D}(\D)$, the set of differentiable functions of compact support then the following inversion theorem holds.

\begin{thm}\label{thm:inversion_formula}
 If $f\in\mathcal{D}(\D)$, then
\bqq
\label{eq:inversion_formula}
f(z)=\frac{1}{2\pi}\int_\R\int_{\partial \D}\tilde{f}(\rho,b)e^{\left(i\rho +\frac{1}{2}\right)\langle z,b \rangle}\rho \tanh\left(\pi\rho\right)d\rho db
\eqq
where $db$ is the circular measure on $\partial \D$ normalized by $\int_{\partial \D} db=1$.
\end{thm}

\subsection{Existence, uniqueness and regularity of the solutions}\label{subsec25}

If we think of $u$, solution of the Swift-Hohenberg equation \eqref{eq:sh}, as a function evolving in space and time (for example an averaged potential membrane \cite{faye-chossat-etal:11,chossat-faye-etal:11}) it is natural to impose that $u$ should be uniformly bounded, i.e., a function in $L^{\infty}(\D)$. Indeed, the space of bounded functions contains all types of solutions of the Swift-Hohenberg: fronts (solutions connecting two homogeneous states), periodic solutions and localized solutions. Note that $L^p(\D)$-spaces with $p<\infty$ are only relevant for localized solutions. In addition, this natural property should be preserved under the time evolution. In other words, the problem is to know if the Swift-Hohenberg equation \eqref{eq:sh} defines a regular semi-flow in the space $L^{\infty}(\D)$ which is global in time. Note first, that the Cauchy problem in this space is not difficult to handle and it is easy to prove that there exists an unique solution of \eqref{eq:sh} for a small time interval which depends upon the $L^{\infty}(\D)$ norm of the initial data. In the Euclidean case, this global time problem  has received much interests in the past decades and can be solved using local energy estimates  \cite{collet-eckmann:92,collet:94}. The proof of an analog global time existence theorem for the hyperbolic plane is far beyond the scope of this paper and instead of working in a functional space that suits for all types of solutions, we prefer to use specific functional space for each type of solutions.

The types of solutions that we are interested in are periodic and localized solutions, which can be defined on well chosen Sobolev spaces where the existence of global solutions is straightforward.

We also recall that the Swift-Hohenberg equation \eqref{eq:sh} posed on $\D$ is a gradient system,
\bqs
u_t=-\nabla \mathcal{E}(u),
\eqs
in $H^2(\D)$, where the energy functional $\mathcal{E}$ is given by
\bqq
\label{eq:energy_functional}
\mathcal{E}(u)=\int_\D \left[ \frac{\left[(\alpha+\Delta_\D)u\right]^2}{2}+\frac{\lambda u^2}{2}-\frac{\nu u^3}{3}+\frac{\eta u^4}{4}\right]\text{dm}(z), \quad z\in \D,
\eqq
and the gradient $\nabla \mathcal{E}(u)=\frac{\delta\mathcal{E}}{\delta u}(u)$ of $\mathcal{E}$ with respect to $u$ is computed in $L^2(\D)$. In particular, $\mathcal{E}$ decreases strictly in time along solution of \eqref{eq:sh} unless the solution is stationary. While we cannot evaluate the energy functional along periodic patterns as they are not localized, whence the integral \eqref{eq:energy_functional}  may not exist, we may, however, define a local energy by integrating over one spatial period of the underlying periodic pattern.

\section{Bifurcation of $\Gamma$-periodic patterns}\label{sec3}

\subsection{Reduction of the problem}\label{subsec31}
\subsubsection{Linear stability analysis in the class of $\Gamma$-periodic perturbations}\label{subsub:linear}
Let us study the stability of the trivial state $u=0$ of equation \ref{eq:sh}. If we look for perturbations in the form of hyperbolic plane waves $e^{\sigma t}e_{\rho,b}(z)$, solving the linearized Swift-Hohenberg equation comes back to solving a "dispersion relation"
\bqq \label{eq:dispersion relation}
\sigma = -(\alpha^2-\rho^2-\frac{1}{4})^2+\lambda
\eqq
Writing $z$ in horocyclic coordinates $(s,\tau)$ (see \ref{subs23}) we may want to look for solutions which are independent of the coordinate $s$ along the horocycles based at the point $b$ on $\partial\D$ and which are periodic in $\tau$. This would imply $\rho=a+i/2$ with $a\in\R$. Indeed $e_{\rho,b}(z)=e^{(\rho i+1/2)\tau}$. However in general $\sigma$ is a complex eigenvalue in this case, which we shall treat in Section \ref{sec5}.  

We now look for solutions within the restricted class of functions which are invariant under the action of a lattice group $\Gamma$ and which are $\text{L}^2$ in the fundamental domain of periodicity. This comes back to looking for Equations (\ref{eq:sh}) projected onto the compact Riemann surface $\D/\Gamma$. As noticed in \ref{subsub:eigenfunctionsLaplaceBeltrami}, the spectrum of the Laplace-Beltrami operator in $\D/\Gamma$ is made of isolated, real and positive eigenvalues with finite multiplicities, which we label by increasing order and which we write $\rho_k^2+\frac{1}{4}$, $k\in\N$. Moreover the corresponding eigenfunctions are expressed as in (\ref{eq:generalized_eigenfunction}) with the measure $m_{\rho_k}$ satisfying the equivariance property \ref{eq:action_periodic_distrib}. We assume that all eigenvalues are greater than $\frac{1}{4}$, which implies $\rho_k\in\R$. This is known to be true for example when $\Gamma$ is the regular octagonal lattice, a case that we shall analyze in detail in the next subsections.    

It follows from the dispersion relation \eqref{eq:dispersion relation} restricted to $\Gamma$-invariant perturbations that 
\begin{enumerate}
\item When $\lambda<0$ all $\Gamma$ periodic perturbations are damped and $0$ is a stable state. 
\item The neutral stability curve consists in a discrete set of points $(\rho_k,\lambda_k)$ on the curve $\lambda=(\alpha^2-\rho^2-\frac{1}{4})^2$. There exists a minimum of $\lambda_k$ at some value $\rho_{k_c}$ as we see in Figure \ref{fig:neutralstab}. Note that by adjusting the value of $\alpha$ we can set $\lambda_{k_c}=0$, which we shall assume in the rest of this section.
\begin{figure}[htp]
 \centering
 \includegraphics[width=0.5\textwidth]{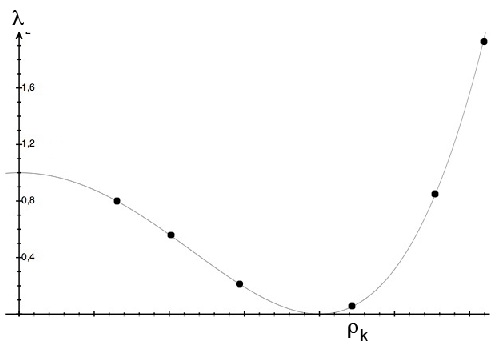}
 \caption{The neutral stability curve (schematic).}
 \label{fig:neutralstab}
\end{figure}
\item In general the critical parameter value $\lambda_k$ is associated with a unique wave number $\rho_{k_c}$. In this case the critical eigenspace is an irreducible representation of the group $\g$ of symmetries of $\D/\Gamma$.
\end{enumerate}  

\subsubsection{Center manifold reduction}\label{subsection:CM}
From these points it follows that we can apply a center manifold reduction to equation (\ref{eq:sh}). We quickly recall the procedure. Let $V$ be the critical eigenspace and $x\in V$. We write $V^\perp$ for the orthogonal complement of $V$ in $\text{L}^2(\D/\Gamma)$. The center manifold theorem asserts the following \cite{chossat-lauterbach:00, haragus-iooss:10}:
\begin{itemize}
\item[(i)] There exists a neighborhood $U$ of $(0,0)$ in $V\times\R$ and a $C^p$ map $\Psi:~V\times\R\rightarrow V^\perp$ ($p>2$) such that the graph of $\Psi$ in $L^2(\D/\Gamma)$ is a locally flow invariant, attracting $C^p$ manifold of equation (\ref{eq:sh}) for $(x,\lambda)\in U$.
\item[(ii)] The map $\Psi$ is equivariant under the action of the group $\g$: for any $g\in\g$, $\Psi(g\cdot x,\lambda)=g\cdot \Psi(x,\lambda)$ for all $(x,\lambda)\in U$.
\end{itemize}
Therefore the bifurcation analysis can be reduced to the projection of equation (\ref{eq:sh}) onto $V$ with $u=x+\Psi(x,\lambda)$. We write this equation  
\bqq \label{eq:bifurcation}
\frac{dx}{dt} = F(x,\lambda)
\eqq
and by point (ii) above $F(g\cdot x,\lambda)=g\cdot F(x,\lambda)$ for all $g\in\g$ and $(x,\lambda)\in V\times\R$. From the assumptions we have $F(0,0)=0$ and $D_xF(0,0)=0$. The Taylor expansion of $F$ can be computed by a recursive method \cite{haragus-iooss:10}. However $V$ can only be known approximately because the eigenfunctions of $\Delta_\D$ do not have explicit expressions and are in general difficult to compute (see below).

We can now apply to (\ref{eq:bifurcation}) the methods of equivariant bifurcation theory \cite{chossat-lauterbach:00} in order to determine the bifurcation diagram.

\subsubsection{Equivariant branching lemma} 
When a differential equation like (\ref{eq:bifurcation}) is invariant under the action of a symmetry group, this imposes geometrical constraints which can be exploited to determine the bifurcation diagram or more generally, to analyze the dynamics. Here we recall some basic facts and methods which will be applied in the next subsection.
\begin{defi}
Given a representation (a linear action) of the group $\g$ in a space $V$, we say that it is {\em absolutely irreducible} if the only endomorphisms of $V$ which commute with this action are scalar multiples of the identity.
\end{defi} 
One can easily check that an absolutely irreducible representation is irreducible, meaning that the only subspaces of $V$ which are invaraint under the representation are $\{0\}$ and $V$, but the converse may not be true. 

We now assume that $\g$ acts absolutely irreducibly in $V$.
This has two consequences:
\begin{itemize}
\item[(i)] $F(0,\lambda)=0$ for all $\lambda$. Indeed $F(0,\lambda)$ is fixed by all $g\in\g$ and the only point which has this property is $0$.
\item[(ii)] $D_xF(0,\lambda)=c(\lambda)\mathcal{I}$ (identity operator).
\end{itemize} 
In the case of equation (\ref{eq:sh}), (ii) is automatically satisfied. Indeed the linear part is $-(\alpha^2+\Delta_\D)^2u+\lambda u$ and we have seen that $V$ is the kernel of $-(\alpha^2+\Delta_\D)^2$ in $L^2(\D/\Gamma)$. Let $P$ be the orthogonal projection on $V$. Setting $x=Pu$, we see that $D_xF(0,\lambda)x=P(-(\alpha^2+\Delta_\D)^2x+\lambda x)=\lambda x$.

\begin{defi}
An isotropy group for the action of $\g$ in $V$ is the largest subgroup which fixes a point in $V$. For example $\g$ itself is an isotropy group (it fixes $0$).
\end{defi}
Different points in $V$ may have the same isotropy group. Let $H$ be an isotropy group. Since the action is linear, the set $V_H=\{x\in V~|~H\cdot x=x\}$ is a subspace of $V$ called the {\em fixed point subspace} of $H$. $V_H$ contains points with higher isotropy (at least it contains $0$) and this induces a stratification of the fixed point subspace. The subset of points which have exactly $H$ as isotropy subgroup is open in $V_H$ and is called the principal stratum of $V_H$. The next proposition is straightforward but of high consequences.
\begin{prop}
Fixed point subspaces are invariant under equivariant maps in $V$.
\end{prop}
We can now state the main result of this subsection (equivariant branching lemma for the Swift-Hohenberg equation).
\begin{thm} \label{thm:equivariant branching lemma}
Under the above hypotheses, suppose that $H$ is an isotropy subgroup such that $\dim{V_H}=1$. Then a branch of $\Gamma$-periodic steady states of (\ref{eq:sh}) bifurcates from $0$ at $\lambda=0$ in $V_H$. Let $N(H)$ be the normalizer of $H$ in $\g$. If $N(H)/H\simeq\Z_2$, then the bifurcation is a pitchfork in $V_H$.
\end{thm}
\begin{proof}
By the previous proposition, the bifurcation equation (\ref{eq:bifurcation}) restricts to the invariant axis $V_H$. We write this scalar equation $\dot x=f(x,\lambda)$ where by assumptions, $f(0,\lambda)=0$ and $f'(0,\lambda)=\lambda$. It follows that we can rewrite this equation 
\bqq
\dot x = x[\lambda+h(x,\lambda)]
\eqq
where $h(x,\lambda)=O(x^2)$. Having eliminated the solution $x=0$ we obtain the bifurcated branch by the implicit function theorem. To prove the second part of the theorem, remark that $N(H)$ keeps $V_H$ invariant (as a set). Either the group $N(H)/H=1$ or $N(H)/H\simeq\Z_2$. This is because $\g$ is a group of isometries whose action in $V$ is orthogonal (orthogonal matrices). In the second case $f(\cdot,\lambda)$ is an odd function which implies that the bifurcation is a pitchfork.  
\end{proof}
Of course it may also happen that solutions generically bifurcate in the principal stratum of a fixed point subspace which has dimension greater than 1. This point will be addressed in the example of the next subsection.
\begin{defi}
Any generic $\Gamma$-periodic steady-state satisfying the hypotheses of the Equivaraint Branching Lemma is called an H-planform.
\end{defi}

Finally, if $x'=g\cdot x$, then the isotropy subgroup of $x'$ is $H'=g\cdot H\cdot g^{-1}$. The image of $x$ under $\g$ is called its $\g$-orbit. Equilibria or more generally flow-invariant sets belonging to the same group orbit share the same properties.

\subsection{Example: the octagonal lattice}\label{subsec:octagonal lattice}

\subsubsection{The octagonal lattice and its symmetries}\label{subsection:octlattice}
The octagonal lattice group $\Gamma$ is generated by the following four hyperbolic translations (also called boosts), see \cite{balazs-voros:86}:

\begin{equation}
\label{eq:boostg}
g_0 = \left(\begin{array}{cc}1+\sqrt{2} & \sqrt{2+2\sqrt{2}} \\ \sqrt{2+2\sqrt{2}} & 1+\sqrt{2}\end{array}\right)
\end{equation}
and $g_j = r_{j\pi/4}g_0r_{-j\pi/4}$, $j=1,2,3$, where $r_\varphi$ indicates the rotation of angle $\varphi$ around the origin in $\D$. The fundamental domain of the lattice is a regular octagon $\Oct$. The opposite sides of the octagon are identified by periodicity, so that the corresponding quotient surface $\D/\Gamma$ is isomorphic to a "double doughnut" (genus two surface) \cite{balazs-voros:86}.

We now determine what is the full symmetry group of the octagonal lattice, or equivalently, of the surface $\D/\Gamma$. Clearly the symmetry group of the octagon itself is part of it. This is the dihedral group $\mD_8$ generated by the rotation $r_{\pi/4}$ and by the reflection $\kappa$ through the real axis, but there is more. We have seen in Section \ref{subs22} that the group $\g$ is generated by the reflecitons through the edges of the elmentary trinagle tiling $\D/\Gamma$. The smallest triangle (up to symmetry) with these properties is the one shown in Figure \ref{fig:triangle}. It has angles $\pi/8$, $\pi/2$ and $\pi/3$ at vertices $P=O$ (the center of $\D$), $Q$, $R$ respectively, and its area is, by Gauss-Bonnet formula, equal to $\pi/24$. There are exactly 96 copies of $\tau$ filling the octagon, hence $|\g|=96$. Figure \ref{fig:tesselation} shows this tessellation of $\Oct$ by triangles. 
\begin{figure}
\centering
\includegraphics[width=0.6\textwidth]{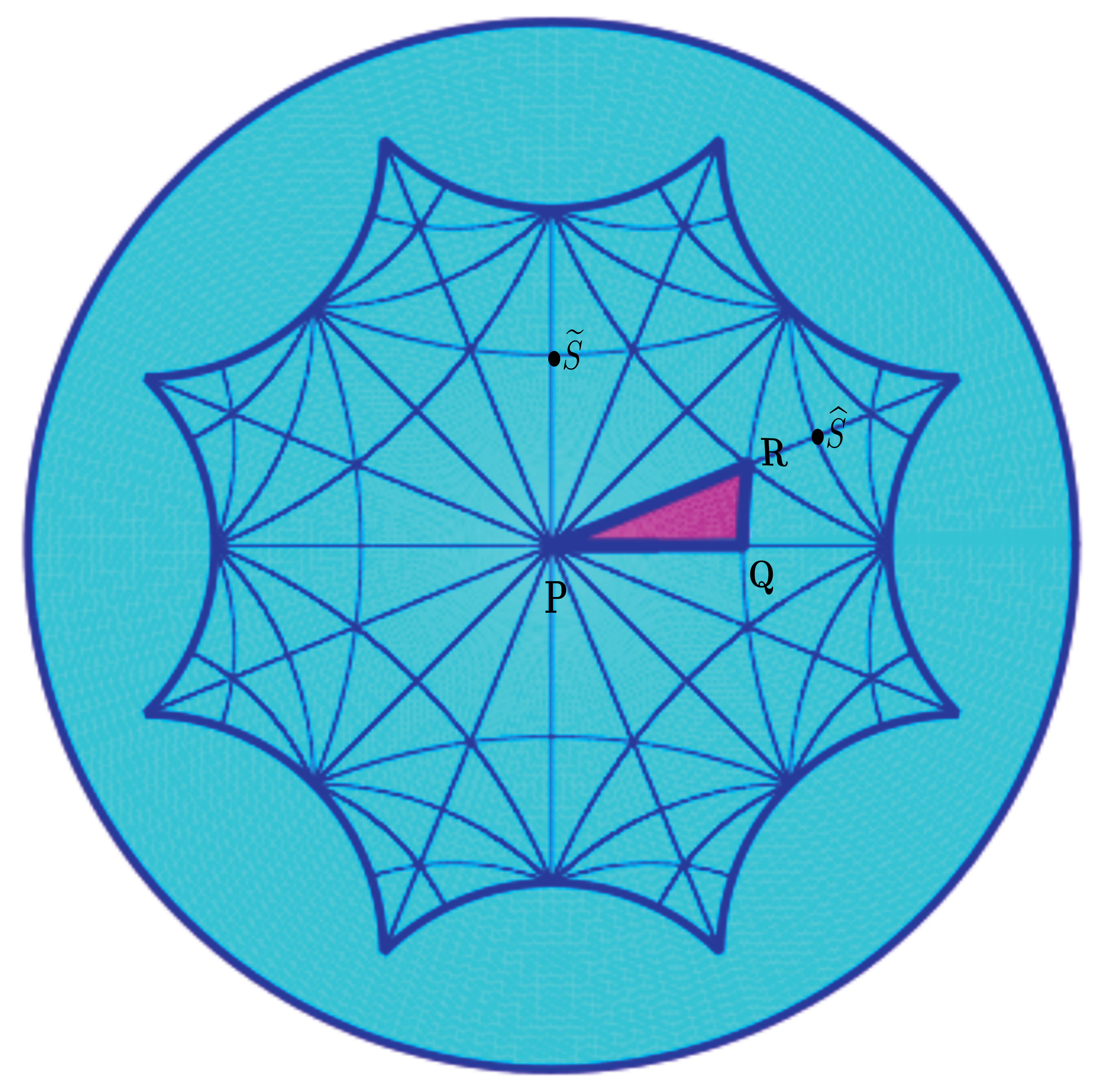}
\caption{Tesselation of the regular hyperbolic octagon with the triangle $\mathbf{T}$ of Figure \ref{fig:triangle}, colored in purple in the plot. We define two points $\widehat{S}$ and $\widetilde{S}$.  $\widehat{S}$ is the center of the rotation $\hat \sigma$ by $\pi$ (mod $\Gamma$), see text in subsection \ref{subsection:octlattice}. $\widetilde{S}$ is the center of the rotation $\tilde{\sigma}$ by $\pi$ (mod $\Gamma$), see text in subsection \ref{subsection:octlattice}.}
\label{fig:tesselation}
\end{figure}

The index 2 subgroup of orientation preserving transformations has 48 elements. In \cite{broughton:91} it has been found that $G\simeq GL(2,3)$, the group of invertible $2\times 2$ matrices over the field $\Z_3$. In summary:
\begin{prop}
The full symmetry group $\g$ of $\D/\Gamma$ is $G\cup\kappa G$ where $G\simeq GL(2,3)$ has 48 elements.
\end{prop}

The isomorphism between $GL(2,3)$ and $G$ can be built as follows. We use the notation $\Z_3=\{0,1,2\}$ and we define:
\begin{itemize}
\item $\rho$ the rotation by $\pi/4$ centered at $P$ (mod $\Gamma$), 
\item $\sigma$ the rotation by $\pi$ centered at $Q$ (mod $\Gamma$) 
\item $\epsilon$ the rotation by $2\pi/3$ centered at $R$ (mod $\Gamma$).  
\end{itemize}
Then we can proceed with the following identification.
\begin{equation*}
\rho = \left(\begin{array}{cc}0&2\\2&2\end{array}\right),~\sigma = \left(\begin{array}{cc}2&0\\0&1\end{array}\right),~\epsilon=\left(\begin{array}{cc}2&1\\2&0\end{array}\right)
\end{equation*}
since these matrices satisfy the conditions $\rho^8=\sigma^2=\epsilon^3=Id$ and $\rho\sigma\epsilon = Id$. Note that $\rho^4=-Id$ where $Id$ is the identity matrix. We shall subsequently use this notation. 
The group $GL(2,3)$ is described e.g. in \cite{lang:93}. \\
The full symmetry group $\g$ is generated by $G$ and $\kappa$, the reflection through the real axis in $\D$. We further define 
\begin{itemize}
\item $\kappa'=\rho\kappa$ the reflection through the side $PR$ of the triangle $\tau$,
\item $\kappa''=\epsilon\kappa'=\sigma\kappa$ the reflection through the third side $QR$. 
\end{itemize}

The group $\g$ and its representations have been studied with the help of the computer algebra software GAP \cite{Gap}. Details are found in \cite{chossat-faye-etal:11}. The main result is summarized in the next proposition and table.

\begin{prop}
There are 13 conjugacy classes in $\g$, hence 13 irreducible representations of $\g$. 4 of them have dimension 1, 2 have dimension 2, 4 have dimension 3 and 3 have dimension 4. Their characters are denoted $\chi_j$, $j=1,...,13$. The character table is shown in table \ref{table:caracteres}. Moreover all these representations are real absolutely irreducible.

\begin{table}[h] 
\begin{center}
\begin{tabular}{|c|c|c|c|c|c|c|c|c|c|c|c|c|c|}\hline Class \# & 1 & 2 & 3 & 4 & 5 & 6 & 7 & 8 & 9 & 10 & 11 & 12 & 13  \\\hline Representative & $Id$ & $\rho$ & $\rho^2$ & $-Id$ & $\sigma$ & $\epsilon$ & $-\epsilon$ & $\kappa$ & $\kappa'$ & $\widehat{\sigma}\kappa$ & $\rho\widehat{\sigma}\kappa$  & $\epsilon\kappa$ & $-\epsilon\kappa$  \\\hline\hline $\chi_1$ &  1 & 1 & 1 & 1 & 1 & 1 & 1 & 1 & 1 & 1 & 1 & 1 & 1 \\\hline $\chi_2$ & 1 & -1 & 1 & 1 & -1 & 1 & 1 & 1 & -1 & -1 & 1 & 1 & 1 \\\hline $\chi_3$ & 1 & -1 & 1 & 1 & -1 & 1 & 1 & -1 & 1 & 1 & -1 & -1 & -1 \\\hline $\chi_4$ & 1 & 1 & 1 & 1 & 1 & 1 & 1 & -1 & -1 & -1 & -1 & -1 & -1 \\\hline $\chi_5$ & 2 & 0 & 2 & 2 & 0 & -1 & -1 & -2 & 0 & 0 & -2 & 1 & 1 \\\hline $\chi_6$ & 2 & 0 & 2 & 2 & 0 & -1 & -1 & 2 & 0 & 0 & 2 & -1 & -1 \\\hline $\chi_7$ & 3 & 1 & -1 & 3 & -1 & 0 & 0 & -1 & -1 & 1 & 3 & 0 & 0 \\\hline $\chi_8$ & 3 & 1 & -1 & 3 & -1 & 0 & 0 & 1 & 1 & -1 & -3 & 0 & 0 \\\hline $\chi_9$ & 3 & -1 & -1 & 3 & 1 & 0 & 0 & 1 & -1 & 1 & -3 & 0 & 0 \\\hline $\chi_{10}$ & 3 & -1 & -1 & 3 & 1 & 0 & 0 & -1 & 1 & -1 & 3 & 0 & 0 \\\hline $\chi_{11}$ & 4 & 0 & 0 & -4 & 0 & -2 & 2 & 0 & 0 & 0 & 0 & 0 & 0 \\\hline $\chi_{12}$ & 4 & 0 & 0 & -4 & 0 & 1 & -1 & 0 & 0 & 0 & 0 & $\sqrt{3}$ & $-\sqrt{3}$ \\\hline $\chi_{13}$ & 4 & 0 & 0 & -4 & 0 & 1 & -1 & 0 & 0 & 0 & 0 & $-\sqrt{3}$ & $\sqrt{3}$ \\\hline \end{tabular}
\end{center}
\caption{Irreducible characters of $\g$. Notation: $\widehat{\sigma}=\epsilon\sigma\epsilon^{-1}$ (rotation by $\pi$ centered at $\hat S$ in Fig. \ref{fig:tesselation})}\label{table:caracteres} 
\end{table}
\end{prop}
\begin{proof} 
We only show here the second part of the proposition (see \cite{chossat-faye-etal:11} for the first part). \\
Absolute irreducibility is obvious for the one dimensional representations and it follows as a corollary form the next proposition \ref{prop:identification irreps dim<3} for representations $\chi_5$ to $\chi_{10}$. \\
It remains to prove the result for the four dimensional representations $\chi_{11}$, $\chi_{12}$ and $\chi_{13}$. For this we consider the action of the group $\mD_8$ generated by $\rho$ and $\kappa$ (the standard symmetry group of an octagon), as defined by either one of these 4D irreducible representations of $\g$. We observe from the character table that in all cases, the character of this action is $\chi(\rho)=0$, $\chi(\rho^2)=0$, $\chi(-Id)=-4$, $\chi(\rho^3)=0$ ($\rho$ and $\rho^3$ are conjugate in $\g$), and $\chi(\kappa)=\chi(\kappa')=0$. We can determine the isotypic decomposition for this action of $D_8$ from these character values. The character tables of the four one dimensional and three two dimensional irreducible representations of $D_8$ can be computed easily either by hand (see \cite{miller:72} for the method) or using a computer group algebra software like GAP. For all one dimensional characters the value at $-Id$ is $1$, while for all two dimensional characters, the value at $-Id$ is $-2$. Since $\chi(-Id)=-4$, it is therefore not possible to have one dimensional representations in this isotypic decomposition. It must therefore be the sum of two representations of dimension 2. Moreover, since $\chi(\rho)=\chi(\rho^2)=\chi(\rho^3)=0$, it can't be twice the same representation. In fact it must be the sum of the representations whose character values at $\rho$ are $\sqrt{2}$ and $-\sqrt{2}$ respectively. Now, these representations are absolutely irreducible (well-know fact which is straightforward to check), hence any $\mD_8$-equivariant matrix which commutes with this action decomposes into a direct sum of two scalar $2\times 2$ matrices $\lambda I_2$ and $\mu I_2$ where $\lambda$ and $\mu$ are real. But the representation of $\g$ is irreducible, hence $\lambda=\mu$, which proves that it is also absolutely irreducible. 
\end{proof}

The 2 and 3 dimensional irreducible representations of $\g$ must act as some particular irreducible representations of finite subgroups of $O(2)$ (planar case) and of $O(3)$ (dimension 3). The next result specifies these actions. The idea is that these representations are not faithful. In other words for every $j=5,\dots,10$ the kernel $\ker{\chi_j}=H_j\neq\{0\}$. Let $\widetilde{\chi_j}$ be the representation induced by the projection $\g\rightarrow \g/H_j$. Then $\widetilde{\chi_j}$ is isomorphic to some irreducible representation of either a dihedral group $D_n$ (symmetry group of the regular $n$-gon) in the 2-d case or the symmetry group of a platonic body in the 3-d case. We write $\mathbb{O}$ for the octahedral group (direct symmetries, {\em i.e.} rotations of a cube), and $\mT_d$ for the full symmetry group of the tetrahedron. These two groups are isomorphic and have $24$ elements. We also write $\Z_2$ for the 2-element group generated by the antipodal reflection in $\R^3$. 
\begin{prop} \label{prop:identification irreps dim<3}
\begin{itemize}
\item[(i)] $\widetilde{\chi_5}$ is isomorphic to $\mD_6$ acting in $\R^2$; 
\item[(ii)] $\widetilde{\chi_6}$ is isomorphic to $\mD_3$ acting in $\R^2$;
\item[(iii)] $\widetilde{\chi_7}$ is isomorphic to $\mathbb{O}$ acting in $\R^3$; 
\item[(iv)] $\widetilde{\chi_8}$ is isomorphic to $\mathbb{O}\times\Z_2$ acting in $\R^3$ (full symmetry group of the cube);  
\item[(v)] $\widetilde{\chi_9}$ is isomorphic to the action of  $\mT_d\times\Z_2$ in $\R^3$;
\item[(vi)] $\widetilde{\chi_{10}}$ is isomorphic to the action of $\mT_d$ in $\R^3$.
\end{itemize}
\end{prop} 
\begin{proof}
The proof of this proposition is given in \cite{chossat-faye-etal:11}. We show a slightly different proof for the cases (iii) to (vi). We see from the character table that in $\chi_7$ to $\chi_{10}$, the element $-Id$ acts trivially. Let us write $C_2=\{Id,-Id\}$. Then $GL(2,3)/C_2$ is a 24 element group which is known to be isomorphic to the permutation group $S_4$, hence to $\mathbb{O}$. Therefore $\g/C_2\simeq \mathbb{O}\times \Z_2$. It follows that the four irreducible representations $\widetilde{\chi_7}$ to $\widetilde{\chi_{10}}$ must be isomorphic to the four irreducible representations of dimension three of the symmetry group of the cube $\mathbb{O}\times \Z_2$. Comparison of the character tables of these representations leads to the result (one can make use of GAP to obtain these tables).
\end{proof} 
The actions listed in the proposition are absolutely irreducible. It follows that the representations $\chi_5$ to $\chi_{10}$ are also absolutely irreducible.

\subsubsection{Classification of H-planforms and bifurcation diagrams}\label{subsubsec:octagonal classification}

It follows from Proposition \ref{prop:identification irreps dim<3} that the generic bifurcation diagrams for all representations $\chi_1$ to $\chi_{10}$ are classical. These diagrams are well-known for $\mD_3$ and $\mD_6$, \cite{golubitsky-stewart-etal:88} and they are also known for the 3-dimensional cases. \\
Let us for example consider cases (iii) and (iv). By Theorem \ref{thm:equivariant branching lemma} each axis of symmetry of $\mathbb{O}$ and $\mathbb{O}\times\Z_2$ gives rise to a branch of equilibria with that symmetry. These are the symmetry axes of a cube and they here of three types: 4 fold symmetry (passing through the centers of opposite faces), 3 fold symmetry (passing through opposite vertices) and 2 fold symmetry (through opposite edges). Moreover each symmetry axis is non trivially mapped to itself by some transformation in the group. It follows that there are three different types of pitchfork branches of equilibria in these cases. For the stability and (generic) non existence of other branches we refer to \cite{melbourne:86}. \\
What remains to do in these cases is to actually compute the bifurcated states, which is quite involved since the eigenfunctions of $\Delta_{\D/\Gamma}$ can't be expressed in an explicit manner. This part will be discussed in the next subsection. 

In the rest of this paragraph we study the case of 4 dimensional representations $\chi_{11}$ to $\chi_{13}$, which do not belong to the list in Proposition  \ref{prop:identification irreps dim<3}. This was studied in details in \cite{faye-chossat:11}, to which we refer for proofs. Here we outline the main results and sketch the proofs. 

First note that the dimension of the fixed point subspaces of isotropy subgroups for a representation $\chi_j$ can be determined from the character table \ref{table:caracteres} thanks to the formula \cite{golubitsky-stewart-etal:88,chossat-lauterbach:00}
\begin{equation} \label{eq:formuletrace}
{\rm dim}(V^{H})=\frac{1}{|H|} \sum_{h\in H}\chi_j(h)
\end{equation}
In order to find the H-planforms we therefore need to find those isotropy subgroups such that $\frac{1}{|H|} \sum_{h\in H}\chi_j(h) =1$. However it is important also to determine all the other (classes of) isotropy subgroups because (i) branches in higher dimensional fixed point subspaces should not be a priori excluded and (ii) this is useful for the study of the stability conditions and local dynamics. In \cite{chossat-faye-etal:11} all subgroups of $\g$ have been determined (using GAP), a necessary step to compute the istropy subgroups using formula \ref{eq:formuletrace}. The list is cumbersome, so we simply list the isotropy subgroups for the representations $\chi_{11}$ to $\chi_{13}$ and refer to \cite{chossat-faye-etal:11} for the details. We shall use the notation $\widetilde{\sigma}=\rho^2\sigma\rho^{-2}$, which is the rotation by $\pi$ around point $\tilde S$ in Figure \ref{fig:tesselation}.
\begin{defi} \label{def:maximalisotropy}
\begin{eqnarray*}
\widetilde{\mC}_{2\kappa} &=& \langle\sigma,\kappa\rangle = \{Id, \sigma,\kappa,\kappa''\} \\
\widetilde{\mC}'_{2\kappa} &=&  \langle\widetilde\sigma,\kappa\rangle = \{Id, \tilde\sigma,\kappa,-\rho^2\kappa''\rho^{-2}\} \\
\widetilde{\mC}_{3\kappa'} &=&  \langle\epsilon,\kappa'\rangle = \{Id, \epsilon,\epsilon^2,\kappa',\epsilon\kappa'\epsilon^2,\epsilon^2\kappa'\epsilon\} \\
\widetilde{\mD}_3 &=&  \langle\widetilde\sigma,\epsilon\rangle = \{Id, \epsilon,\epsilon^2,\widetilde\sigma,\epsilon\widetilde\sigma\epsilon^2,\epsilon^2\widetilde\sigma\epsilon\}
\end{eqnarray*}
\end{defi}
The meaning of these groups for the octagonal lattice is not straightforward and will be explained in the next subsection.

\noindent {\em 1. The $\chi_{12}$ and $\chi_{13}$ cases.}

The (conjugacy classes of) isotropy subgroups are ordered by set inclusion and we call {\em lattice of isotropy types} the corresponding graph.

\begin{lem}
The lattices of isotropy types for the representations $\chi_{12}$ and $\chi_{13}$ are identical and are shown in Figure \ref{fig:isotropies_chi12}. The numbers in parentheses indicate the dimension of corresponding fixed-point subspaces.
\begin{figure}[h]
 \centering
 \includegraphics[width=7cm]{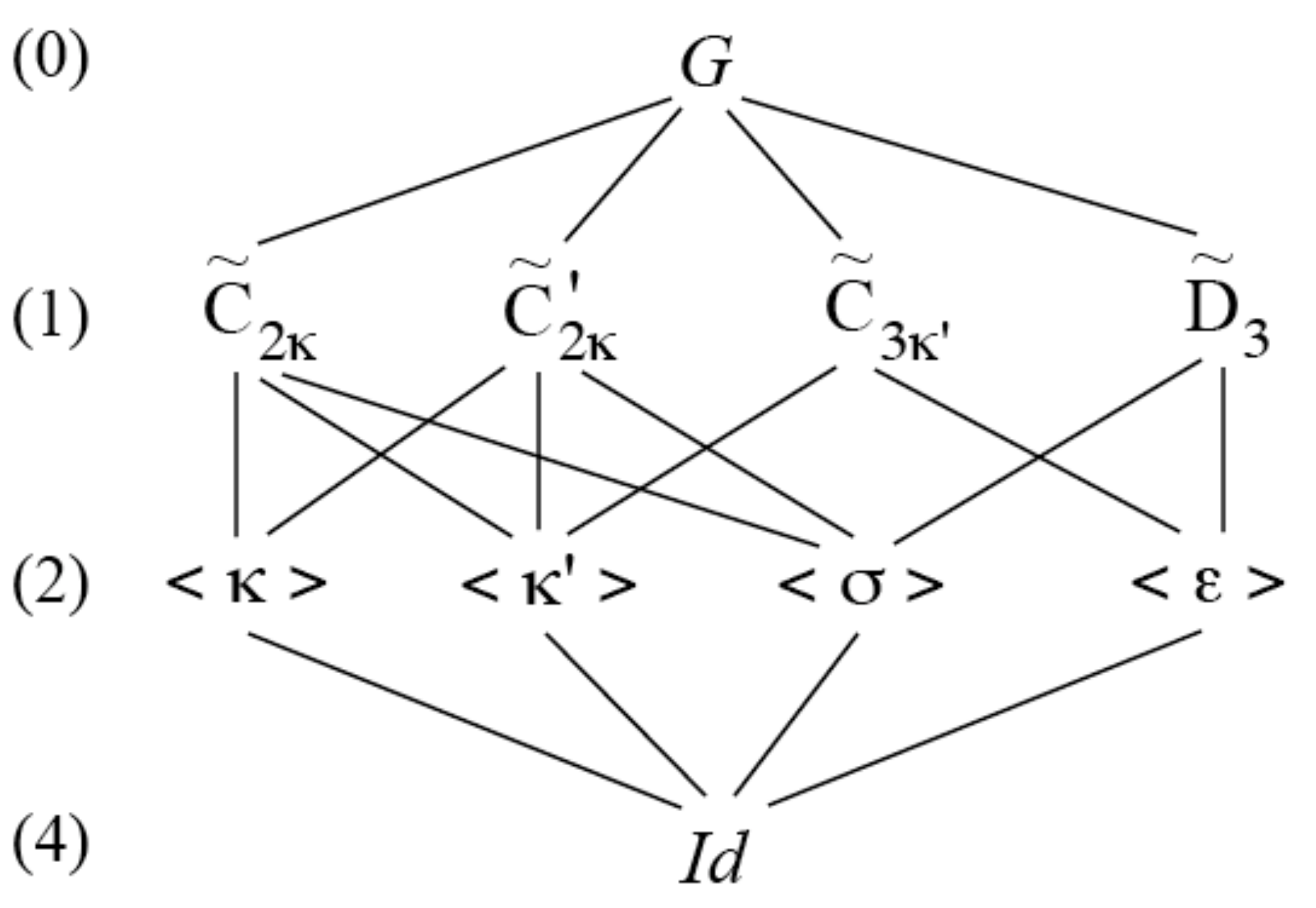}
 \caption{The lattice of isotropy types for the representations $\chi_{12}$ and $\chi_{13}$.}
 \label{fig:isotropies_chi12}
\end{figure}
\end{lem}
\begin{rmk}
(i) The difference between $\chi_{12}$ and $\chi_{13}$ is quite subtle as we see from the character table and their geometrical propreties are similar as well as their bifurcation diagrams for Equation (\ref{eq:bifurcation}). \\
(ii) We can see on the lattice of isotropy types that each fixed-point plane contains at least 2 or 3 axes of symmetry. The exact number is 4. For example $\text{Fix}(\langle\sigma\rangle)$ contains one copy of $\text{Fix}(\widetilde{\mC}_{2\kappa})$, one copy of $\text{Fix}(\widetilde{\mC}'_{2\kappa})$ and two copies of $\text{Fix}(\widetilde{\mD}_3)$.
\end{rmk}
Thanks to Theorem \ref{thm:equivariant branching lemma} each axis of symmetry contains bifurcated solutions, and moreover the branches are pitchfork because one can easily check that for each such isotropy subgroup $H$ (as defined in \ref{def:maximalisotropy}), $N(H)/H\simeq\Z_2$. In order to study the stability of these solutions and to look for other equilibria we need to know the Taylor expansion of the bifurcation equation \eqref{eq:bifurcation} up to some order large enough to remove degeneracies. It turns out that in this case cubic order in $X$ is enough. However here again the calculations are cumbersome and we refer the reader to \cite{faye-chossat:11} for details. \\
We first need to choose suitable coordinates for $X\in\R^4$ and we do so for the representation $\chi_{12}$ (we know similar results hold for $\chi_{13}$. Let $(z_1,\bar z_1,z_2,\bar z_2)$ be the coordinates which diagonalize the matrix of the 8-fold symmetry $\rho$ in $\R^4$. The equation \eqref{eq:bifurcation} truncated at order 3 in these coordinates reads
\bqq
\dot z_1=\left[\lambda + a(|z_1 |^2+|z_2 |^2)\right] z_1+ b\left[ \sqrt{3} \left(3 z_1^2 +\bar z_2^2\right)\bar z_1-{\mathbf{i}} \left(z_2^2 +3 \bar z_1^2\right) z_2  \right]
\label{eq:f1cubic}
\eqq
\bqq
\dot z_2=\left[\lambda + a(|z_1 |^2+|z_2 |^2)\right] z_2+b\left[\sqrt{3} \left(3 z_2^2 +\bar z_1 ^2\right)\bar z_2+{\mathbf{i}} \left(z_1^2 +3\bar z_2^2\right)z_1 \right]
\label{eq:f2cubic}
\eqq
and their complex conjugates, where $a$, $b$ are real coefficients. These equations can be written as a gradient system. \\
We can now state the main result for representations $\chi_{12}$ and $\chi_{13}$.
\begin{thm}\label{thm:bifchi12}
 Provided that $(a,b)\in \mathcal{P}=\{ (a,b)\in\R^2 ~|~  3a+2b\sqrt{3}<0 \text{ and } 3a+10b\sqrt{3}<0\}$, the following holds for Equations \eqref{eq:f1cubic}-\eqref{eq:f2cubic}.
\begin{itemize}
\item[(i)]  No solution bifurcates in the principal strata of the planes of symmetry.
\item[(ii)] The branches of equilibria with maximal isotropy are pitchfork and supercritical.
\item[(iii)] If $b>0$ (resp. $b<0$), the equilibria with isotropy type $\widetilde{\mC}_{3\kappa'}$ (resp. $\widetilde{\mD}_3$) are stable in $\R^4$. Branches with isotropy $\widetilde{\mC}_{2\kappa}$ and $\widetilde{\mC}'_{2\kappa}$ are always saddles.
\end{itemize} 
\end{thm}
We didn't prove the non existence of bifurcated equilibria with trivial isotropy, however there is a strong evidence that such solutions are generically forbidden for equations \eqref{eq:f1cubic} and \eqref{eq:f2cubic}. \\
We now turn to the $\chi_{11}$ case.

2. {\em The $\chi_{11}$ case.}

\begin{lem}
The lattice of isotropy types for the representation $\chi_{11}$ is shown in Figure \ref{fig:isotropies_chi11}. The numbers in parentheses indicate the dimension of corresponding fixed-point subspaces.
\begin{figure}[h]
 \centering
 \includegraphics[width=5cm]{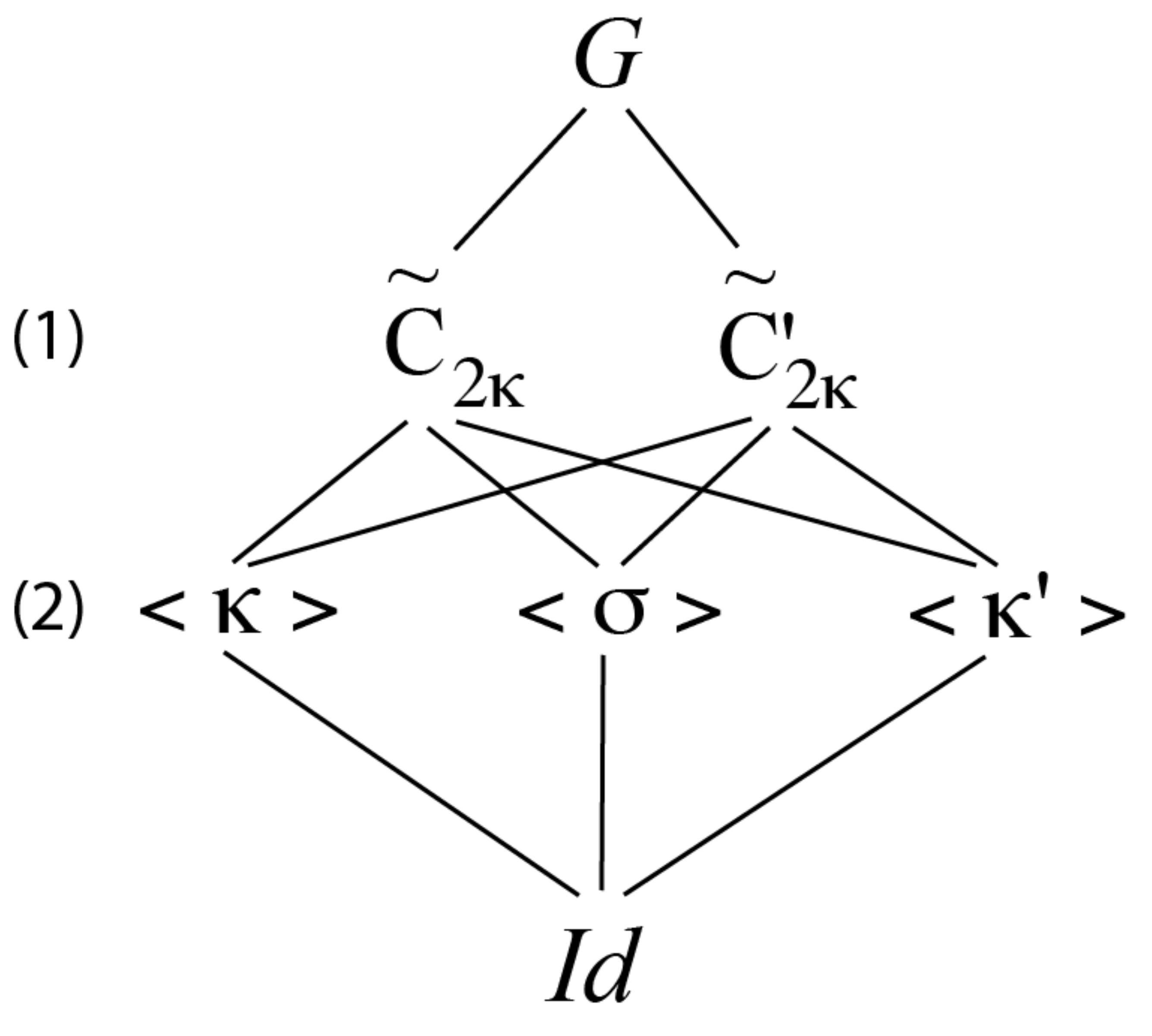}
 \caption{The lattice of isotropy types for the representation $\chi_{11}$.}
 \label{fig:isotropies_chi11}
\end{figure}
\end{lem}
We see that this representation is somewhat different from the other 4 dimensional representations of $\g$, and indeed it leads to quite different bifurcation diagrams. In particular il allows for the bifurcation of equilibria with isotropy $\sigma$ or $\kappa'$ under generic conditions on the coefficients of terms of order 5 in the Taylor expansion of the bifurcation equation. Indeed there is no term with even order and the only term of order 3 is the "radial" one: $\|X\|^2X$, which can't provide pattern selection. The next order is 5 and one can show that there are 4 independant $\g$-equivariant terms  at this order. Details are provided in \cite{faye-chossat:11}. Among these terms, three are gradients while one is non gradient. It was shown in \cite{faye-chossat:11} that when the latter is non zero, a robust heteroclinic cycle can bifurcate from the trivial state. However this phenomenon, which is interesting in the context of non gradient systems, does not occur for the Swift-Hohenberg equation.

As before let $(z_1,\bar z_1,z_2,\bar z_2)$ be the coordinates which diagonalize the matrix of the 8-fold symmetry $\rho$ in $\R^4$. The equation \eqref{eq:bifurcation} truncated at order 5 and with no non gradient terms reads in these coordinates
\bqq
\dot z_1=\left[\lambda + A(|z_1 |^2+|z_2 |^2)\right] z_1+ \frac{\partial}{\partial \bar z_1}Q(z_1,\bar z_1,z_2,\bar z_2)
\label{eq:f1quintic}
\eqq
\bqq
\dot z_2=\left[\lambda + A(|z_1 |^2+|z_2 |^2)\right] z_2+\frac{\partial}{\partial \bar z_2}Q(z_1,\bar z_1,z_2,\bar z_2)
\label{eq:f2quintic}
\eqq 
where $Q$ is defined as follows (C.C. means "complex conjugate"):
$$
Q(z_1,\bar z_1,z_2,\bar z_2)=a(|z_1 |^2+|z_2 |^2)^3 + b(z_1^4\bar z_1\bar z_2 + 2z_1\bar z_1^2 z_2^3 - \bar z_1 z_2^4\bar z_2 - 2z_1^3 z_2\bar z_2^2 + C.C.) + d(z_1 z_2^5-z_1^5 z_2 + C.C.).
$$
Compared with the equations in \cite{faye-chossat:11}, we see that a term with coefficient $c$ in the latter does not appear here. This is due to the gradient structure and the correspondance between the two equations is obtained by setting $c=-2b$ in \cite{faye-chossat:11}.

In these coordinates the equations of the fixed-point subspaces are 
\begin{itemize}
\item[-] $\text{Fix}(\sigma)$:  $z_2=i(1+\sqrt{2})z_1$,
\item[-] $ \text{Fix}(\kappa)$: $z_1=i \bar z_1 \text{ and } z_2 =-i \bar z_2$,
\item[-] $ \text{Fix}(\kappa')$: $z_1=\frac{\sqrt{2}}{2}(i -1)\bar z_1 \text{ and } z_2=\frac{\sqrt{2}}{2}(i +1)\bar z_2$,
\end{itemize}
and the intersections of these planes are the fixed point axes for the isotropy subgroups $\widetilde{\mC}_{2\kappa}$ and $\widetilde{\mC}'_{2\kappa}$. Using these informations one can show the following
\begin{thm}
Provided that $A<0$ in (\ref{eq:f1quintic})-(\ref{eq:f2quintic}), the branches of equilibria with isotropy $\widetilde{\mC}_{2\kappa}$ and $\widetilde{\mC}'_{2\kappa}$ are pitchfork and supercritical. Moreover: 
\begin{itemize}
\item[(i)] $ \text{Fix}(\kappa)$ contains two copies of each type of axes of symmetry and there is generically no other solutions bifurcating in this plane.
\item[(ii)] $\text{Fix}(\sigma)$ and $ \text{Fix}(\kappa')$ contain each only one copy of both types axes of symmetry. Moreover in either plane solutions can bifurcate in the principal stratum, depending on the values of the coefficients $b$ and $d$ in the bifurcation equations (\ref{eq:f1quintic})-(\ref{eq:f2quintic}). The bifurcation and stability diagrams in terms of $b$ and $d$ are shown in Figure \ref{fig:diag23}. The regions I-II-III which are referred to in this figure are defined in Figure \ref{fig:dom}. Region IV is similar to II (bifurcation of solutions in the principal stratum).
\end{itemize}
\end{thm}
\begin{figure}[htp]
 \centering
\subfigure[In $\text{Fix}(\sigma)$, solutions with isotropy $\sigma$ exist in the regions $II$ and $IV$. We have set $x=-b$ and $y=d$.]{\label{fig:domP2}
\includegraphics[width=6cm]{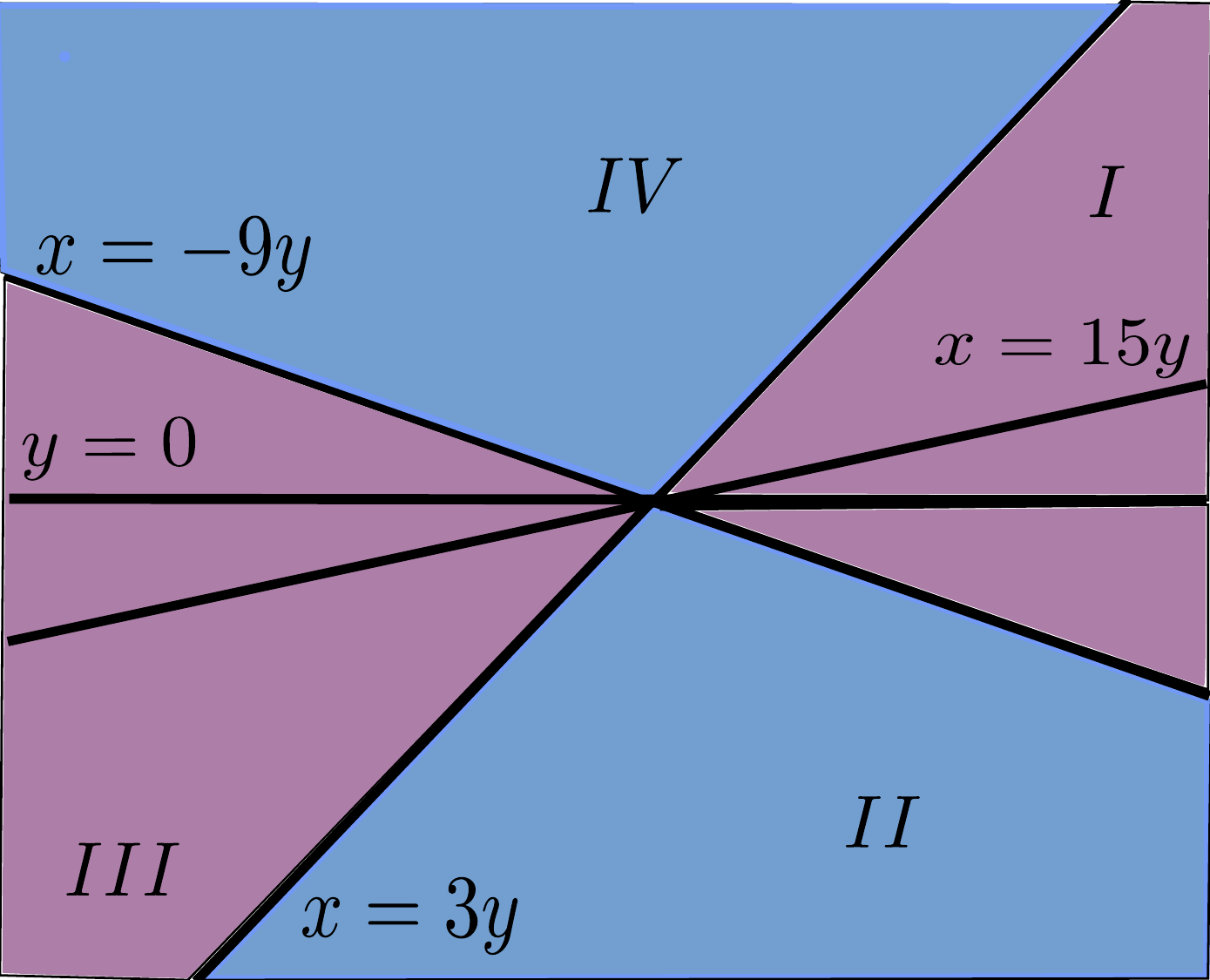}}
\hspace{.3in}
\subfigure[In $\text{Fix}(\kappa')$, solutions with isotropy $\kappa'$ exist in the regions $II$ and $IV$. We have set $x=-2b$ and $y=b+d$.]{\label{fig:domP3}
\includegraphics[width=6cm]{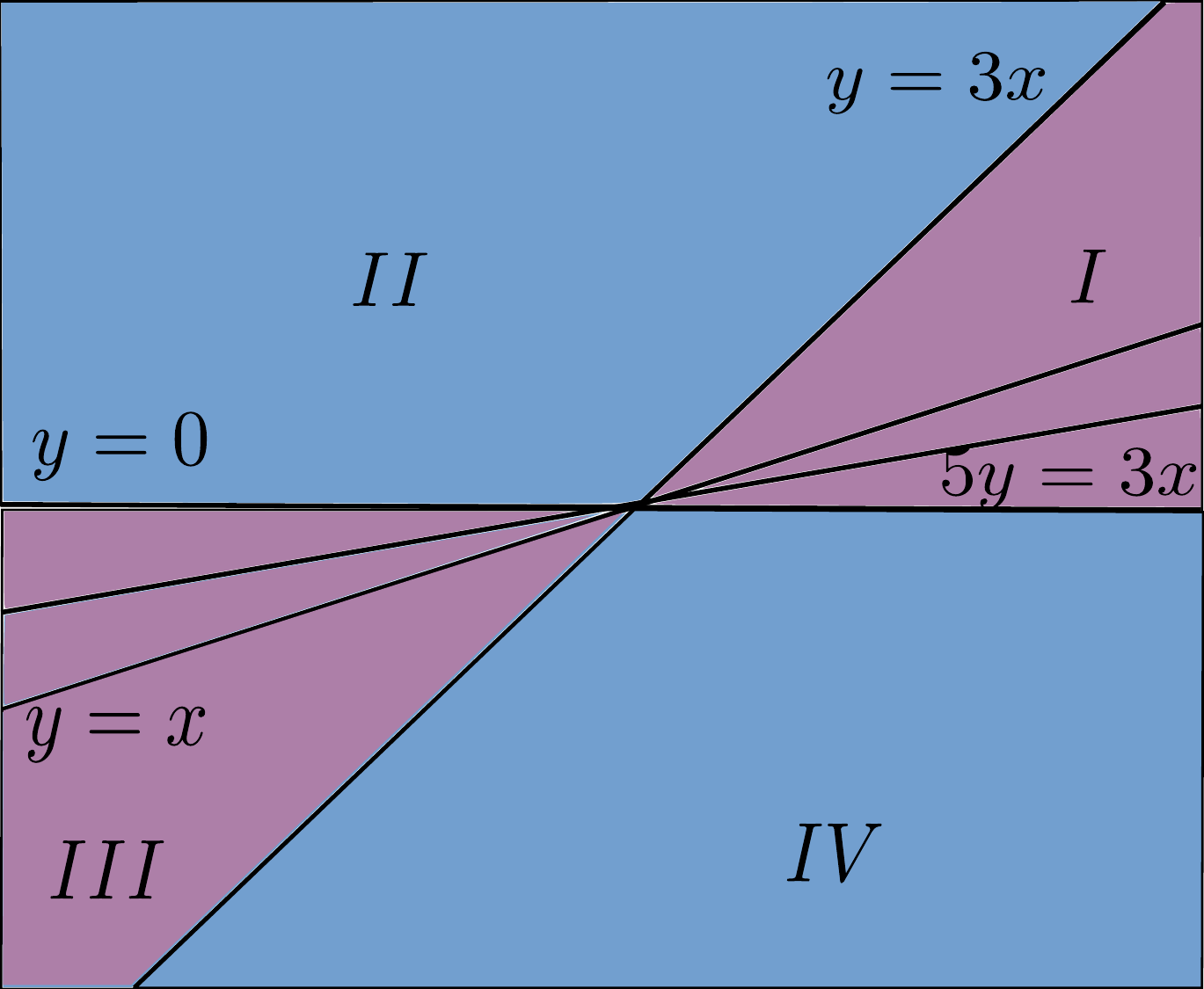}}
\caption{Conditions of existence of solutions with isotropy $\sigma$ and $\kappa'$ (in blue).}
\label{fig:dom}
\end{figure}

\begin{figure}[htp]
 \centering
\includegraphics[width=10cm]{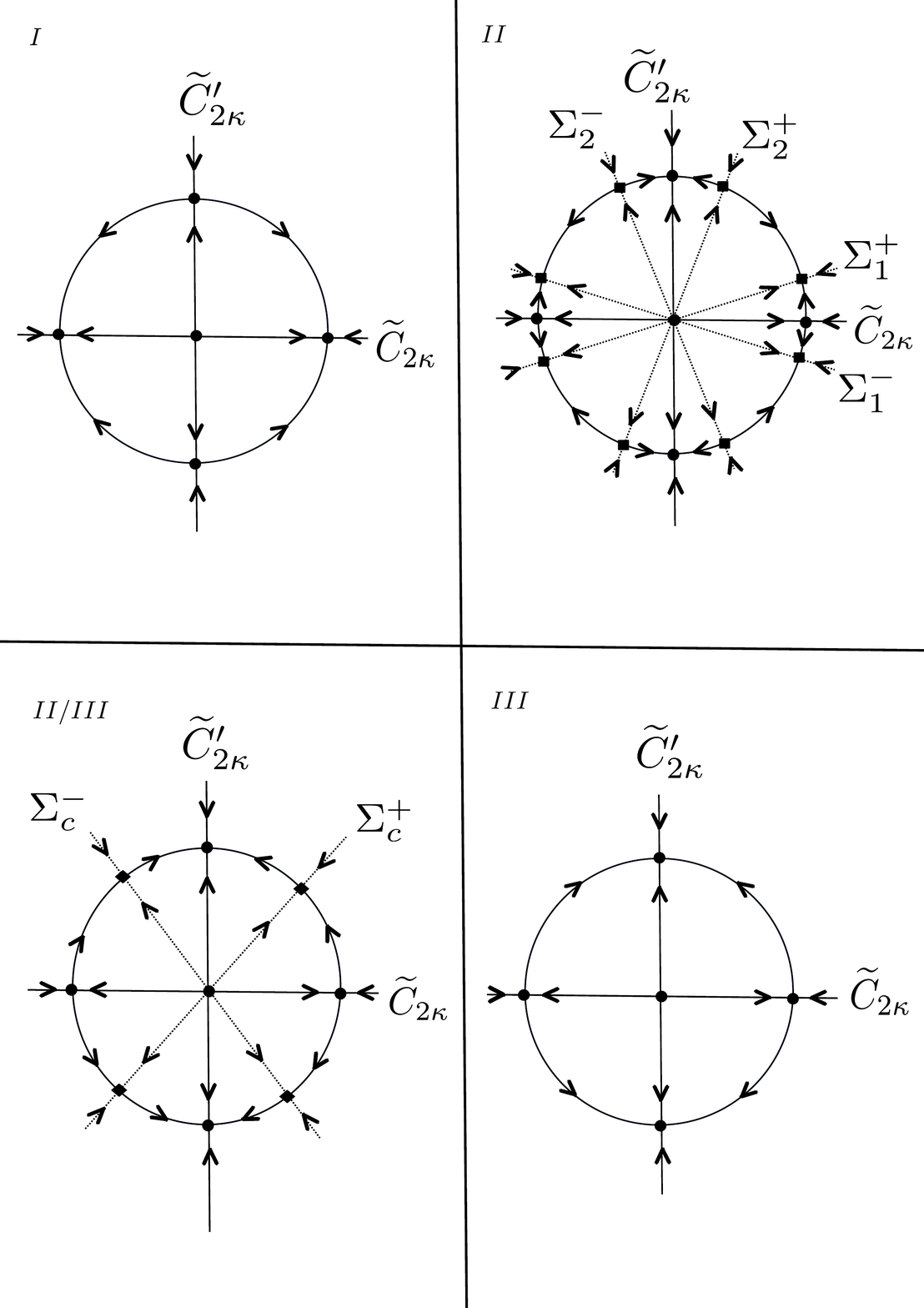}
\caption{Changes of phase diagram in $\text{Fix}(\sigma)$ and $\text{Fix}(\kappa')$ as the coefficients $(b,d)$ pass from regions $I-II-III$ in Figure \ref{fig:domP2} (for $\text{Fix}(\sigma)$) and in Figure \ref{fig:domP3} (for $\text{Fix}(\kappa')$). $\Sigma_1^{\pm}$ and $\Sigma_2^{\pm}$ indicate solutions with isotropy $\sigma$ or $\kappa'$. The case $II/III$ corresponds to coefficient values at the boundary between region $II$ and $III$ ( saddle-node bifurcation of equilibria with isotropy $\sigma$ or $\kappa'$).}
\label{fig:diag23}
\end{figure}

\subsubsection{Computation of H-planforms}

It follows from the definition that H-planforms are eigenfunctions of the Laplace-Beltrami operator in $\D$ which satisfy certain isotropy conditions: (i) being invariant under a lattice group $\Gamma$ and (ii) being invariant under the action of an isotropy subgroup of the symmetry group of the fundamental domain $\D/\Gamma$ (mod $\Gamma$). Therefore in order to exhibit H-planforms, we need first to compute eigenvalues and eigenfunctions of $\Delta_\D$ in $\D$, and second to find those eigenfunctions which satisfy the desired isotropy conditions. In this subsection, we use the notations $\mu$ and $\Psi$ for eigenvalue and eigenfunction of $\Delta_\D$:
\bqs
-\Delta_\D \Psi(z)=\mu \Psi(z),\quad \forall z\in \Oct.
\eqs

Over the past decades, computing the eigenmodes of the Laplace-Beltrami operator on compact manifolds has received much interest from physicists. The main applications are certainly in quantum chaos \cite{balazs-voros:86,aurich-steiner:89,aurich-steiner:93,schmit:91,cornish-turok:98} and in cosmology \cite{inoue:99,cornish-spergel:99,lehoucq-weeks-etal:02}.

In order to find these H-planforms, we use the finite-element method with periodic boundary conditions. This choice is dictated by the fact that this method will allow us to compute all the first $n$ eigenmodes and among all these we will identify those which correspond to a given isotropy group.

\begin{figure}[htp]
\centering
\includegraphics[width=0.6\textwidth]{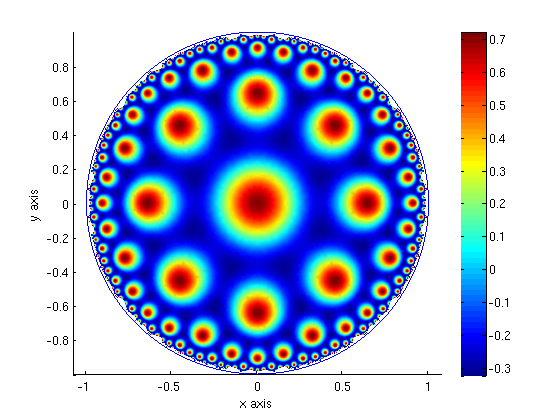}
\caption{Plot of the eigenfunction of the Laplace-Beltrami operator in the Poincar\'e disk $\D$ with $\g$-symmetry with eigenvalue $\mu=23.0790$.}
\label{fig:chi1}
\end{figure}

As there exists an extensive literature on the finite element methods (see for an overview \cite{ciarlet-lions:91,allaire:05}) and as numerical analysis is not the main goal of this review, we do not detail the method itself but rather focus on the way to actually compute the eigenmodes of the Laplace-Beltrami operator. We mesh the full octagon with 3641 nodes in such a way that the resulting mesh enjoys a $\mD_8$-symmetry. We implement, in the finite element method of order 1, the periodic boundary conditions of the eigenproblem and obtain the first 100 eigenvalues of the octagon. Our results, as reported in \cite{chossat-faye-etal:11}, are in agreement with those of Aurich and Steiner \cite{aurich-steiner:89}.

\begin{figure}[htp]
\centering
\subfigure[H-planform $\Psi_3$ with symmetry $\mD_{8}$.]{
\label{fig:hpd3_3_disk}
\includegraphics[width=0.45\textwidth]{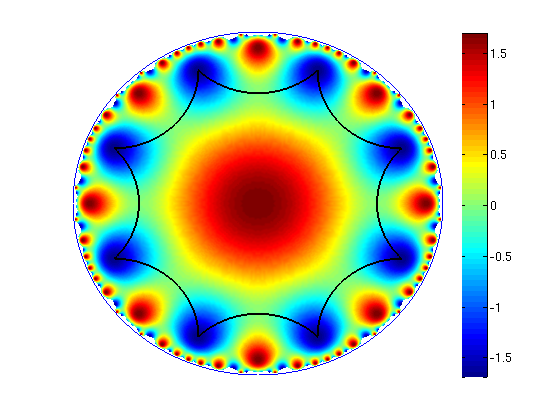}}\\
\subfigure[H-planform $\Psi_1$ with symmetry $\gamma_1\cdot \mD_8$.]{
\label{fig:hpd3_1_disk}
\includegraphics[width=0.45\textwidth]{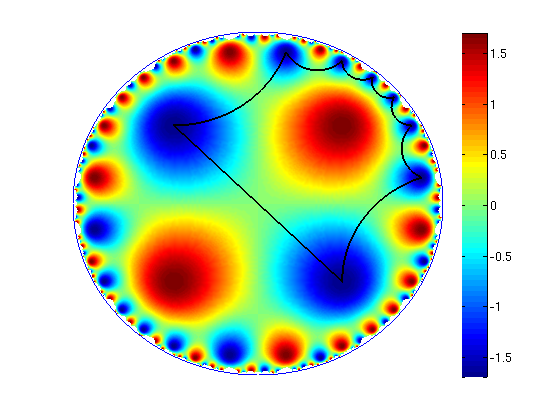}}
\hspace{.3in}
\subfigure[H-planform $\Psi_2$ with symmetry $\gamma_2\cdot \mD_8$.]{
\label{fig:hpd3_2_disk}
\includegraphics[width=0.45\textwidth]{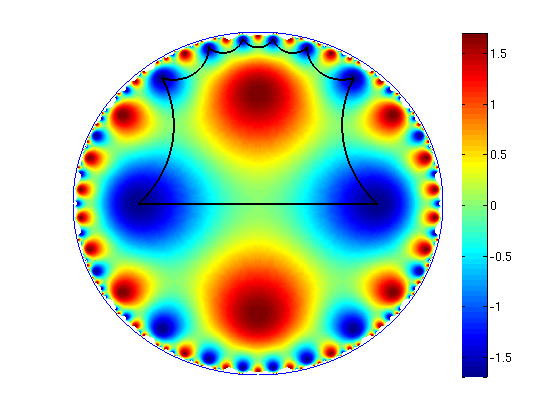}}
\caption{Plot of the eigenfunctions of the Laplace-Beltrami operator in the Poincar\'e disk $\D$ associated to the lowest non-negative eigenvalue $\mu=3.8432$ corresponding to the irreducible representation $\chi_8$. In (a) we also plot the octagon (black line) and in (b),(c) its image by $\gamma_1,\gamma_2$ respectively.}
\label{fig:hpd3_disk}
\end{figure}

We plot in Figure \ref{fig:chi1}, the first eigenfunctions of the Laplcae-Beltrami operator with full octagonal symmetry $\g$ with non zero eigenvalue. It is associated to representation $\chi_1$. Next, in Figure \ref{fig:hpd3_disk}, we plot the corresponding eigenfunctions of the Laplace-Beltrami operator associated to the lowest non-negative eigenvalue $\mu=3.8432$ with multiplicity 3. We identify each solution by its symmetry group. It is clear that \ref{fig:hpd3_1_disk} and \ref{fig:hpd3_2_disk} can be obtained from \ref{fig:hpd3_3_disk} by hyperbolic transformations. From the definitions of $g_0$ in \eqref{eq:boostg}, we see that $g_0=a_{r_0}$ with $r_0=\ln \left(1+\sqrt{2}+\sqrt{2+\sqrt{2}} \right)$. If we define $\gamma_k\in\g$ by:
\bqq
\label{eq:gmk}
\gamma_k=\text{rot}_{k\pi/4}a_{r_0/2}\text{rot}_{-k\pi/4}
\eqq
then Figure \ref{fig:hpd3_1_disk} (resp. \ref{fig:hpd3_2_disk}) is obtained from \ref{fig:hpd3_3_disk}  by applying $\gamma_1$ (resp. $\gamma_2$). A gallery of eigenfunctions can be found in \cite{chossat-faye-etal:11,faye-chossat:11}.

\subsubsection{Case study: computation of the bifurcation equations for irreducible representation $\chi_8$}

For most of the irreducible representations, it is practically impossible to compute the reduced equation \eqref{eq:bifurcation} given by the center manifold theorem. The main reason being that we only know numerically the eigenfunctions of the Laplace-Beltrami operator. It turns out that we have been able to successfully conduct this computation in the case of 3 dimensional irreducible representation $\chi_8$, when the evolution equation is a neural field equation \cite{faye-chossat:12}. In that particular study, the choice of $\chi_8$ was dictated by the simple interpretation of H-planforms given in Figure \ref{fig:hpd3_disk} in terms of preferred orientations within a hypercolumn of the visual cortex and allows the computation of geometric visual hallucinations across the cortex \cite{faye-chossat:12}. As already explained in the previous section, the lowest non-negative eigenvalue $\mu=3.8432$ is of multiplicity 3 and associated to the irreducible representation $\chi_8$. We then select the parameter $\alpha$ in equation \eqref{eq:sh} to be equal to $\alpha_c$ such that for $\mu_c=\mu=3.8432$ we have $\alpha_c^2=\mu_c$. Restricting to the class of $\Gamma$-periodic functions, this ensures that the first modes which bifurcate from $\lambda=0$ are associated to the irreducible representation $\chi_8$, all other modes being damped to zero (see discussion in section \ref{subsub:linear}). We then rewrite the Swift-Hohenberg equation \eqref{eq:sh} with $\alpha=\alpha_c$
\bqq
\label{eq:SHc}
u_t=-(\alpha_c^2+\Delta_\D)^2u+\lambda u+ \mathcal{N}(u).
\eqq

\begin{prop}\label{prop:fps}
 For the three dimensional irreducible representation $\chi_8$ of $\g$, the isotropy subgroups with one dimensional fixed point subspace are the following:
\begin{align*}
 \mD_8&=\langle \rho, \kappa \rangle \\
 \widetilde{\mC}_{6\kappa'}&= \langle -\epsilon, \kappa'\rangle \\
 \widetilde{\mD}_{2\kappa}&= \langle -Id,\sigma,\kappa \rangle.
\end{align*}
\end{prop}
In Figure \ref{fig:mis}, we represented the different axes of symmetry of the cube with isotropy subgroups given in proposition \ref{prop:fps}. Note that planforms in Figure \ref{fig:hpd3_disk} correspond to the three coordinate axes of the cube.

\begin{figure}[htp]
\centering
\includegraphics[width=0.7\textwidth]{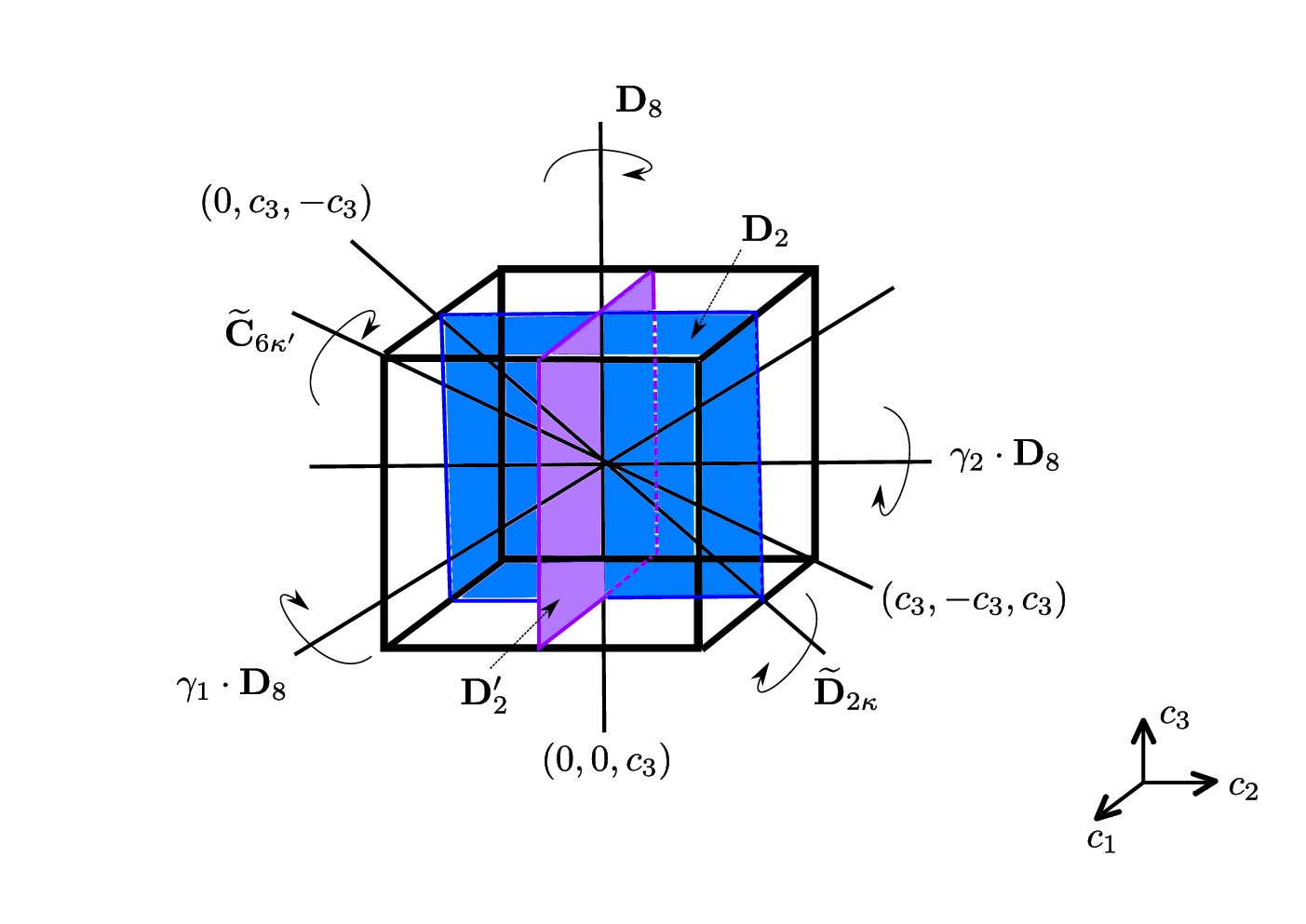}
\caption{Maximal isotropy subgroups $\mD_8, \widetilde{\mathbf{C}}_{6\kappa'}$ and $\widetilde{\mD}_{2\kappa}$ of $\mathbb{O} \ltimes \mathbb{Z}_2$. The axes $\gamma_1\cdot\mD_8$ and $\gamma_2\cdot\mD_8$ are copies of $\mD_8$ by the elements $\gamma_1,\gamma_2\in\g$ (see \eqref{eq:gmk}). The plane $(0,c_2,c_3)$ (resp. $(c_1,0,c_3)$) has symmetry $\mD_2$ (resp. $\mD_2'$).}
\label{fig:mis}
\end{figure}

If we denote $\Psi_1$ the H-planform in Figure \ref{fig:hpd3_1_disk}, $\Psi_2$ the H-planform in Figure \ref{fig:hpd3_2_disk} and $\Psi_3$ the H-planform corresponding to the symmetry group $\mD_8$ in Figure \ref{fig:hpd3_3_disk}, then $(\Psi_1,\Psi_2,\Psi_3)$ is a basis for the irreducible representation $\chi_8$. This can be easily seen through the identification of each H-planform to the three coordinate axes of the cube in Figure \ref{fig:mis}. Note that we have normalized planforms such that:
\bqs
\langle \Psi_i,\Psi_j \rangle =\frac{1}{4\pi}\int_\Oct\Psi_i (z) \Psi_j (z) \text{dm}(z)=\delta_{i,j}.
\eqs

We rewrite equation \eqref{eq:SHc} as a dynamical system in the infinite-dimensional phase space $\mathcal{X}=\text{L}^2_{per}(\Oct)$ consisting $\Gamma$-periodic functions:
\bqs
u_t= \mL_\lambda u +\mathcal{N}(u)
\eqs
with 
\bqs
 \mL_\lambda = -(\alpha_c^2+\Delta_\D)^2+\lambda.
\eqs
The linear part $\mL_\lambda$ is a closed linear operator with dense domain $\mathcal{Y}=\text{H}^4_{per}(\Oct)$ where
\bqs
\text{H}^4_{per}(\Oct)=\{u \in \text{H}^4_{loc}(\D) ~|~ u(\gamma \cdot z)=u(z), \forall (\gamma,z)\in \g\times \D \}
\eqs
and the nonlinear map $\mathcal{N}:\mathcal{Y}\rightarrow \mathcal{Y}$ is $\mathcal{C}^p$ for all positive integer $p$. The linear operator $\mL_0$ is closed in $\mathcal{X}$, with dense and compactly embedded domain $\mathcal{Y}$. Then $\mL_0$ has a compact resolvent, and its spectrum $\sigma(\mL_0)$ is purely point spectrum, only. Moreover we have $\sigma(\mL_0)\cap i\R=\{ 0 \}$. The eigenvalue $0$ is geometrically and algebraically triple (consequence of the absolute irreducibility of the representation of $\g$ in that eigenspace). We denote $V$ the corresponding three-dimensional eigenspace. A direct calculation also shows that there exists $\omega_0>0$ and $C>0$ such that for all $|\omega|>\omega_0$, we have the resolvent estimate:
\bqs
\|(i\omega -\mL_0)^{-1}\|_{\mathcal{L}(\mathcal{X})}\leq \frac{C}{|\omega|}.
\eqs
We can apply the equivariant center manifold reduction introduced in \ref{subsection:CM}. So there exists a map $\Phi:V\times \R \rightarrow V^\perp$ such that all solutions $u$ that lie on the center manifold can be written
\bqs
u=x_1\Psi_1+x_2\Psi_2+x_3\Psi_3+\Phi(x_1,x_2,x_3,\lambda),
\eqs
and $(c_1,c_2,c_3)$ satisfies the following reduced system:
\bqq
\label{eq:NF}
\left\{
\begin{array}{lcl}
 \dfrac{dx_1}{dt} &=&\lambda x_1+\left[ a (x_2^2+x_3^2)+b x_1^2 \right]x_1+\text{h.o.t.}\\
 \dfrac{dx_2}{dt} &=&\lambda x_2+\left[ a (x_1^2+x_3^2)+b x_2^2 \right]x_2+\text{h.o.t.}\\
 \dfrac{dx_3}{dt} &=&\lambda x_3+\left[ a (x_1^2+x_2^2)+b x_3^2 \right]x_3+\text{h.o.t.}
\end{array}
\right.
\eqq
We refer to \cite{melbourne:86} for the computation of the normal form \eqref{eq:NF} and  for a review on bifurcation problems with octahedral symmetry.

Taylor expanding the map $\Phi$:
\bqs
\Phi(x_1,x_2,x_3,\lambda)=\sum_{0\leq r+s+l+m\leq 3}x_1^rx_2^sx_3^l\lambda^m\Phi_{rslm}+\text{h.o.t}
\eqs
and denoting $\cR$:
\bqs
\cR(u,\lambda)=\lambda u+\mathcal{N}(u)=\cR_{11}(u,\lambda)+\cR_{20}(u,u)+\cR_{30}(u,u,u)
\eqs
with
\begin{align*}
 \cR_{11}(u,\lambda)&=\lambda u \\
 \cR_{20}(u,v)&=\nu uv \\
 \cR_{300}(u,v,w)&=-\eta uvw ,
\end{align*}
we obtain the following system of equations:
\begin{align}
\label{eq:sys_oct}
 0&= -\mL_0 \Phi_{0020}-\cR_{20}(\Psi_3,\Psi_3) \nonumber \\
 0&= -\mL_0 \Phi_{1010}-2\cR_{20}(\Psi_1,\Psi_3) \nonumber \\
 a &= \langle 2\cR_{20}(\Phi_{0020},\Psi_1)+2\cR_{20}(\Phi_{1010},\Psi_3)+3\cR_{30}(\Psi_1,\Psi_3,\Psi_3),\Psi_1\rangle \nonumber \\
 b &= \langle 2\cR_{20}(\Psi_3,\Phi_{0020})+\cR_{30}(\Psi_3,\Psi_3,\Psi_3),\Psi_3\rangle.
\end{align}
Here $\langle\cdot,\cdot\rangle$ is the scalar product on $\text{L}^2_{per}(\Oct)$.

In order to solve the two first equations of the previous system, we need to know if the functions $\Psi_3(z)\Psi_3(z)$ and $\Psi_1(z)\Psi_3(z)$ can be expressed as a linear combination of eigenfunctions of the Laplace-Beltrami operator on $\Oct$. In general, it is very difficult to obtain these expressions because the eigenfunctions are only known numerically and one needs the computation of the associated Clebsch-Gordan coefficients. It turns out that in our case we have been able to conjecture and numerically verify the following relations:
\begin{align*}
 \Psi_1(z)\Psi_3(z) &= \frac{1}{\sqrt{3}}\Psi_{\chi_{10}}^{\mD_{2\kappa}'}(z) \\
 \Psi_3^2(z)&= \frac{6}{5}\Psi_{\chi_6}^{\widetilde{\mD}_{8\kappa}}(z)+1 
\end{align*}
where the corresponding isotropy subgroups are given by:
\bqs
\mD_{2\kappa}'=<-Id,\rho^2 \kappa,\rho^2\sigma> \text{ and } \widetilde{\mD}_{8\kappa}= <\rho, \rho^2\sigma\rho^{-2}, \kappa >.
\eqs
The notation $\Psi_{\chi_{10}}^{\mD_{2\kappa}'}(z)$ means that the product $ \Psi_1(z)\Psi_3(z)$ is an eigenfunction of the Laplace-Beltrami operator associated to the irreducible representation $\chi_{10}$ with isotropy subgroup $\mD_{2\kappa}'$. Similarly the notation $\Psi_{\chi_6}^{\widetilde{\mD}_{8\kappa}}(z)$ stands for an eigenfunction of the Laplace-Beltrami operator associated to the irreducible representation $\chi_{6}$ with isotropy subgroup $\mD_{8\kappa}$. 
Furthermore we have normalized $\Psi_{\chi_{10}}^{\mD_{2\kappa}'}$ and $\Psi_{\chi_6}^{\widetilde{\mD}_{8\kappa}}$ such that:
\bqs
\langle \Psi_{\chi_{10}}^{\mD_{2\kappa}'},\Psi_{\chi_{10}}^{\mD_{2\kappa}'}\rangle=\langle \Psi_{\chi_6}^{\widetilde{\mD}_{8\kappa}},\Psi_{\chi_6}^{\widetilde{\mD}_{8\kappa}}\rangle=1.
\eqs
In Figure \ref{fig:Psi}, we plot the eigenfunctions $\Psi_{\chi_{10}}^{\mD_{2\kappa}'}$ and $\Psi_{\chi_6}^{\widetilde{\mD}_{8\kappa}}$ of the Laplace-Beltrami operator in the octagon $\Oct$. One interesting remark is that the product $\Psi_1 \Psi_3$ corresponding to the three dimensional irreducible representation $\chi_8$ produces an eigenfunction associated to another three dimensional irreducible representation: $\chi_{10}$ whereas $\Psi_3^2$ is the linear combination of the constant function which has $\g$ as isotropy subgroup and thus corresponds to $\chi_1$ and the eigenfunction $\Psi_{\chi_6}^{\widetilde{\mD}_{8\kappa}}$ which is associated to two dimensional irreducible representation $\chi_6$. We denote $\mu_{\chi_{10}}^{\mD_{2\kappa}'}$ and $\mu_{\chi_6}^{\widetilde{\mD}_{8\kappa}}$ the corresponding eigenvalues:
\bqs
-\Delta_\D \Psi_{\chi_{10}}^{\mD_{2\kappa}'}= \mu_{\chi_{10}}^{\mD_{2\kappa}'} \Psi_{\chi_{10}}^{\mD_{2\kappa}'} \text{ and } -\Delta_\D \Psi_{\chi_6}^{\widetilde{\mD}_{8\kappa}} = \mu_{\chi_6}^{\widetilde{\mD}_{8\kappa}} \Psi_{\chi_6}^{\widetilde{\mD}_{8\kappa}}.
\eqs

\begin{figure}[htp]
\centering
\subfigure[Plot of $\Psi_{\chi_{10}}^{\mD_{2\kappa}'}$.]{
\label{fig:Psi1Psi3}
\includegraphics[width=0.45\textwidth]{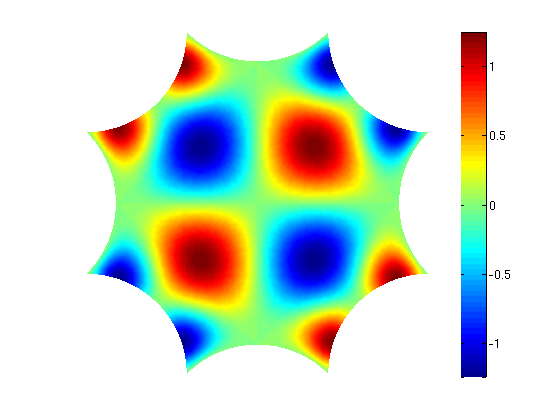}}
\hspace{.3in}
\subfigure[Plot of $\Psi_{\chi_6}^{\widetilde{\mD}_{8\kappa}}$.]{
\label{fig:Psi6}
\includegraphics[width=0.45\textwidth]{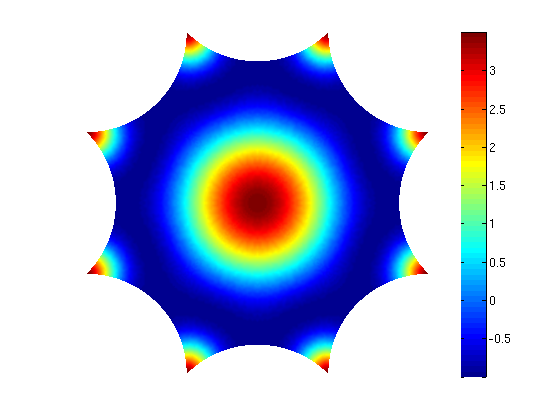}}
\caption{Plot of the eigenfunctions of the Laplace-Beltrami operator in the octagon $\Oct$ corresponding to the irreducible representations $\chi_{10}$ with eigenvalue $\mu_{\chi_{10}}^{\mD_{2\kappa}'}=15.0518$ (left) and $\chi_6$ with eigenvalue $\mu_{\chi_6}^{\widetilde{\mD}_{8\kappa}}=8.2501$ (right). }
\label{fig:Psi}
\end{figure}

With these notations, the two first equations of system \eqref{eq:sys_oct} give
\begin{align*}
 \Phi_{0020}&=\text{Span}\left( \Psi_1,\Psi_2,\Psi_3 \right)+\nu \left[\frac{1}{\alpha_c^4} +\frac{6}{5(\alpha_c^2-\mu_{\chi_6}^{\widetilde{\mD}_{8\kappa}})^2}\Psi_{\chi_6}^{\widetilde{\mD}_{8\kappa}}\right] \\
 \Phi_{1010}&=\text{Span}\left( \Psi_1,\Psi_2,\Psi_3 \right)+\frac{2\nu}{\sqrt{3}(\alpha_c^2-\mu_{\chi_{10}}^{\mD_{2\kappa}'})^2}\Psi_{\chi_{10}}^{\mD_{2\kappa}'}.
\end{align*}
A straightforward but lengthly calculation gives the expression of the coefficients $a$ and $b$ in the reduced equation \eqref{eq:NF}
 \bqq
\label{eq:alpha}
a= \nu^2\left[\frac{2}{\alpha_c^4}-\frac{8}{5(\alpha_c^2-\mu_{\chi_6}^{\widetilde{\mD}_{8\kappa}})^2}+\frac{4}{3(\alpha_c^2-\mu_{\chi_{10}}^{\mD_{2\kappa}'})^2} \right]-\eta
\eqq
\bqq
\label{eq:beta}
b = \nu^2\left[\frac{2}{\alpha_c^4} +\frac{72}{25(\alpha_c^2-\mu_{\chi_6}^{\widetilde{\mD}_{8\kappa}})^2}\right]-\frac{61}{25}\eta.
\eqq


From the analysis derived in \cite{melbourne:86}, we have the following result.
\begin{thm}\label{lemma:stability}
The stability of the branches of solutions corresponding to the three maximal isotropy subgroups given in proposition \ref{prop:fps} is:
 \begin{enumerate}
  \item [(i)] the $\mD_8$ branch is stable if and only if $a<b<0$,
  \item [(ii)] the $\widetilde{\mathbf{C}}_{6\kappa'}$ branch is stable if and only if $2a+b<0$ and $b-a<0$,
  \item [(iii)] the $\widetilde{\mD}_{2\kappa}$ branch is never stable.
 \end{enumerate}
\end{thm}

For each value of $(\nu,\eta)$ in $[0,5]\times[0,5]$ we have numerically computed the coefficients $a$ and $b$ given in equations \eqref{eq:alpha}-\eqref{eq:beta} and then checked if the stability conditions in Lemma \ref{lemma:stability} are satisfied. Our results are ploted in Figure \ref{fig:stab_oct}. We can see the different regions of the plane $(\nu,\eta)$ where the branches $\mD_8$ and $\widetilde{\mC}_{6\kappa'}$ are stable: in bright gray the region where $\mD_8$ is stable and in dark gray the region where $\widetilde{\mC}_{6\kappa'}$ is stable. Note that there is a whole region of parameter space where all the branches are unstable.

\begin{figure}[htp]
\centering
\includegraphics[width=0.7\textwidth]{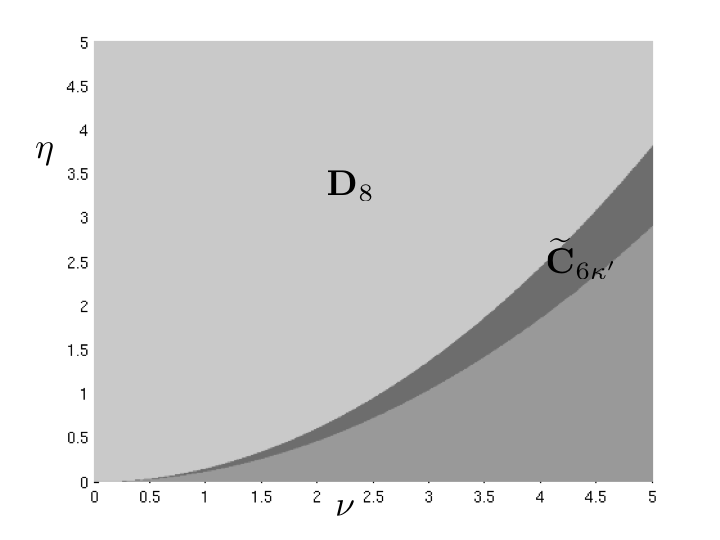}
\caption{Regions of the plane $(\nu,\eta)$ where the branches $\mD_8$ and $\widetilde{\mC}_{6\kappa'}$ are stable: in bright gray the region where $\mD_8$ is stable and in dark gray the region where $\widetilde{\mC}_{6\kappa'}$ is stable.}
\label{fig:stab_oct}
\end{figure}

\section{Radially localized solutions}\label{sec4}

Recently, there has been much progress made in understanding radially localized solutions
in the planar Swift-Hohenberg equation. For the
Swift-Hohenberg equation near the Turing instability, three types of
small amplitude radially symmetric localized solutions have been
proven to exist: a localized ring decaying to almost zero at the core,
a spot with a maximum at the origin (called spot A) and a spot with minimum at the origin
(called spot B); see
\cite{lloyd-sandstede:09,mccalla-sandstede:10,mccalla:11,mccalla-sandstede:12}. The
proofs rely on matching, at order $O(1/r^2)$, the ``core" manifold that
describes solutions that remain bounded near $r=0$ with the
``far-field" manifold that describes how solutions decay to the
trivial state for large $r$. The core manifold is found by carrying
out an asymptotic expansion involving Bessel functions while the
far-field manifold is calculated by carrying out a radial normal form
expansion near $r=\infty$. The main technical difficulty is that the
far-field normal form is only valid down to order $O(1/r)$ and so the
manifold has to be carefully followed up to order $O(1/r^2)$. Localized
rings occur due to a localized pulse in the far-field normal form
equations and require that the bifurcation of rolls at $r=\infty$ is
subcritical. Spot A solutions occur due to the unfolding of a
quadratic tangency of the core manifold and the cubic tangency of the
far-field manifold with the trivial state at onset; see \cite[Figure
4]{lloyd-sandstede:09}. The spot B state is formed by `gluing' the
spot A and localized ring solution. Crucially, all these localized
radial states are $\text{L}^2$-functions that can not be found via a
Lyapunov-Schmidt or center manifold reduction.

In this section, we will
only be interested in the existence of spot A type of solutions for
the Swift-Hohenberg equation \eqref{eq:sh}. The proof of the existence of such solutions will closely follow the one presented by Faye \etal \cite{faye-rankin-etal:12} for neural field equations on the Poincar\'e disk. In the hyperbolic case, the major difficulty comes from the fact that it is not clear how to define the
core and far-field manifolds in order to carry out the matching. It
turns out that the far-field manifold is easier to define than in the
Euclidean case since there is no bifurcation in the
far-field. However, calculating the core manifold is
significantly more involved than in the Euclidean case and constitutes
the main challenge in the existence proof of spots in hyperbolic
geometry. From now on we shall use the terms, spot and bump interchangeably to refer to the spot A states.

\subsection{Notations and definitions}

Throughtout this section, we work in geodesic polar coordinates $z=(\tau,\varphi)\in\D$, with $z=\tanh(\tau/2)e^{i\varphi}$. In these coordinates, the measure element defined in equation \eqref{eq:measure_element} is transformed into $\text{dm}(z)=\sinh(\tau)d\tau d\varphi$. Furthermore, in order to fix ideas, we set the value of $\alpha$ in equation \eqref{eq:sh} to be equal to $1$:
\bqq
\label{eq:sh1}
u_t=-(1+\Delta_\D)^2u+\lambda u+ \nu u^2- \eta u^3.
\eqq
The Laplace-Beltrami operator defined in equation \eqref{eq:laplace} can be written in geodesic polar coordinates as
\bqq
\label{eq:laplace_polar}
\Delta_\D=\frac{\partial^2}{\partial \tau^2}+\coth(\tau)\frac{\partial}{\partial \tau}+\sinh(\tau)^{-2}\frac{\partial^2}{\partial \varphi^2}.
\eqq
We define the radial part of the Laplace-Beltrami operator to be
\bqq
\label{eq:laplace_radial}
\Delta_\tau=\frac{\partial^2}{\partial \tau^2}+\coth(\tau)\frac{\partial}{\partial \tau}.
\eqq

Stationary radial solutions $u(\tau)$ of equation \eqref{eq:sh1} depend only on the radial variable $\tau=d_\D(z,0)$ and therefore satisfy the ordinary differential equation
\bqq
\label{eq:sh_stationary}
0=-(1+\Delta_\tau)^2u+\lambda u+ \nu u^2- \eta u^3.
\eqq
We are interested in finding localized solutions $u(\tau)$ of \eqref{eq:sh_stationary} that decay to zero as $\tau\rightarrow \infty$ and that belong to the functional space $\text{L}^2(\R^+,\sinh(\tau)d\tau)$. We shall therefore seek for such solutions for $\lambda<0$, where the background state $u=0$ is stable with respect to perturbations of the form $e^{\sigma t}e_{\rho,b}(z)$ (see \ref{subsub:linear}). In fact, our results are restricted to $0<-\lambda \ll 1$, and we shall construct localized radial solutions with small amplitude that bifurcate from $u=0$ at $\lambda=0$ into the region $\lambda<0$. 

In the Euclidean case, Bessel functions play a key role in this analysis of equation \eqref{eq:sh_stationary} close to the core $\tau=0$. In the hyperbolic setting, the analog of the Bessel functions are the associated Legendre functions of the first king $\mP_\alpha^\beta$ and second kind $\Q_\alpha^\beta$. We now recall their definition (see Erdelyi \cite{erdelyi:85}).
\begin{defi}
We denote $\mP_\alpha^\beta(z)$ and $\Q_\alpha^\beta(z)$ the two linearly independent solutions of the equation 
\bqq (1-z^2)\frac{d^2}{dz^2}
  u(z)-2z\frac{d}{dz} u(z)+ \left(
    \alpha(\alpha+1)-\frac{\beta^2}{1-z^2}\right)u(z)=0.
\label{eq:associated_legendre}
\eqq
$\mP_\alpha^\beta(z)$ and $\Q_\alpha^\beta(z)$ are respectively called associated Legendre function of the first and second kind. For $\beta=0$, we use the simplified notation $\mP_\alpha(z)=\mP_\alpha^0(z)$ and $\Q_\alpha(z)=\Q_\alpha^0(z)$.
\end{defi}

\subsection{Main result}

We can now state the main result of this section.
\begin{thm}[Existence of spot solutions]\label{thm:existence_spot_disk}
 Fix $\nu \neq 0$ and any $\eta \in \R$, then there exists $\lambda_*<0$ such that the Swift-Hohenberg equation \eqref{eq:sh1} has a stationary localized radial solution $u(\tau)\in L^2(\mathbb{R}^+,\sinh(\tau)d\tau)$ for each $\lambda\in(\lambda_*,0)$: these solutions stay close to $u=0$ and, for each fixed $\tau_*>0$, we have the asymptotics
\bqq
\label{eq:asymptotic_disk}
v(\tau)=\mathbf{c} \sqrt{|\lambda|}~\mP_{-\frac{1}{2}+i\frac{\sqrt{3}}{2}}(\cosh \tau)+O(\lambda) \text{ as }\lambda \rightarrow 0,
\eqq
uniformly in $0\leq \tau \leq \tau_*$ for an appropriate constant $\mathbf{c}$ with $\mathrm{sign}(\mathbf{c})=\mathrm{sign} (\nu)$.
\end{thm}

As for the Euclidean case, this theorem states that spots bifurcate for any value of $\nu \neq 0$, regardless of the value of $\eta$. The matching arguments in \ref{subsub:ffm} yield a similar theorem for the case $\lambda>0$ where one finds a bifurcating branch of solutions also given by \eqref{eq:asymptotic_disk}. However, these solutions are not $L^2(\mathbb{R}^+,\sinh(\tau)d\tau)$-functions.

Note that spots are initially unstable with respect to the PDE dynamics of the Swift-Hohenberg equation \eqref{eq:sh}. However, our numerical continuation results indicate that spots stabilize in a certain region in the $(\lambda,\nu)$-parameter space, at least when the nonlinear part of equation \eqref{eq:sh} is $\mathcal{N}(u)=\nu u^2-u^3$.

In the following subsection we sketch the proof of Theorem \ref{thm:existence_spot_disk}. It closely follows the one given in \cite{faye-rankin-etal:12} to which we refer for further details.

\subsection{Proof of Theorem \ref{thm:existence_spot_disk}}\label{sub:proof}

\subsubsection{The equation near the core}\label{subsub:enc}

We first rewrite equation \eqref{eq:sh_stationary} as a four dimensional system of non-autonomous differential equations to yield
\begin{align*}
 \partial_\tau u_1 &= u_3,\\
 \partial_\tau u_2 &= u_4,\\
 \left(\partial^2_\tau+\coth(\tau) \partial_\tau +1 \right)u_1 &= u_2,\\
 \left(\partial^2_\tau +\coth(\tau)\partial_\tau +1 \right) u_2 &= \lambda u_1+\nu u_1^2-\eta u_1^3,
\end{align*}
and we may rewrite \eqref{eq:sh_stationary} as a spatial dynamical system 
\bqq
\label{eq:core}
 U_\tau=\mathcal{A}(\tau)U+\mathcal{F}(U,\lambda),
\eqq
with
\bqs
\mathcal{A}(\tau)=\left( \begin{matrix}
 0 & 0 & 1 & 0 \\
 0 & 0 & 0 & 1 \\
 -1 & 1 & -\coth(\tau) & 0 \\
 0 & -1 & 0 & -\coth(\tau)
\end{matrix}
 \right),\quad \mathcal{F}(U,\lambda)=\left( \begin{matrix}
 0  \\
 0  \\
 0  \\
\lambda u_1+\nu u_1^2-\eta u_1^3
\end{matrix}
 \right).
\eqs
and $U=(u_1,u_2,u_3,u_4)^T$.

We begin our analysis by characterizing all small radial solutions of \eqref{eq:core} that are bounded and smooth in the interval $[0,\tau_0]$ for any fixed $\tau_0$. We set $\lambda=0$ and linearize \eqref{eq:core} about $U=0$ to get the linear system $\partial_\tau U=\mathcal{A}(\tau)U$. Proposition~\ref{prop:linear_syst_solution} defines the four linearly independent solutions of the linear system $\partial_\tau U=\mathcal{A}(\tau)U$.
\begin{prop}\label{prop:linear_syst_solution}
The linear system $\partial_\tau U=\mathcal{A}(\tau)U$ has four linearly independent solutions given by
\begin{eqnarray*}
 V_1(\tau)&=&\left(\mP_{\alpha_0}(\cosh \tau),0,\mP_{\alpha_0}^1(\cosh \tau),0 \right)^\mT \\
 V_2(\tau)&=&\left(V_2^1(\tau),\mP_{\alpha_0}(\cosh \tau),V_2^3(\tau),\mP_{\alpha_0}^1(\cosh \tau) \right)^\mT\\
 V_3(\tau)&=&\left(\Q_{\alpha_0}(\cosh \tau),0,\Q_{\alpha_0}^1(\cosh \tau),0 \right)^\mT\\
 V_4(\tau)&=&\left(V_4^1(\tau),\Q_{\alpha_0}(\cosh \tau),V_4^3(\tau),\Q_{\alpha_0}^1(\cosh \tau)\right)^\mT
\end{eqnarray*}
where 
\bqq
\label{eq:nu}
\alpha_0=-\frac{1}{2}+i\frac{\sqrt{3}}{2}
\eqq
and
\begin{align*}
 V_2^1(\tau)=&\mP_{\alpha_0}(\cosh \tau)\int_0^\tau \mP_{\alpha_0}(\cosh s)\Q_{\alpha_0}(\cosh s) \sinh(s)ds
-\Q_{\alpha_0}(\cosh \tau)\int_0^\tau \left(\mP_{\alpha_0}(\cosh s)\right)^2\sinh(s)ds,\\
V_2^3(\tau)=&\mP_{\alpha_0}^1(\cosh \tau)\int_0^\tau \mP_{\alpha_0}(\cosh s)\Q_{\alpha_0}(\cosh s) \sinh(s)ds
-\Q_{\alpha_0}^1(\cosh \tau)\int_0^\tau \left(\mP_{\alpha_0}(\cosh s)\right)^2\sinh(s)ds,\\
V_4^1(\tau)=&\mP_{\alpha_0}(\cosh \tau)\int_0^\tau \left(\Q_{\alpha_0}(\cosh s)\right)^2 \sinh(s)ds
-\Q_{\alpha_0}(\cosh \tau)\int_0^\tau \mP_{\alpha_0}(\cosh s)\Q_{\alpha_0}(\cosh s)\sinh(s)ds,\\
V_4^3(\tau)=&\mP_{\alpha_0}^1(\cosh \tau)\int_0^\tau \left(\Q_{\alpha_0}(\cosh s)\right)^2 \sinh(s)ds
-\Q_{\alpha_0}^1(\cosh \tau)\int_0^\tau \mP_{\alpha_0}(\cosh s)\Q_{\alpha_0}(\cosh s)\sinh(s)ds.
\end{align*}
\end{prop}

\begin{proof}
From equation \eqref{eq:associated_legendre}, the associated
  Legendre functions $\mP_{\alpha}(\cosh \cdot)$ and $\Q_{\alpha}(\cosh \cdot)$
  form a basis of solutions for the equation \bqs
\partial^2_\tau \Psi(\tau)+\coth(\tau) \partial_\tau \Psi(\tau)-{\alpha}({\alpha}+1)\Psi(\tau)=0. 
\eqs
If $\alpha=-\frac{1}{2}+i\frac{\sqrt{3}}{2}$ then we have ${\alpha}({\alpha}+1)=-1$. This implies that $\mP_{-\frac{1}{2}+i\frac{\sqrt{3}}{2}}(\cosh \tau)$ and $\Q_{-\frac{1}{2}+i\frac{\sqrt{3}}{2}}(\cosh \tau)$ are solutions of
\bqs
\partial^2_\tau \Psi(\tau)+\coth(\tau) \partial_\tau \Psi(\tau)+\Psi(\tau)=0.
\eqs
From now on $\alpha=\alpha_0=-\frac{1}{2}+i\frac{\sqrt{3}}{2}$. The solution of linear system $\partial_\tau U=\mathcal{A}(\tau)U$ can be found by inspecting the equivalent system
\bqs
\left(\partial^2_\tau+\coth(\tau) \partial_\tau +1\right)u_1=u_2,\quad \left(\partial^2_\tau+\coth(\tau) \partial_\tau +1\right)u_2=0.
\eqs
Consequently, solutions $V_1(\tau)=\left(\mP_{\alpha_0}(\cosh \tau),0,\mP_{\alpha_0}^1(\cosh \tau),0 \right)^\mT$ and $V_3(\tau)=\left(\Q_{\alpha_0}(\cosh \tau),0,\Q_{\alpha_0}^1(\cosh \tau),0 \right)^\mT$ are found to be two linearly independent solutions, where we have used the relations
\bqs
\partial_\tau \left(\mP_{\alpha_0}(\cosh \tau) \right)=\mP_{\alpha_0}^1(\cosh \tau), \quad \partial_\tau \left(\Q_{\alpha_0}(\cosh \tau) \right)=\Q_{\alpha_0}^1(\cosh \tau).
\eqs
In order to find the other two linearly independent solutions, we have to solve the two equations
\bqq
\label{eq:teq1}
\left(\partial_\tau^2+\coth \tau \partial_\tau+1\right) u(\tau)=\mP_{\alpha_0}( \cosh \tau),
\eqq
\bqq
\label{eq:teq2}
\left(\partial_\tau^2+\coth \tau \partial_\tau+1\right) u(\tau)=\Q_{\alpha_0}( \cosh \tau).
\eqq

From formula in \cite[page 123]{erdelyi:85}
 \bqs
 \mP_\alpha( \cosh \tau)\Q_\alpha^1(\cosh \tau)-\mP_\alpha^1( \cosh
\tau)\Q_\alpha(\cosh \tau)=-\frac{1}{\sinh \tau}, \quad \tau>0, \quad \nu\in\C
 \eqs
and straightforward computations, we obtain that the solutions of \eqref{eq:teq1} are given by
\begin{align*}
u(\tau)=c_1\mP_{\alpha_0}( \cosh \tau)+c_2\Q_{\alpha_0}(\cosh \tau)+\mP_{\alpha_0}(\cosh \tau)\int_0^\tau \mP_{\alpha_0}(\cosh s)\Q_{\alpha_0}(\cosh s) \sinh(s)ds\\
-\Q_{\alpha_0}(\cosh \tau)\int_0^\tau \left(\mP_{\alpha_0}(\cosh s)\right)^2\sinh(s)ds,
\end{align*}
with $c_1,c_2\in\R$ two real constants. Equivalently, the solutions of \eqref{eq:teq2} are given by
\begin{align*}
u(\tau)=c_3\mP_{\alpha_0}( \cosh \tau)+c_4\Q_{\alpha_0}(\cosh \tau)+\mP_{\alpha_0}(\cosh \tau)\int_0^\tau \left(\Q_{\alpha_0}(\cosh s)\right)^2 \sinh(s)ds\\
-\Q_{\alpha_0}(\cosh \tau)\int_0^\tau \mP_{\alpha_0}(\cosh s)\Q_{\alpha_0}(\cosh s)\sinh(s)ds,
\end{align*}
with constants $c_3,c_4\in\R$.

Chosing only linearly independent solutions, we finally obtain the result stated in the proposition.
\end{proof}

In Table \ref{table:expansion}, we summarize the expansions of the
associated Legendre functions in the limits $\tau\rightarrow 0$ and
$\tau \rightarrow \infty$; Proposition~\ref{prog:exp_consts} defines
the constants in the $\tau\rightarrow\infty$ limit. Thus, $V_1(\tau)$
and $V_2(\tau)$ stay bounded as $\tau\rightarrow 0$, while the norms
of $V_3(\tau)$ and $V_4(\tau)$ behave like $\ln \tau$ as
$\tau\rightarrow 0$. We expect that the set of solutions of
\eqref{eq:core} that are bounded as $\tau\rightarrow 0$ forms a
two-dimensional manifold in $\R^4$ for each fixed $\tau>0$. We denote
the projection $P_-^{cu}(\tau_0)$ onto the space spanned by
$V_1(\tau_0),V_2(\tau_0)$ with null space given by the span of
$V_3(\tau_0),V_4(\tau_0)$.

\begin{table}[h]   
\begin{center}
\begin{tabular}{|c|c|c|} \hline  & $\tau\rightarrow 0$ & $\tau \rightarrow \infty$ \\ \hline
 $\mP_{\alpha_0}(\cosh \cdot)$ & $1+O(\tau^2)$ & $C_0\cos(\frac{\sqrt{3}\tau}{2}+\Phi_0)e^{-\frac{\tau}{2}}+O(e^{-\frac{3\tau}{2}})$  \\ 
 $\mP_{\alpha_0}^1(\cosh \cdot)$ & $\tau \left(-\frac{1}{2}+O(\tau^2)\right)$ & $C_0\cos(\frac{\sqrt{3}\tau}{2}+\Phi_0+\frac{2\pi}{3})e^{-\frac{\tau}{2}}+O(e^{-\frac{3\tau}{2}})$\\ 
 $\Q_{\alpha_0}( \cosh \cdot)$ & $(-1+O(\tau^2))\ln \tau +O(1)$ & $C_1\cos(\frac{\sqrt{3}\tau}{2}-\Phi_1)e^{-\frac{\tau}{2}}+O(e^{-\frac{3\tau}{2}})$\\ 
 $\Q_{\alpha_0}^1( \cosh \cdot)$ & $(1+O(\tau^2))\tau \ln \tau -\frac{1}{\tau}+O(1)$ & $C_1\cos(\frac{\sqrt{3}\tau}{2}-\Phi_1+\frac{2\pi}{3})e^{-\frac{\tau}{2}}+O(e^{-\frac{3\tau}{2}})$\\ \hline
\end{tabular}
\end{center}
\caption{Expansions of associated Legendre functions $\mP_{\alpha_0}^k(\cosh \cdot)$ and $\Q_{\alpha_0}^k(\cosh \cdot)$ for $\tau\rightarrow 0$ and $\tau \rightarrow \infty$; see \cite{erdelyi:85,virchenko-fedotova:01}. $C_0,C_1,\Phi_0$ and $\Phi_1$ are all real constants given in Proposition~\ref{prog:exp_consts}.}\label{table:expansion}
\end{table}

\begin{prop}\label{prog:exp_consts}
 The constants $C_0,C_1,\Phi_0$ and $\Phi_1$, given in Table \ref{table:expansion}, are
 \bqq
\label{eq:C0_C1_formula}
C_0=2\sqrt{\frac{2\sqrt{3}}{3\pi\tanh\left(\frac{\sqrt{3}\pi}{2}\right)}}\text{ and } C_1=\sqrt{\frac{2\pi\sqrt{3}\tanh\left(\frac{\sqrt{3}\pi}{2}\right)}{3}},
\eqq
 \bqq
\label{eq:relation_Phi}
\Phi_0=\arg \left(\frac{\Gamma\left(i\frac{\sqrt{3}}{2} \right)}{\Gamma\left(\frac{1}{2}+i\frac{\sqrt{3}}{2}\right)} \right) \text{ and } \Phi_0+\Phi_1=-\frac{\pi}{2}.
\eqq
\end{prop}
\begin{proof}
See \cite{faye-rankin-etal:12}.
\end{proof}

We are now able to present the hyperbolic equivalent of Lemma 1 \cite{lloyd-sandstede:09} for the center-unstable manifold $\mathcal{W}_-^{cu}(\lambda)$, the set of bounded and continuous solutions of \eqref{eq:core} close to $\tau=0$. That is, fix $\tau_0>0$ and $\delta_0>0$:
\bqs
\mathcal{W}_-^{cu}(\lambda)=\{U\in \mathcal{C}^0([0,\tau_0],\R^4)\text{ solution of \eqref{eq:core} }~|~ \sup_{0\leq \tau\leq
    \tau_0}|U(\tau)|<\delta_0\text{ for }|\lambda|<\delta_0\}.
\eqs
    
\begin{lem}\label{lem:core_disk}
  Fix $\tau_0>0$, then there exist constants $\delta_0,\delta_1$ such
  that the set $\mathcal{W}_-^{cu}(\lambda)$ of solutions $U(\tau)$ of
  \eqref{eq:core} for which $\sup_{0\leq \tau\leq
    \tau_0}|U(\tau)|<\delta_0$ is, for $|\lambda|<\delta_0$, is a smooth
  two-dimensional manifold. Furthermore,
  $U\in\mathcal{W}_-^{cu}(\lambda)$ with
  $|P_-^{cu}(\tau_0)U(\tau_0)|<\delta_1$ if and only if 
\bqq
\begin{split}
U(\tau_0)&=\tilde d_1 V_1(\tau_0)+\tilde d_2 V_2(\tau_0)+V_3(\tau_0)O_{\tau_0}(|\lambda||\tilde d|+|\tilde d|^2)\\
~&~+V_4(\tau_0)\left(-\left[\mathcal{I}+o(1) \right] \nu \tilde{d}_1^2+O_{\tau_0}(|\lambda||\tilde d|+|\tilde d_2|^2+|\tilde d_1|^3) \right),
\label{eq:estimate_tau0_disk}
\end{split}
\eqq
with
\bqs
 \mathcal{I}=\int_0^{\infty}\left(\mP_{\alpha_0}( \cosh s) \right)^3 \sinh sds <\infty,\\
\eqs
for some $\tilde d=(\tilde d_1,\tilde d_2)\in\R^2$ with $|\tilde d|<\delta_1$, where the right hand side in \eqref{eq:estimate_tau0_disk} depends smoothly on $(\tilde d,\lambda)$.
\end{lem}

\begin{proof}
We observe that four independent solutions to the adjoint problem $\partial_\tau U=-\mathcal{A}^T(\tau)U$ are given by (see \cite{faye-rankin-etal:12}):
\begin{eqnarray*}
 W_1(\tau)&=&\sinh \tau  \left(\Q_{\alpha_0}^1(\cosh \tau),W_1^2(\tau),-\Q_{\alpha_0}(\cosh \tau),W_1^4(\tau) \right)^\mT,\\
 W_2(\tau)&=&\sinh \tau \left(0,\Q_{\alpha_0}^1(\cosh \tau),0,-\Q_{\alpha_0}(\cosh \tau)\right)^\mT,\\
 W_3(\tau)&=&\sinh \tau \left(\mP_{\alpha_0}^1( \cosh \tau),W_3^2(\tau),-\mP_{\alpha_0}( \cosh \tau),W_3^4(\tau) \right)^\mT \\
 W_4(\tau)&=&\sinh \tau \left(0,\mP_{\alpha_0}^1( \cosh \tau),0,-\mP_{\alpha_0}( \cosh \tau)\right)^\mT,
\end{eqnarray*}
where
\begin{align*}
 W_1^2(\tau)=&\mP_{\alpha_0}^1(\cosh \tau)\int_0^\tau \mP_{\alpha_0}^1( \cosh s) \Q_{\alpha_0}^1(\cosh s) \sinh(s)ds- \Q_{\alpha_0}^1(\cosh \tau)\int_0^\tau \left(\mP_{\alpha_0}^1( \cosh s) \right)^2\sinh(s)ds,\\
 W_1^4(\tau)=&-\Q_{\alpha_0}(\cosh \tau)-\mP_{\alpha_0}(\cosh \tau)\int_0^\tau \mP_{\alpha_0}^1( \cosh s) \Q_{\alpha_0}^1(\cosh s) \sinh(s)ds+ \Q_{\alpha_0}(\cosh \tau)\int_0^\tau \left(\mP_{\alpha_0}^1( \cosh s) \right)^2\sinh(s)ds,\\
 W_3^2(\tau)=&\mP_{\alpha_0}^1(\cosh \tau)\int_0^\tau \left(\Q_{\alpha_0}^1(\cosh s)\right)^2 \sinh(s)ds-\Q_{\alpha_0}^1(\cosh \tau)\int_0^\tau \mP_{\alpha_0}^1(\cosh s)\Q_{\alpha_0}^1(\cosh s)\sinh(s)ds,\\
 W_3^4(\tau)=&-\mP_{\alpha_0}( \cosh \tau)-\mP_{\alpha_0}(\cosh \tau)\int_0^\tau \left(\Q_{\alpha_0}^1(\cosh s)\right)^2 \sinh(s)ds-\Q_{\alpha_0}(\cosh \tau)\int_0^\tau \mP_{\alpha_0}^1(\cosh s)\Q_{\alpha_0}^1(\cosh s)\sinh(s)ds.
\end{align*}

It follows from 
\bqs
\sinh \tau\left(\mP_{\alpha_0}^1( \cosh \tau)\Q_{\alpha_0}(\cosh \tau)-\mP_{\alpha_0}( \cosh \tau)\Q_{\alpha_0}^1(\cosh \tau)\right)=1,
\eqs
that
\bqs
\langle V_i(\tau),W_j(\tau) \rangle_{\R^4}=\delta_{i,j} \quad i,j=1,\dots,4,
\eqs
is independent of $\tau$. For a given $\tilde d=(\tilde d_1,\tilde d_2)\in\R^2$, we consider the fixed-point equation:
\begin{eqnarray}
 U(\tau)&=& \sum_{j=1}^2\tilde d_j V_j(\tau)+\sum_{j=1}^2 V_j(\tau)\int_{\tau_0}^\tau\langle W_j(s),\mathcal{F}(U(s),\lambda) \rangle ds + \sum_{j=3}^4 V_j(\tau)\int_{0}^{\tau}\langle W_j(s),\mathcal{F}(U(s),\lambda) \rangle ds\nonumber \\
~&=& \sum_{j=1}^2\tilde d_j V_j(\tau)+\sum_{j=1}^2 V_j(\tau)\int_{\tau_0}^\tau W_{j,4}(s)\mathcal{F}_4(U(s),\lambda) ds + \sum_{j=3}^4 V_j(\tau)\int_{0}^{\tau} W_{j,4}(s)\mathcal{F}_4(U(s),\lambda) ds \nonumber \\
\label{eq:fixed_point_disk}
\end{eqnarray}
on $\mathcal{C}^0([0,\tau_0],\R^4)$, where $W_{j,4}(\tau)$ (resp. $\mathcal{F}_4(U(\tau),\lambda)$) denotes the fourth component of $W_j(\tau)$ (resp. $\mathcal{F}(U(\tau),\lambda)$). 

Hence, we have that:
\begin{itemize}
 \item Each solution $U\in\mathcal{C}^0([0,\tau_0],\R^4)$ of \eqref{eq:fixed_point_disk} gives a solution of \eqref{eq:core} that is bounded on $[0,\tau_0]$.
 \item Every bounded solution $U\in\mathcal{C}^0([0,\tau_0],\R^4)$ of
   \eqref{eq:core} satisfies \eqref{eq:fixed_point_disk} provided
   that we add $\tilde d_3 V_3(\tau)+\tilde d_4V_4(\tau)$ to the right
   hand side for an appropriate $\tilde d\in\R^4$.
 \item Existence of solutions of \eqref{eq:fixed_point_disk} is given by the uniform contraction mapping principle for sufficiently small $(\tilde d_1,\tilde d_2)$ and $\lambda$.
 \item The resulting solution $U$ satisfies
$U(\tau)=\sum_{j=1}^2\tilde d_j V_j(\tau)+O_{\tau_0}(|\lambda||\tilde d|+|\tilde d|^2)$ on $[0,\tau_0]$.
\end{itemize}

We also need to compute the quadratic coefficient in $\tilde d$ in
front of $V_4(\tau_0)$. Using a Taylor expansion, we find that this coefficient is given by
\bqs
\nu \int_0^{\tau_0}W_{4,4}(s) \left(\mP_{\alpha_0}( \cosh s) \right)^2ds=-\nu \left[\int_0^{\infty}\left(\mP_{\alpha_0}( \cosh s) \right)^3 \sinh sds +o(1)\right].
\eqs
\end{proof}

\subsubsection{The far-field equations and matching}\label{subsub:ffm}

In this section, we look into the far-feild regime, where the radial variable $\tau$ is large. We make spatial dynamical system \eqref{eq:core} autonomous by augmenting the system with the equation  $\partial_\tau \epsilon=-\epsilon(2+\epsilon)$ where $\epsilon(\tau)=\coth \tau-1$ to yield the new system
\bqq
\frac{d}{d\tau}\left(\begin{array}{lcl}
                   u_1\\
                   u_2\\
                   u_3\\
                   u_4\\
                   \epsilon
                  \end{array}
 \right)= \left(\begin{matrix}
                   u_3\\
                   u_4\\
                   -u_1+u_2-(1+\epsilon) u_3\\
                   -u_2-(1+\epsilon) u_4+\lambda u_1+\nu u_1^2-\eta u_1^3\\
                   -\epsilon(2+\epsilon)
                  \end{matrix}
 \right).
\label{eq:far_field}
\eqq

In the remainder of this section, we focus on the regime $\epsilon \approx 0$ which corresponds to the far field $\tau \gg 1$. We denote  $\mathcal{A}(\infty,\lambda)$ by the matrix
\bqs
 \mathcal{A}(\infty,\lambda)=\left( \begin{matrix}
 0 & 0 & 1 & 0 \\
 0 & 0 & 0 & 1 \\
 -1 & 1 & -1 & 0 \\
 \lambda & -1 & 0 & -1
\end{matrix}
 \right),
\eqs
where $\partial_{\tau}U = \mathcal{A}(\tau,\lambda)U$ is
the linearisation of \eqref{eq:far_field} about the
trivial state.  We find that the matrix $\mathcal{A}(\infty,0)$ has
four eigenvalues $\alpha_0,\bar \alpha_0$ with multiplicity two ($\alpha_0$ is
defined in equation \eqref{eq:nu}); see
Figure~\ref{fig:eigenvalues}. As $\Re (\nu)=-1/2$, the trivial
state $U=0$ is asymptotically stable at $\lambda=0$ and then there is
no bifurcation at the far field. Recall that in the Euclidean case a Turing instability occurs at infinity. In Figure \ref{fig:eigenvalues}, we summarize how the eigenvalues $\ell$ of $\mathcal{A}(\infty,\lambda)$ split close to $\lambda=0$. For $\lambda>0$, there exist four complex conjugate eigenvalues with $\Re (\ell) =-1/2$. For $\lambda<0$, there exist four complex conjugate eigenvalues with $\Re (\ell) \neq-1/2$ and stable manifold $\mathcal{W}^s_+$ is the union of the stable fast manifold (which we denote by $\mathcal{W}^{sf}_+$) and the stable slow manifold (which we denote by $\mathcal{W}^{ss}_+$) corresponding to the fast and slow decay to the trivial state. 

\begin{figure}[h]
\centering
\includegraphics[width=0.8\textwidth]{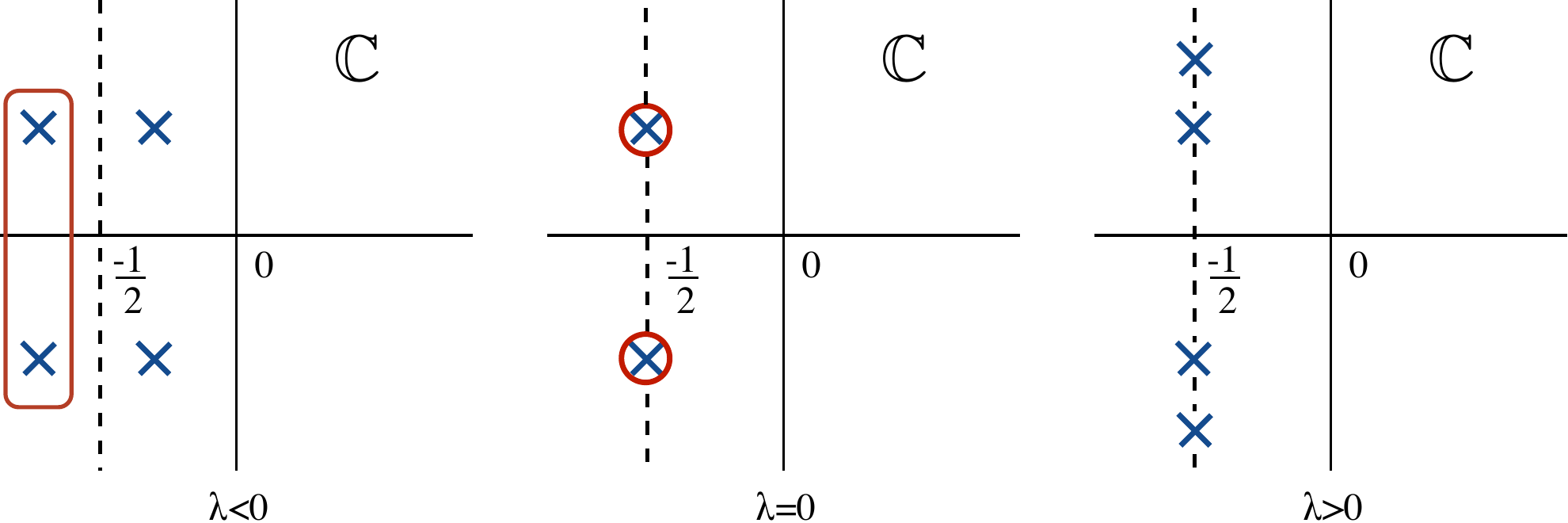}
 \caption{Schematic splitting of the eigenvalues $\ell$ of $\mathcal{A}(\infty,\lambda)$ for different values of $\lambda$. At $\lambda=0$, the multiplicity is two. Eigenvalues in the red box correspond to the stable fast manifold $\mathcal{W}_+^{sf}(\lambda)$.}
 \label{fig:eigenvalues}
\end{figure}

First we argue that the centre-unstable manifold $\mathcal{W}^{cu}_-$ and stable manifold $\mathcal{W}^s_+$ should intersect. We have that $V_1(\tau)$ and $V_3(\tau)$ decay like $e^{-\tau/2}$, while $V_2(\tau)$ and $V_4(\tau)$ decay like $\tau e^{-\tau/2}$ as $\tau\rightarrow \infty$. Hence the tangent space of the stable manifold at $(u,\lambda)=0$ is spanned by $\left(V_1(\tau),V_2(\tau),V_3(\tau),V_4(\tau)\right)$. On the other hand, we showed in Lemma \ref{lem:core_disk} that the tangent space of the core manifold is spanned by $V_1(\tau)$ and $V_2(\tau)$. Then these tangent spaces would intersect along the two-dimensional subspace spanned by $V_1(\tau)$ and $V_2(\tau)$.

In order to show that the centre-unstable manifold $\mathcal{W}^{cu}_-$ intersects with the stable fast manifold $\mathcal{W}^{sf}_+$, we need to find an explicit description of $\mathcal{W}^{sf}_+$. To do this, we use successive, well chosen change of variables to put \eqref{eq:far_field} into normal form. We first define the linear change of coordinates
\bqs
 U=\widetilde{A}\left( \begin{matrix}
 1  \\
 0  \\
 \alpha_0  \\
 0 
\end{matrix}
 \right)+\widetilde{B}\left( \begin{matrix}
 0 \\
 2\alpha_0+1  \\
 1 \\
 \alpha_0(2\alpha_0+1)
\end{matrix}
 \right)+\text{c.c.},
\eqs
or equivalently,
\bqs
\left( \begin{matrix}
 \widetilde{A}\\
\widetilde{B}
\end{matrix}
 \right)=\left( \begin{matrix}
\frac{1}{2} u_1+i\frac{\sqrt{3}}{3}\left(-\frac{1}{2} u_1-\frac{1}{3} u_2- u_3-\frac{2}{3} u_4 \right)\\
-\frac{1}{3} \left(\frac{1}{2} u_2+u_4\right)-i\frac{\sqrt{3}}{6} u_2\end{matrix}
 \right),\quad\mbox{and}\quad  U=( u_1, u_2, u_3, u_4)^T.
\eqs

In these coordinates, the linear part of \eqref{eq:far_field} becomes at $\lambda=0$,
\begin{eqnarray}\label{eq:transformed_disk}
 \partial_\tau \widetilde{A}&=&\left(-\frac{1}{2}-\frac{\epsilon}{2}+ i\left[ \frac{\sqrt{3}}{2}-\frac{\epsilon\sqrt{3}}{6}\right]\right)\widetilde{A}+\left(1+\frac{\epsilon}{3} \right)\widetilde{B}+\epsilon \left(\frac{1}{2}- \frac{i\sqrt{3}}{6}\right) \overline{\widetilde{A}}-\frac{\epsilon}{3}\overline{\widetilde{B}}, \nonumber \\
 \partial_\tau \widetilde{B}&=&\left( -\frac{1}{2}-\frac{\epsilon}{2}+ i\left[ \frac{\sqrt{3}}{2}-\frac{\epsilon\sqrt{3}}{6}\right] \right)\widetilde{B}+\epsilon \left(-\frac{1}{2}+\frac{i\sqrt{3}}{6}\right)\overline{\widetilde{B}}, \nonumber \\
\partial_\tau \epsilon &=& -\epsilon(2+\epsilon).
\end{eqnarray}

\begin{lem}\label{lem:hyper_normal}
 Fix $0<m<\infty$, then there exists a change of coordinates
\bqq
\label{eq:change_coordinate_AB_disk}
\left(\begin{matrix}
       A\\
B
      \end{matrix}
 \right)=e^{-i\phi(r)}[1+\mathcal{T}(\epsilon)]
\left( \begin{matrix}
 \widetilde{A}\\
\widetilde{B}
\end{matrix}
 \right)+O((|\lambda|+|\widetilde{A}|+|\widetilde{B}|)(|\widetilde{A}|+|\widetilde{B}|)),
\eqq
so that \eqref{eq:transformed_disk} becomes

\begin{eqnarray}\label{eq:normal_transformed_disk}
 \partial_\tau A&=& \left(-\frac{1}{2}-\frac{\epsilon}{2} \right)A +  B +\text{h.o.t.},  \nonumber\\
\partial_\tau B &=& \left(-\frac{1}{2}-\frac{\epsilon}{2} \right)B - \frac{1}{3} \lambda A +\text{h.o.t.}, \nonumber\\
\partial_\tau \epsilon &=&-\epsilon(2+\epsilon).
\end{eqnarray}

The coordinate change is polynomial in $(A,B,\epsilon)$ and smooth in $\lambda$ and $\mathcal{T}(\epsilon)=O(\epsilon)$ is linear and upper triangular for each $\epsilon$, while $\phi(r)$ satisfies
\bqs
\partial_r \phi(r)=\frac{\sqrt{3}}{2} +O(|\lambda|+|\epsilon|+|A|^2),\quad \phi(0)=0.
\eqs

\end{lem}

Note that at $(\alpha,\lambda)=(0,0)$, the trivial state $(A,B)=(0,0)$ is hyperbolic such that the higher order terms in equation \eqref{eq:normal_transformed_disk} are exponentially small for $\tau\gg 1$ and $\lambda$ small enough and can be neglected. We can also directly solve the linear part of equation \eqref{eq:normal_transformed_disk}, for $\lambda<0$, to obtain
\bqq
\left(\begin{array}{ll}
     A(\tau) \\
     B(\tau)
      \end{array}
 \right)= \frac{1}{\sqrt{\sinh(\tau)}}\left[  q_1 e^{-\tau\sqrt{\frac{-\lambda}{3}}}\left(\begin{array}{c}
      1 \\
       -\sqrt{\frac{-\lambda}{3}}
      \end{array}
 \right)
+q_2 e^{\tau\sqrt{\frac{-\lambda}{3}}}\left(\begin{array}{c}
        1\\
        \sqrt{\frac{-\lambda}{3}}
      \end{array}
 \right)\right].
\label{eq:center_stable_AB_disk}
\eqq

We want to find solutions that have a finite energy density with
respect to the hyperbolic measure, i.e. functions that are in
$\text{L}^2(\R^+,\sinh(\tau)d\tau)$. This restriction implies that we
need to track the stable fast manifold $\mathcal{W}_+^{sf}(\lambda)$ of
equation \eqref{eq:center_stable_AB_disk} which corresponds to
eigenvalues $\ell$ of $\mathcal{A}(\infty,\lambda)$ with real part
less than $-\frac{1}{2}$ as shown in Figure
\ref{fig:eigenvalues}. Thus, for each fixed $\tau_0\gg 1$ and
for all sufficiently small $\lambda<0$, we can write the
$\tau=\tau_0$-fiber of the stable fast manifold $\mathcal{W}_+^{sf}(\lambda)$
of equation \eqref{eq:center_stable_AB_disk} near $U=0$ as

\bqq
\mathcal{W}_+^{sf}(\lambda)\mid_{\tau=\tau_0}~:~\left(\begin{array}{ll}
     A \\
     B
      \end{array}
 \right)= e^{-\tau_0/2}\left[ -\mu \sqrt{\frac{-\lambda}{3}}\left(1+O_{\tau_0}(|\lambda|) \right)\left(\begin{array}{ll}
       \tau_0 \\
       1
      \end{array}
 \right)
+\mu \sqrt{\frac{-\lambda}{3}} \left(\begin{array}{ll}
        1\\
        0
      \end{array}
 \right)\right],
\label{eq:stable_fast_AB_disk}
\eqq
for $\mu \in \C$.

We can now finish the proof of Theorem~\ref{thm:existence_spot_disk}. To do this, we need to  find nontrivial intersections of the stable fast manifold $\mathcal{W}_+^{sf}(\lambda)$ with the centre-unstable manifold $\mathcal{W}_-^{cu}(\lambda)$. To this end, we write the expansion \eqref{eq:estimate_tau0_disk} for each fixed $\tau_0\gg 1$ in the $(\tilde A,\tilde B)$ coordinates and afterwards in the coordinates $(A,B)$. Using the expansions of the associated Legendre  functions given in Table \ref{table:expansion} we arrive at the expression
\bqs
\left(\begin{array}{ll}
       \tilde A \\
       \tilde B
      \end{array}
 \right)=e^{-\tau_0/2} \left[ e^{i\left(\frac{\sqrt{3}}{2}\tau_0+\Phi_0\right)}\left( \begin{array}{ll}
\frac{C_0}{2}\tilde d_1(1+O(1))+\tau_0\tilde d_2\left(-i\frac{\sqrt{3}C_0}{6}+O(1)\right)\\
-\tilde d_2(i\frac{\sqrt{3}C_0}{6}+O(1))+\frac{C_1\sqrt{3}}{6}\left( -\nu(\mathcal{I}+o(1))\tilde{d}_1^2  \right)
                                \end{array}
 \right)\right.
\eqs
\bqq
\label{eq:Wcu_tatb_disk}
\left. +e^{i\left(\frac{\sqrt{3}}{2}\tau_0+\Phi_0\right)}\left( \begin{array}{ll}
                         O_{\tau_0}(\lambda|\tilde d|+|\tilde d|^2)\\
                         O_{\tau_0}(\lambda|\tilde d|+|\tilde d_2|^2+|\tilde d_1|^3)
                        \end{array}
\right)\right].
\eqq

We can apply the transformation \eqref{eq:change_coordinate_AB_disk} to equation \eqref{eq:Wcu_tatb_disk} and obtain the expansion
\begin{align}
\mathcal{W}_-^{cu}(\lambda)\mid _{\tau=\tau_0}~:~\left(\begin{array}{ll}
        A \\
        B
      \end{array}
 \right)=&e^{i(\Phi_0+O(\tau_0^{-2})+O_{\tau_0}(\lambda|\tilde d|+|\tilde d|^2)}\left( \begin{array}{c}
                         O_{\tau_0}(\lambda|\tilde d|+|\tilde d|^2)\\
                         O_{\tau_0}(\lambda|\tilde d|+|\tilde d_2|^2+|\tilde d_1|^3)
                        \end{array}
\right)\nonumber\\
&+e^{i(\Phi_0+O(\tau_0^{-2})+O_{\tau_0}(\lambda|\tilde d|+|\tilde d|^2)}\left( \begin{array}{c}
 \frac{C_0}{2}\tilde d_1(1+O(1))+\tau_0\tilde d_2(-i\frac{\sqrt{3}C_0}{6}+O(1))\\
 -\tilde d_2(i\frac{\sqrt{3}C_0}{6}+O(1))
                                \end{array}
 \right)\nonumber\\
\label{eq:Wcu_AB_disk}
&+e^{i(\Phi_0+O(\tau_0^{-2})+O_{\tau_0}(\lambda|\tilde d|+|\tilde d|^2)}\left( \begin{array}{c}
 0\\
\frac{C_1\sqrt{3}}{6}\left(-\nu(\mathcal{I}+o(1))\tilde{d}_1^2 \right)
                                \end{array}
 \right).
\end{align}

The final step of the analysis consists in finding nontrivial intersections of the stable fast manifold $\mathcal{W}_+^{sf}(\lambda)$ given above in equation \eqref{eq:stable_fast_AB_disk} and the core manifold $\mathcal{W}_-^{cu}(\lambda)$ given in \eqref{eq:Wcu_AB_disk}. We can easily solve this problem in $(\tilde d_1,\tilde d_2)$ to find that
\bqs
\tilde d_1 = \frac{2\sqrt{-\lambda}}{\nu C_1\mathcal{I}},\quad \text{and}\quad \tilde d_2=O(\lambda).
\eqs
Then $\mathbf{c}=\dfrac{2}{\nu C_1\mathcal{I}}$ in equation \eqref{eq:asymptotic_disk} and $\text{sign}(\mathbf{c})=\text{sign}(\nu)$. This completes the proof of Theorem~\ref{thm:existence_spot_disk}.

\subsection{Numerical computation of spots}\label{sub:numerics}

In this section, we describe the use of numerical continuation (and
the continuation package
\textsc{AUTO07p}~\cite{doedel-champneys-etal:97}) to compute solutions
of the systems of ODEs described by \eqref{eq:core}. Solutions of
these spatial dynamical system correspond to steady states of the
Swift-Hohenberg equation \eqref{eq:sh} where the radial
coordinate  $\tau$ has been recast as time in
\textsc{AUTO07p}'s boundary value problem (BVP) solver. The BVP is set
up on the domain  $\tau\in[0,L]$ with homogeneous
Neumann boundary conditions given by $u_2(0)=u_4(0)=u_1(L)=u_3(L)=0$.
Typical parameters for the \textsc{AUTO07p} radial
computations are $L=1000$ and \textsc{AUTO07p}'s NTST=400 with
standard relative tolerances that are specified in \textsc{AUTO07p}'s
manual. Initial data for continuations are obtained by first solving the radial Swift-Hohenberg equation \eqref{eq:sh_stationary} with the {\sc fsolve} routine of Matlab and then using a parameter continuation to add in the radial terms.

\begin{figure}[h]
\centering
\includegraphics[width=0.6\textwidth]{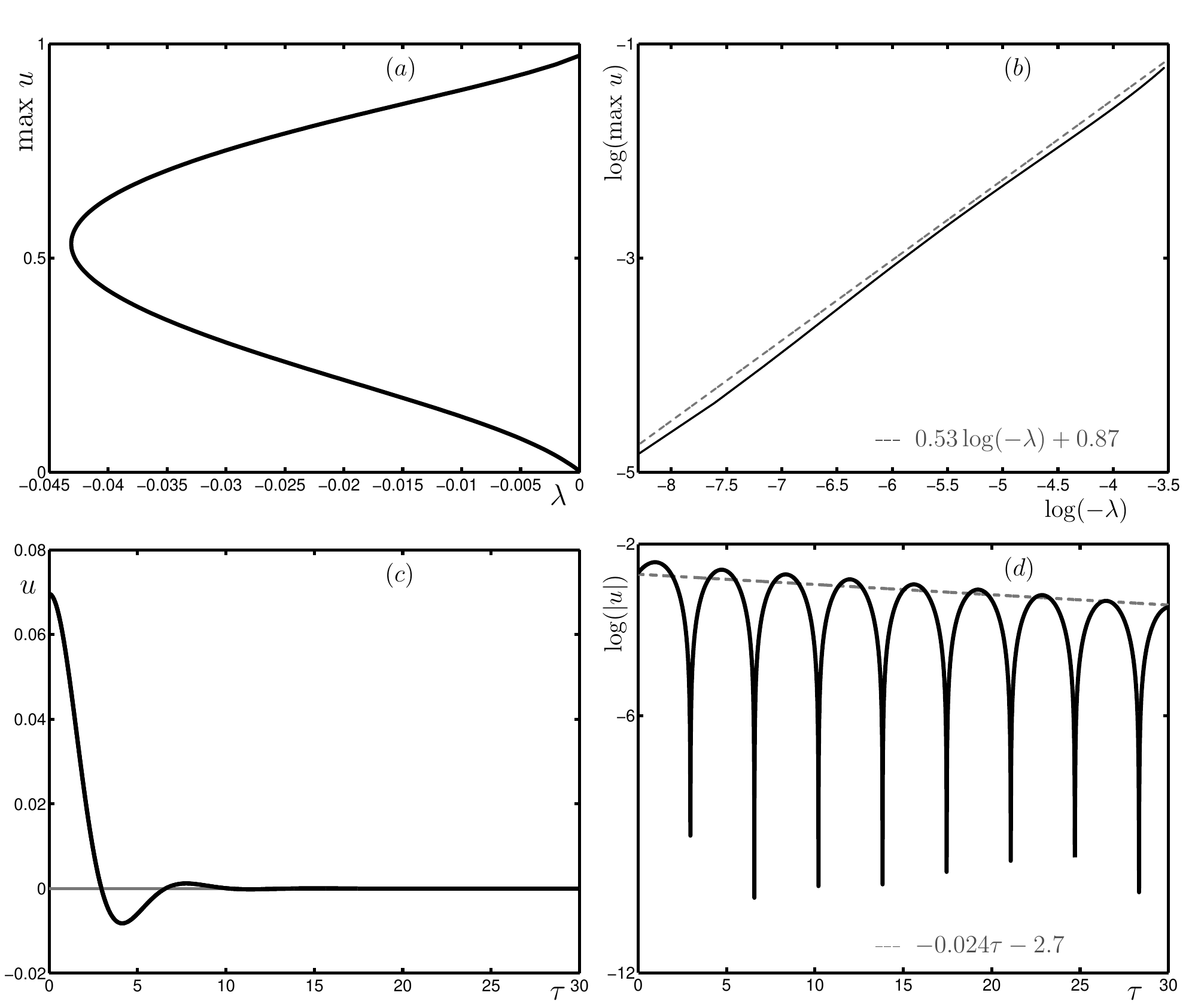}
 \caption{Spots of \eqref{eq:sh_stationary} are shown for $\nu=0.8$ and $\eta=1$. (a) Bifurcation diagram of spots. (b) Log-plot of the maximum height $u(0)$ versus the parameter $\lambda$: the solid line is the result of numerical computations, which we fit with the dashed line $0.53\log(-\lambda)+0.87$. $\left(\text{c} \right)$ Profile of a typical spot. (d) Decay rate of the tails of the solutions rescaled solutions $u(\tau)e^{\tau/2}$ in black at $\lambda=-0.015$. Linear best fits is shown in gray. }
 \label{fig:spot08}
\end{figure}

In Theorem \ref{thm:existence_spot_disk}, we have shown the existence of spots of equation \eqref{eq:sh_stationary} for any fixed $\nu\neq 0$ and $\eta\in \R$ when $0<|\lambda|\ll 1$. We compute spot solutions for $\nu=0.8$ and $\eta=1$ and summarize the results in Figure \ref{fig:spot08}. Spots do bifurcate off $u=0$ at $\lambda=0$ and turn around at a saddle-node bifurcation at $\lambda=-0.0425$. At this fold, spots regain stability with respect to radial perturbations, but they remain unstable with respect to general perturbations. The corresponding bifurcation diagram (see Figure  \ref{fig:spot08}(a)) is similar to the one computed in Euclidean geometry by Lloyd \& Sandstede \cite{lloyd-sandstede:09}. The computations confirm the scaling $u(0)\approx \sqrt{|\lambda|}$ as $\lambda\rightarrow 0$. We also verify in Figure \ref{fig:spot08}(d) that the spots found in Theorem  \ref{thm:existence_spot_disk} are $\text{L}^2(\R^+,\sinh(\tau)d\tau)$ functions. Close to onset at $\lambda=-0.015$, we plot the decay rate of the tail of the rescaled spots $u(\tau)e^{\tau/2}$ which confirms that spots decay faster than $e^{-\tau/2}$ and hence are in $\text{L}^2(\R^+,\sinh(\tau)d\tau)$. 

\begin{figure}[h]
\centering
\includegraphics[width=0.6\textwidth]{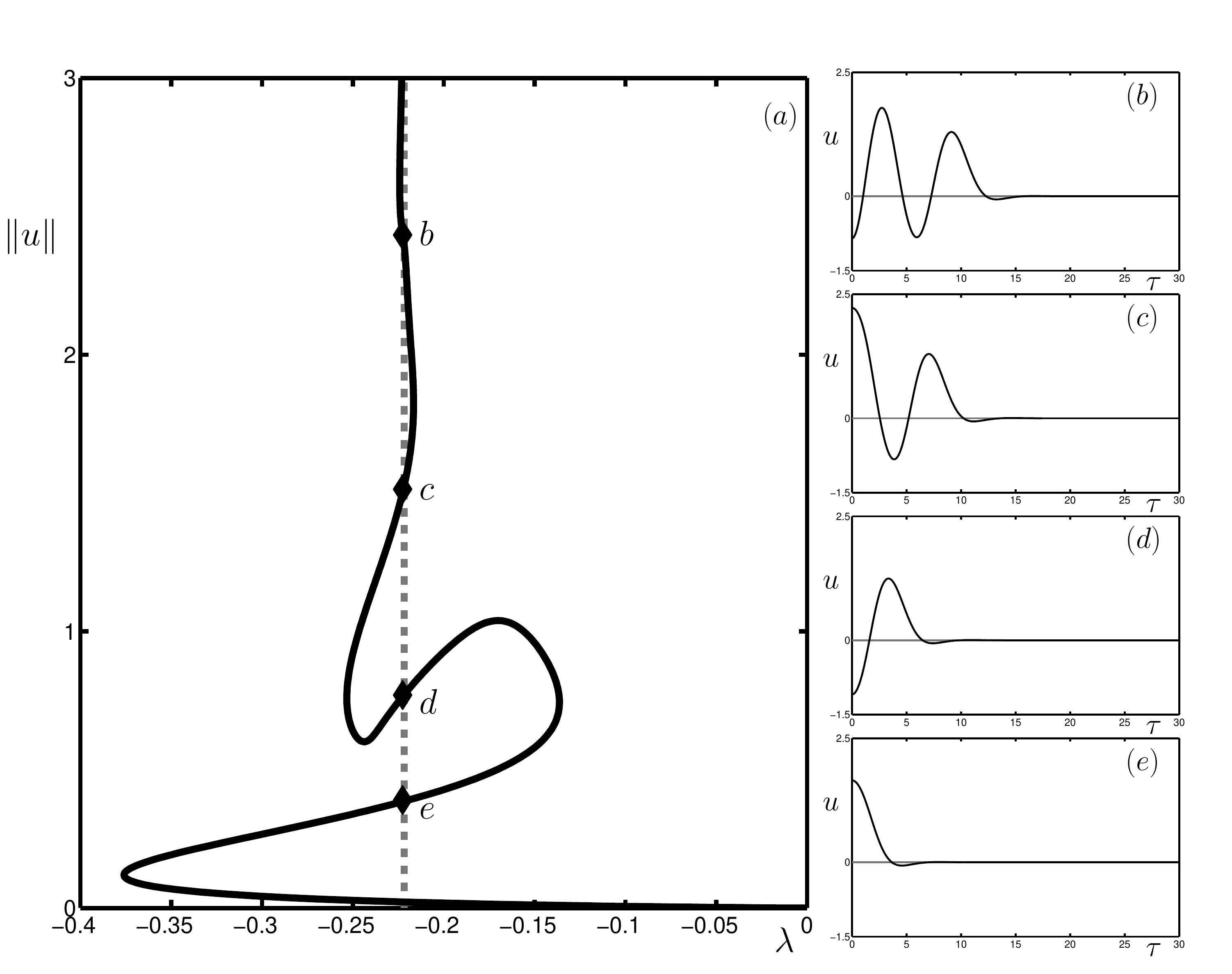}
 \caption{Bifurcation diagram in $\lambda$ of spots of \eqref{eq:sh_stationary} is shown in panel (a) for $\nu=1.6$ and $\eta=1$. Panels (b)-(e) contain plots of spots at $\lambda=-0.2218$ at different values of $\|u\|$.}
 \label{fig:spot16}
\end{figure}

Now, we compute spots for $\nu=1.6$ and $\eta=1$ and show the results in Figure \ref{fig:spot16}.  Panel (a) shows a bifurcation diagram in $\lambda$ where the branches are represented in terms of the Euclidean $\text{L}^2$-norm $\|u\|=\left(\int_0^\infty u(\tau)^2\tau d\tau \right)^{1/2}$. This is due to numerical difficulties of computing the hyperbolic $\text{L}^2$ norm presented by the $\sinh$ function in the integrand but, as indicated in \cite{faye-rankin-etal:12}, the Euclidean radial norm is a good solution measure. The panels (b)-(e) show solution profiles at different points on the bifurcation diagram for fixed value of $\lambda$ as indicated. As one moves up on the branch, rolls are added one by one to the tails of the spots: this corresponds to adding concentric rings that surround the spot. Panels (b)-(e) show that the amplitude at the core is still much larger than the amplitude of the concentric rings that are added. The bifurcation diagram in panel (a) presents similar characteristics of "snaking"-type diagram found in Euclidean geometry with what could be an analog of a Maxwell point at $\lambda_M=-0.2218$ \cite{burke-knobloch:06,burke-knobloch:07c,lloyd-sandstede:09,avitabile-lloyd-etal:10}. Note that this "hyperbolic snaking" structure was not reported in \cite{faye-rankin-etal:12} in the case of neural field equation.


\section{Horocyclic traveling waves}\label{sec5}

The basic material of this section was introduced in \cite{chossat-faugeras:09} in another context (analysis of a model equation for the detection of textures by the visual cortex of mammals). \\
Instead of being periodic on a lattice in $\D$, we may look for states which assume the form of hyperbolic plane waves (or {\em horocyclic waves}) as defined in \ref{subsub:linear}. Let us consider the horocycles which come in contact with $\partial\D$ at the point $b=1$ (see Figure \ref{fig:horocycle}). Horocyclic waves with base point $b$ would be constant along the horocycles and periodic along the "hyperbolic" coordinate. In other words, writing $z=n_s\cdot a_\tau\cdot O$ in horocyclic coordinates, the solution would satisfy the invariance properties
$$u(n_s\cdot a_\tau\cdot O)= u(a_\tau\cdot O) \text{ and } u(a_{\tau+T}\cdot O)=u(a_\tau\cdot O)$$
for some period $T$. Such solutions would be the counterpart of the "stripe" or "roll waves" solutions which occur for the Swift-Hohenberg equation and most equations which allow for pattern selection in the Euclidian plane. We shall see in this section that such solutions do indeed bifurcate in $\D$ and have many common features with the Euclidean case with some important differences. In particular, these waves are {\em traveling} at constant, generically non zero speed in $\D$, while they are steady in the Euclidean plane.

Assuming horocyclic invariance as defined above reduces the coordinates to the single variable $\tau$. The Laplace-Beltrami operator then reduces to (see \cite{helgason:00} where the normalization of the transformations $a_\tau$ is slightly different from ours)
\bqs
\Delta^0_\D = \frac{\partial^2}{\partial\tau^2} - \frac{\partial}{\partial\tau}
\eqs
Now the equation (\ref{eq:sh}) reads
\bqq
\label{eq:sh-horocyclic}
u_t=-(\alpha^2+\Delta^0_\D)^2u+\lambda u+\nu u^2-\eta u^3, \quad z\in\D,
\eqq
and $u$ is a function of $\tau$ and $t$ only.

As was mentionned in \ref{subsub:linear}, the linear stability analysis of the trivial state of (\ref{eq:sh}) with respect to the elementary eigenfunctions $e_{\rho,b}$ of the Laplace-Beltrami operator comes back to solving the "dispersion relation" (\ref{eq:dispersion relation}). These "hyperbolic plane waves" are precisely functions of the coordinate $\tau$ only and they are in exact correspondance with the elementary eigenfunctions of $\Delta^0_\D$, which have the form $e^{i k\tau}$, $k\in\C$. The correspondance is given by the relation $i k=i\rho+1/2$, which translate for the eigenvalues of $\Delta_\D$ to $-k^2-i k=-\rho^2-1/4$. In order for these "waves" to be periodic in $\tau$ we require $k\in\R$. Observe that the eigenvalues are still complex. The dispersion relation (\ref{eq:dispersion relation}) for perturbations $e^{\sigma t +i k\tau}$ now reads
\bqs
\sigma = -k^4+(2\alpha^2+1)k^2-\alpha^4+\lambda + 2k(\alpha^2-k^2)i
\eqs
With no loss of generality we now assume that $\alpha^2=1/2$. Then from the dispersion relation it follows that the most unstable modes occur at $k=\pm1$ with critical parameter value $\lambda=-3/4$, which corresponds to eigenvalues $\sigma=\mp i$. Therefore the bifurcation of periodic hyperbolic plane waves is a Hopf bifurcation to time periodic solutions. 

This is however a simple case of Hopf bifurcation thanks to the translational invariance of Equation (\ref{eq:sh-horocyclic}). Indeed let us look for solutions in the following form: $u(\tau,t)=v(\tau-(1+\omega) t)=v(y)$ with $\omega$ to be determined. This is a uniformly traveling wave at constant speed $1+\omega$. Then (\ref{eq:sh-horocyclic}) becomes
\bqq
\label{eq:sh-travelingwave}
0=-(\frac{1}{2}+\frac{d^2}{dy^2} - \frac{d}{dy})^2v - (1+\omega)\frac{dv}{dy} +\lambda v+\nu v^2-\eta v^3
\eqq
At $\lambda=-3/4$ and $\omega=0$ the linear part of this equation has an eigenvalue at $0$ with eigenvectors $\zeta_0=e^{i y}$ and $\bar\zeta_0$ and solving this bifurcation problem is a classical exercise. \\
We set $\lambda=-3/4+\mu$  and $v(y)=A\zeta_0+\bar A\bar\zeta_0 + \Psi(A,\bar A,\omega,\mu)$ where $\Psi$ is a map on the orthogonal complement of the kernel of the critical linear part spanned by $\zeta_0$,       $\bar\zeta_0$ (we consider here the usual $L^2$ inner product of $2\pi$ periodic functions). The Taylor expansion of $\Psi$ and the bifurcation equation for $A$, $\bar A$ and $\omega$ are then calculated using the Lyapunov-Schmidt decomposition method. 

The following equations are obtained at leading order:
\bqs
0 = (\mu - i\omega)A + \left(2\nu^2\left(\frac{78-4 i}{75}\right)-3\eta\right) A^2\bar A + A\cdot o(\|A\|^2+|\omega|+|\mu|)
\eqs
and complex conjugate. The bifurcated solutions are therefore given at leading order by the amplitudes
\bqs
\|A\|^2 = -\left(\frac{156}{75}\nu^2-3\eta\right)^{-1}\mu
\eqs 
and frequency $1-8/75\nu ^2\|A\|^2$. Observe that the bifurcation is supercritical if $\eta>0$ and $\nu$ is small enough, and that the speed of the traveling wave is lower that the critical one if $\nu\neq 0$ (it depends on higher order terms if $\nu=0$).

\begin{figure}[htp]
\centering
\subfigure[]{
\label{fig:tv}
\includegraphics[width=0.45\textwidth]{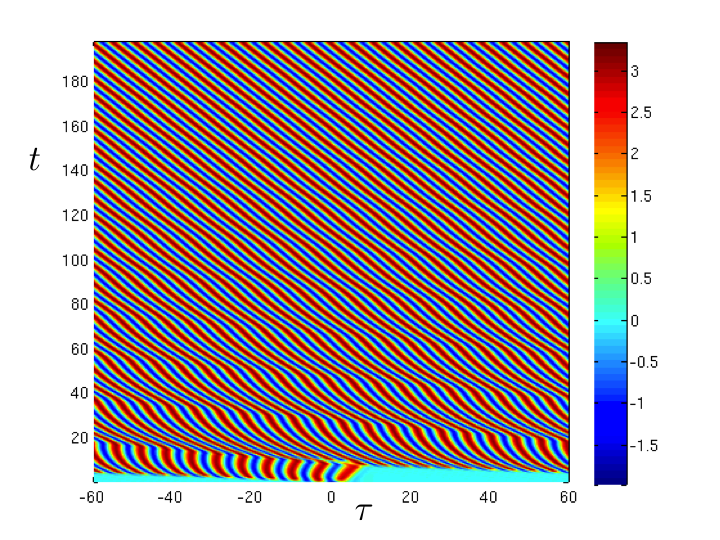}}
\subfigure[]{
\label{fig:htv}
\includegraphics[width=0.45\textwidth]{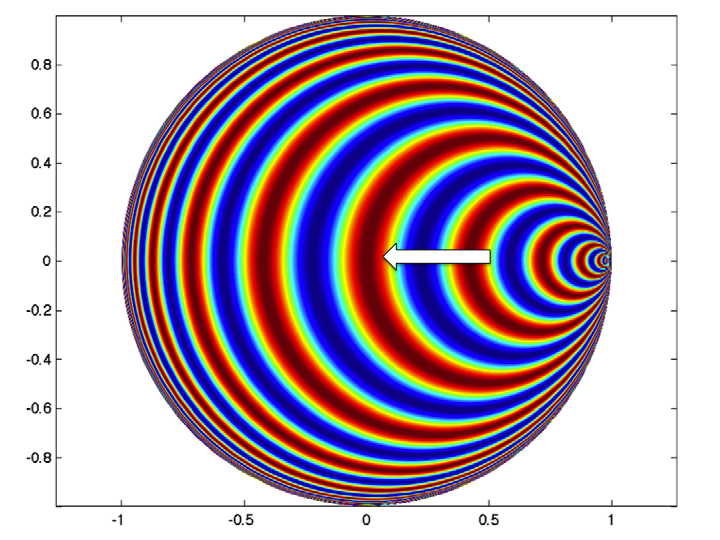}}
\caption{Left: Color plot of a solution of \eqref{eq:sh-horocyclic} in horocyclic coordinate $\tau$. Right: Snapshot of the same solution profile of \eqref{eq:sh-horocyclic} in the Poincar\'e disk at time $t=200$. Values of the parameters are: $\mu=-0.01$, $\nu=2.7$ and $\eta=1$.}
\label{fig:traveling_wave}
\end{figure}

In Figure \ref{fig:traveling_wave}, we show the example of a bifurcating horocyclic traveling wave of equation \eqref{eq:sh-horocyclic}. Numerical simulation of traveling wave is carried out using MATLAB. We take as initial condition a small localized solution around $\tau=0$ and run the simulation for $t\in[0,200]$. We use a semi-implicit finite differences method to compute the solution of \eqref{eq:sh-horocyclic}. Space and time discretizations are taken to be $\Delta t=0.05$ and $\Delta \tau =0.02$. The values of the parameters are taken to be $\mu=-0.01$, $\nu=2.7$ and $\eta=1$. It can be seen from Figure \ref{fig:tv} that the solution converges to a traveling wave solution in horocyclic coordinate. We also plot in Figure \ref{fig:htv} the corresponding solution profile at time $t=200$ in the Poincar\'e disk where the white arrow indicates the direction of propagation of the wave.

\section{Discussion and some open problems}\label{sec6} 

In this review, we have analyzed the bifurcation of i) spatially periodic solutions, ii) radially localized solutions and iii) traveling waves for the Swift-Hohenberg equation with quadratic-cubic nonlinearity defined on the two-dimensional hyperbolic space: $\D$ (Poincar\'e disk). 

For the bifurcation problem of periodic solutions, we have made use of the concept of periodic lattice in $\D$  to further reduce the problem to one on a compact Riemann surface $\D/\Gamma$ , where $\Gamma$ is a co-compact, torsion-free Fuchsian group. Following the method introduced by \cite{chossat-faye-etal:11} in a different context, we have applied techniques from equivariant bifurcation theory in the case of an octagonal periodic pattern, where we have been able to determine the generic bifurcation diagrams of each irreducible representation. As a case study, for the irreducible representation associated to the lowest nonzero eigenvalue, we completely computed the coefficients of the reduced equation on the center manifold and determined in parameter space $(\nu,\mu)$ the stability of each bifurcated branch of solutions.

To prove the existence of a branch of bump solutions bifurcating near onset we have followed the presentation in Faye \etal \cite{faye-rankin-etal:12}, which applied techniques developed by Scheel \cite{scheel:03} and Lloyd \& Sandstede \cite{lloyd-sandstede:09} . In the Poincar\'e disk, the analysis near the core manifold required the development of a detailed knowledge of the asymptotics of the associated Legendre functions, which are the counterpart of the Bessel functions in the Euclidean case. It turns out that the essential difference between the two geometries comes from the far field. At infinity, Bessel function $J_0( r )$ scales in term of the radial coordinate $r$ as $1/\sqrt{r}$ whereas the associated Legendre function $\mathcal{P}_{-\frac{1}{2}+i\frac{\sqrt{3}}{2}}(\cosh(\tau))$ scales in term of the polar geodesic coordinate $\tau$ as $e^{-\tau/2}$ for $\alpha = 1$ in equation \eqref{eq:sh}. Moreover, in the  Euclidean case and for the trivial state at infinity, there is a Turing instability bifurcation. However, in the hyperbolic case, the trivial state is always asymptotically stable at infinity and this simplifies the resulting matching problem. We have also used numerical continuation to track branches of radially symmetric solutions away from onset and found for high value of the parameter $\nu$ what could be the analog of a snaking diagram for the Poincar\'e disk (see Figure \ref{fig:spot16}(a)). 

Finally, we have seen that a branch of traveling waves invariant along horocycles and periodic in the transverse direction, bifurcates off the trivial state through of Hopf bifurcation. This striking result highlights once again, the inherent differences that exist between pattern formation in Euclidean and hyperbolic geometry.

This analysis is far from being complete and there still exist many open questions. We list below some problems that we consider to be important.

\begin{itemize}
\item We studied the bifurcation of H-planforms for the regular octagonal lattice. There are infinitely many types of periodic lattices in the Poincar\'e disc.  Of course, the methods presented in Section \ref{sec3} would easily be adapted to other lattices. However this raises the question of the observability of such patterns in a natural system or under direct simulation of the Swift-Hohenberg equation. Indeed, in addition to the high degeneracy of the bifurcation problem, there is a rigidity of the lattices in hyperbolic space which, as explained in Section \ref{sec3}, could make the observation of such patterns unlikely. We may imagine mechanisms such as spatial frequency locking which could overcome this difficulty, but this is a very challenging problem. 
\item In section \ref{sec4}, we presented an existence result of localized radially symmetric solution. Other localized radial states can be expected and the theory developed in \cite{lloyd-sandstede:09,mccalla:11,mccalla-sandstede:12} could be extended to yield the existence of rings and other types of localized radially symmetric solution, which were not investigated in this paper. We also think that the existence of horocyclic traveling waves might lead to the formation of target or source emitters. 
\item In the Euclidean case, localized hexagonal patches have been studied in the Swift-Hohenberg equation with quadratic-cubic nonlinearity \cite{lloyd-sandstede-etal:08}. It would be of interest to numerically investigate if one can find localized patches with $\mathbf{D}_8$-symmetry for the range of parameters $(\nu,\mu)$ in Figure \ref{fig:stab_oct} where the branch of solution with $\mathbf{D}_8$-symmetry is stable. Note that Faye \etal \cite{faye-rankin-etal:12} have numerically computed branches of solutions with $\mathbf{D}_8$-symmetry bifurcating off the branch of radially localized solutions for a neural field equation set on the Poincar\'e disk.
\end{itemize}

{\noindent \bf Acknowledgment:} GF is grateful to James Rankin and David Lloyd for their helpful comments on the functionality of AUTO.

\end{document}